\documentclass[11pt,reqno]{amsart}
\usepackage[utf8]{inputenc}
\usepackage{geometry}
\usepackage{xcolor}
\usepackage{mathrsfs}
\usepackage{enumitem}
\usepackage{amstext}
\usepackage{amsthm}
\usepackage{amssymb}
\usepackage{graphicx}
\PassOptionsToPackage{normalem}{ulem}
\usepackage{ulem}
\usepackage[unicode=true,
 bookmarks=false,
 breaklinks=false,pdfborder={0 0 1},backref=false,colorlinks=false]
 {hyperref}

\makeatletter

\providecommand{\tabularnewline}{\\}
\providecolor{lyxadded}{rgb}{0.047,0.78,0.44}
\providecolor{lyxdeleted}{rgb}{1,0,0}
\DeclareRobustCommand{\mklyxadded}[1]{\textcolor{lyxadded}\bgroup#1\egroup}
\DeclareRobustCommand{\mklyxdeleted}[1]{\textcolor{lyxdeleted}\bgroup\mklyxsout{#1}\egroup}
\DeclareRobustCommand{\mklyxsout}[1]{\ifx\\#1\else\sout{#1}\fi}
\numberwithin{equation}{section}
\numberwithin{figure}{section}
\theoremstyle{plain}
\newtheorem{thm}{\protect\theoremname}[section]
\theoremstyle{remark}
\newtheorem{rem}[thm]{\protect\remarkname}
\theoremstyle{plain}
\newtheorem{lem}[thm]{\protect\lemmaname}
\theoremstyle{remark}
\newtheorem*{rem*}{\protect\remarkname}
\theoremstyle{definition}
\newtheorem{example}[thm]{\protect\examplename}
\theoremstyle{plain}
\newtheorem{prop}[thm]{\protect\propositionname}
\theoremstyle{plain}
\newtheorem{cor}[thm]{\protect\corollaryname}
\theoremstyle{definition}
\newtheorem{defn}[thm]{\protect\definitionname}

\usepackage{amscd}
\usepackage{latexsym}

\usepackage{amscd}
\usepackage[mathscr]{euscript}
\usepackage{mathrsfs}
\usepackage{xcolor}
\usepackage{color}

\usepackage{lscape}%to rotate pages
\usepackage{epstopdf}

\usepackage{array}
\usepackage{tabularx}
\usepackage{longtable}
\usepackage{ltxtable}

\usepackage{color}

\usepackage{tikz}
\usepackage{tikz-cd}
\usetikzlibrary{positioning}
\usetikzlibrary{decorations,arrows}
\usetikzlibrary{decorations.pathmorphing}
\usepgflibrary{decorations.pathreplacing} 
\usepackage{pgf}
\usetikzlibrary{patterns}

\tikzset{
  schraffiert/.style={pattern=horizontal lines,pattern color=#1},
  schraffiert/.default=black
}

\tikzset{
    ultra thin/.style= {line width=0.1pt},
    very thin/.style=  {line width=0.2pt},
    thin/.style=       {line width=0.4pt},% thin is the default
    semithick/.style=  {line width=0.6pt},
    thick/.style=      {line width=0.8pt},
    very thick/.style= {line width=1.2pt},
    ultra thick/.style={line width=2.4pt}
}

 \usetikzlibrary{shapes}

\definecolor{hellgrau}{rgb}{0.93,0.93,0.93}
\definecolor{hellergrau}{rgb}{0.97,0.97,0.97}
\definecolor{hellgruen}{rgb}{0.6,1.35,0.5}
\definecolor{grau}{rgb}{0.93,0.93,0.93}
\definecolor{hellblau}{rgb}{0.8,0.8,2.0}
\definecolor{blau}{rgb}{0.3,0.5,2.0}
\definecolor{hellrot}{rgb}{2.0,0.6,0.6}
\definecolor{gruen}{rgb}{0.3,0.75,0.2}
\definecolor{rot}{rgb}{0.9,0.1,0.1}
\definecolor{orang}{rgb}{1.3,0.65,0}

\allowdisplaybreaks

\DeclareMathAlphabet{\mathpzc}{OT1}{pzc}{m}{it}

\newcommand{\Aa}{{\mathcal A}}
\newcommand{\Cc}{{\mathcal C}}
\newcommand{\Ff}{{\mathcal F}}
\newcommand{\Hh}{{\mathcal H}}
\newcommand{\Ll}{{\mathcal L}}

\newcommand{\Pp}{{\mathcal P}}

\newcommand{\supp}{{\operatorname{supp}}}			%support
\newcommand{\precd}{\prec\!\!\prec}					%double prec (intervals)
\newcommand{\strict}{\subseteq_{\text{strict}}} 	%strict inclusion
\newcommand{\sign}{\mbox{sign}}						%sign
\newcommand{\ind}{\mbox{ind}}						%index
\newcommand{\indB}{\mbox{ind}_B}						%B-index

\newcommand{\motw}{=_2}

\newcommand{\rbr}[1]{\left(#1\right)}
\newcommand{\sbr}[1]{\left[#1\right]}

\newcommand{\abs}[1]{\left|#1\right|}

\allowdisplaybreaks

\usepackage{xcolor}
\newcommand{\Hmm}[1]{\leavevmode{\marginpar{\tiny%
$\hbox to 0mm{\hspace*{-0.5mm}$\leftarrow$\hss}%
\vcenter{\vrule depth 0.1mm height 0.1mm width \the\marginparwidth}%
\hbox to 0mm{\hss$\rightarrow$\hspace*{-0.5mm}}$\\\relax\raggedright #1}}}

\definecolor{green}{RGB}{0, 180, 0}
\definecolor{cyan}{RGB}{0, 180, 180}
\definecolor{yellow}{RGB}{211,211,0}

\renewcommand{\subset}{\subseteq}
\makeatother

\providecommand{\corollaryname}{Corollary}
\providecommand{\definitionname}{Definition}
\providecommand{\examplename}{Example}
\providecommand{\lemmaname}{Lemma}
\providecommand{\propositionname}{Proposition}
\providecommand{\remarkname}{Remark}
\providecommand{\theoremname}{Theorem}

\begin{document}
\global\long\def\theenumi{\alph{enumi}}%

\global\long\def\theenumii{\arabic{enumii}}%

%%%%%%%%% Additional symbols from Yannik %%%%%%%%%%%%%%%%%%%%%%%%%%

%%%%%%%%% canonical symbols %%%%%%%%%%%%%%%%%%%%%%%%%%

\global\long\def\ui{\mathbf{\textrm{i}}}%
\global\long\def\ue{\mathbf{\textrm{e}}}%
\global\long\def\ud{\mathbf{\textrm{d}}}%
\global\long\def\sgn{\mathrm{sign}}%
\global\long\def\id{\mathbf{1}}%
\global\long\def\C{\mathbb{C}}%
\global\long\def\R{\mathbb{R}}%
\global\long\def\Q{\mathbb{Q}}%
\global\long\def\N{\mathbb{N}}%
\global\long\def\T{\mathbb{T}}%
\global\long\def\Z{\mathbb{Z}}%
\global\long\def\Id{\mathbf{\mathbb{I}}}%

%%%%%%%%% Caligraphic symbols %%%%%%%%%%%%%%%%%%%%%%%%%%

\global\long\def\Aa{\mathbf{\mathcal{A}}}%

\global\long\def\Cc{\mathbf{\mathcal{C}}}%

\global\long\def\Ee{\mathbf{\mathcal{E}}}%

\global\long\def\Ff{\mathcal{F}}%

\global\long\def\Hh{\mathcal{H}}%

\global\long\def\Ll{\mathcal{L}}%

\global\long\def\Pp{\mathcal{P}}%

%%%%%%%%% Paper specific symbols %%%%%%%%%%%%%%%%%%%%%%%%%%

\global\long\def\co{\textbf{c}}%

\global\long\def\cop{\textbf{c'}}%

\global\long\def\cok{\textbf{c}_{k}}%

\global\long\def\cm{[\textbf{c},m]}%

\global\long\def\cmn{[\textbf{c},m,n]}%

\global\long\def\Co{\mathscr{C}}%

\global\long\def\emp{\emptyset}%
\global\long\def\tr{\textit{tr}}%
\global\long\def\Nz{\N_{0}}%
\global\long\def\Nmo{\N_{-1}}%
\global\long\def\L{L}%

%%%%%%%%% More Paper specific symbols - from Yannik %%%%%%%%%%%%%%%%%%%%%%%%%%

\global\long\def\ext{\varrho}%

\global\long\def\ypi{\mu}%

\global\long\def\yPi{\mathcal{M}}%

\global\long\def\res{R}%

\global\long\def\Des{D}%

\global\long\def\paths{\Gamma}%

%%%%%%%%%%%%%%%%%%%%%%%%%%%%%%%%%%%%%%%%%%%%%%%%%%%%%%
%%%%%%%%% special symbols %%%%%%%%%%%%%%%%%%%%%%%%%%%%
%%%%%%%%%%%%%%%%%%%%%%%%%%%%%%%%%%%%%%%%%%%%%%%%%%%%%%

\global\long\def\precd{\prec_{\mathrm{str}}}%
 %double prec (intervals)

\global\long\def\strict{\subseteq_{\text{str}}}%

\global\long\def\Istr{I_{\mathrm{str}}}%

%strict inclusion
\global\long\def\ind{\mbox{ind}}%
%index 
\global\long\def\indA{\mbox{ind}_{A}}%
%A-index 
\global\long\def\indB{\mbox{ind}_{B}}%
%B-index 
\global\long\def\indC{\mbox{ind}_{C}}%
%C-index 
\global\long\def\Im{\mathrm{Im}}%
\global\long\def\dR{\delta_{R}}%
%Whether an eigenvalue is right edge of spectral band
\global\long\def\motw{=_{2}}%
\global\long\def\tA{\mathit{A}}%
\global\long\def\tB{\mathit{B}}%
\global\long\def\tC{\mathit{G}}%

%%%%%%%%% Spectrum %%%%%%%%%%%%%%%%%%%%%%%%%%%%

\global\long\def\emap{E_{\alpha,V}}%

\global\long\def\spec{\mathrm{spec}}%
\global\long\def\sigc{\sigma_{\co}}%
\global\long\def\sigcm{\sigma_{[\co,m]}}%
\global\long\def\sigcmn{\sigma_{[\co,m,n]}}%
\global\long\def\sigcz{\sigma_{[\co,0]}}%

\global\long\def\bc{b_{\co,V}}%
\global\long\def\bcop{b_{\cop,V}}%
\global\long\def\bcm{b_{[\co,m],V}}%
\global\long\def\bcmn{b_{[\co,m,n],V}}%

\global\long\def\Sc{\Sigma_{\co}}%
\global\long\def\Scs{\Sigma_{\co}^{\textrm{spec}}}%
\global\long\def\Scop{\Sigma_{\cop}}%
\global\long\def\Scops{\Sigma_{\cop}^{\textrm{spec}}}%
\global\long\def\Scm{\Sigma_{[\co,m]}}%
\global\long\def\Scms{\Sigma_{[\co,m]}^{\textrm{spec}}}%
\global\long\def\Scmn{\Sigma_{[\co,m,n]}}%
\global\long\def\Scmns{\Sigma_{[\co,m,n]}^{\textrm{spec}}}%
\global\long\def\Sck{\Sigma_{\co_{k}}}%
\global\long\def\Scks{\Sigma_{\cok}^{\textrm{spec}}}%

\global\long\def\Lc{L_{\co,V}}%
\global\long\def\Lcop{L_{\cop,V}}%
\global\long\def\Lcm{L_{[\co,m],V}}%
\global\long\def\Lcmn{L_{[\co,m,n],V}}%
\global\long\def\Lck{L_{\co_{k},V}}%

%%%%%%%%% Lambdas %%%%%%%%%%%%%%%%%%%%%%%%%%%%

\global\long\def\lc{\lambda_{\co}}%
\global\long\def\lcm{\lambda_{[\co,m]}}%
\global\long\def\lcmn{\lambda_{[\co,m,n]}}%
\global\long\def\lo{\lambda_{\mathbf{o}}}%

%%%%%%%%% Omegas %%%%%%%%%%%%%%%%%%%%%%%%%%%%

\global\long\def\ohn{\omega_{\frac{p_{k}}{q_{k}}}}%
\global\long\def\ohnm{\omega_{\frac{p_{k-1}}{q_{k-1}}}}%
\global\long\def\ohnmm{\omega_{\frac{p_{k-2}}{q_{k-2}}}}%
\global\long\def\ohnp{\omega_{\frac{p_{k+1}}{q_{k+1}}}}%
\global\long\def\ohrat{\omega_{\pq}}%

%%%%%%%%% Hamiltonians %%%%%%%%%%%%%%%%%%%%%%%%%%%%

\global\long\def\Ham{H_{\alpha,V}}%
\global\long\def\Hrat{H_{\frac{p}{q},V}}%
\global\long\def\Hc{H_{\co,V}}%
\global\long\def\Hcm{H_{[\co,m]}}%
\global\long\def\Hcmn{H_{[\co,m,n]}}%
\global\long\def\Hcmo{H_{[\co,m,1]}}%
\global\long\def\Hcz{H_{[\co,0]}}%
\global\long\def\Hco{H_{[\co,1]}}%

%%%%%%%%% Characteristic Polynomial %%%%%%%%%%%%%%%%%%%%%%%%%%%%

\global\long\def\polyc{P_{\co,V}}%
\global\long\def\Pc{P_{\co}}%
% Characteristic polynomials 

%%%%%%%%% Thetas %%%%%%%%%%%%%%%%%%%%%%%%%%%%

\global\long\def\thc{\theta_{\co}}%
\global\long\def\thcm{\theta_{[\co,m]}}%
\global\long\def\thcmn{\theta_{[\co,m,n]}}%

%%%%%%%%% Fractions  %%%%%%%%%%%%%%%%%%%%%%%%%%%%

\global\long\def\alk{\alpha_{k}}%

\global\long\def\pq{\frac{p}{q}}%
\global\long\def\pqk{\frac{p_{k}}{q_{k}}}%
\global\long\def\pc{p_{\co}}%
\global\long\def\pcm{p_{[\co,m]}}%
\global\long\def\pcmn{p_{[\co,m,n]}}%
\global\long\def\qc{q_{\co}}%
\global\long\def\qcm{q_{[\co,m]}}%
\global\long\def\qcmn{q_{[\co,m,n]}}%

%%%%%%%%% Traces %%%%%%%%%%%%%%%%%%%%%%%%%%%%
\global\long\def\tc{t_{\co}}%
\global\long\def\tcm{t_{[\co,m]}}%
\global\long\def\tcmn{t_{[\co,m,n]}}%
\global\long\def\tco{t_{[\co,1]}}%
\global\long\def\tcmo{t_{[\co,-1]}}%
\global\long\def\tcz{t_{[\co,0]}}%

%%%%%%%%% Transfer Matrices %%%%%%%%%%%%%%%%%%%%%%%%%%%%

\global\long\def\Mc{M_{\co}}%
\global\long\def\Mz{M_{[0]}}%
\global\long\def\Mzz{M_{[0,0]}}%

%%%%%%%%% Polynomials %%%%%%%%%%%%%%%%%%%%%%%%%%%%

\global\long\def\CP{\mathrm{S}}%
 %Chebychev Polynomials 

%%%%%%%%% Fourier Transform %%%%%%%%%%%%%%%%%%%%%%%%%%%%

\global\long\def\f{f}%
\global\long\def\Span{\mathrm{Span}}%

%%%%%%%%% Other stuff %%%%%%%%%%%%%%%%%%%%%%%%%%%%

\global\long\def\db{\delta_{B}}%

\global\long\def\Nc{N_{\co}}%
\global\long\def\Nc{N_{\co}}%
\global\long\def\Ncm{N_{[\co,m]}}%
\global\long\def\Ncmn{N_{[\co,m,n]}}%
\global\long\def\Ic{I_{\co}}%
\global\long\def\Icop{I_{\cop}}%
\global\long\def\Ico{I_{[\co,1]}^{1}}%
\global\long\def\Icz{I_{[\co,0]}}%

\global\long\def\Jcz{J_{[\co,0]}}%
\global\long\def\Kcz{K_{[\co,0]}}%
\global\long\def\cz{[\co,0]}%
\global\long\def\Sturm{\omega_{\alpha}}%
\global\long\def\IDS{N_{\alpha,V}}%
\global\long\def\Jcm{J_{[\co,m]}}%
\global\long\def\Kcm{K_{[\co,m]}}%
\global\long\def\Icm{I_{[\co,m]}}%
\global\long\def\Icmn{I_{[\co,m,n]}}%
\global\long\def\Icmo{I_{[\co,m,1]}}%
\global\long\def\Icmoi{I_{[\co,m,1]}^{i}}%
\global\long\def\Icmi{I_{[\co,m]}^{i}}%
\global\long\def\Icmj{I_{[\co,m]}^{j}}%
\global\long\def\Icmni{I_{[\co,m,n]}^{i}}%
\global\long\def\Icmnj{I_{[\co,m,n]}^{j}}%
\global\long\def\Icmii{I_{[\co,m]}^{i+1}}%
\global\long\def\Icmnii{I_{[\co,m,n]}^{i+1}}%

\global\long\def\fl#1{\emph{\ensuremath{\left\lfloor #1\right\rfloor }}}%
\global\long\def\set#1#2{\left\{  #1~:~#2\right\}  }%
\global\long\def\newmacroname{\{\}}%

\global\long\def\vect#1{\ensuremath{\begin{pmatrix}#1\end{pmatrix}}}%

\global\long\def\Mat#1{\ensuremath{{\begin{pmatrix}#1\end{pmatrix}}}}%

\global\long\def\sign{\mbox{sign}}%

\global\long\def\supp{\ensuremath{{\operatorname{supp}}}}%

\global\long\def\la{\langle}%
 
\global\long\def\ra{\rangle}%

\global\long\def\rbr#1{\ensuremath{\left(#1\right)}}%

\global\long\def\sbr#1{\ensuremath{\left[#1\right]}}%

\global\long\def\abs#1{\ensuremath{\left|#1\right|}}%

\global\long\def\gst{\;\middle|\;}%
 
\global\long\def\oalpha{\overline{\alpha}}%
 
\global\long\def\obeta{\overline{\beta}}%

\global\long\def\traceOld#1{t_{\rbr{#1}}}%
 
\global\long\def\spectOld#1{\sigma_{\rbr{#1}}}%

\global\long\def\mat#1{M_{#1}}%
 
\global\long\def\matv#1{M_{\sbr{#1}}}%
 
\global\long\def\trace#1{t_{#1}}%
 
\global\long\def\tracev#1{t_{\sbr{#1}}}%
 
\global\long\def\spect#1{\sigma_{#1}}%
 
\global\long\def\spectv#1{\sigma_{\sbr{#1}}}%
 
\global\long\def\op#1{H_{#1}}%
 
\global\long\def\opv#1{H_{\sbr{#1}}}%

\global\long\def\typeA{\mathcal{A}}%
 
\global\long\def\typeB{\mathcal{B}}%

\global\long\def\pprec{\prec\!\!\!\prec}%

\let\oldsubset\subset 
\global\long\def\subset{\subseteq}%
 
\global\long\def\subsets{\oldsubset_{s}}%

\global\long\def\w{\gamma}%
 
\global\long\def\wo{\gamma}%

\title[Sturmian Hamiltonians with a large coupling constant]{a review of a work by Raymond: Sturmian Hamiltonians with a large
coupling constant - periodic approximations and gap labels}
\author{Ram Band, Siegfried Beckus, Barak Biber, Laurent Raymond, Yannik Thomas}
\address{Department of Mathematics\\
Technion - Israel Institute of Technology\\
Haifa, Israel\\
and\\
Institute of Mathematics\\
University of Potsdam\\
Potsdam, Germany}
\email{ramband@technion.ac.il}
\address{Institute of Mathematics\\
University of Potsdam\\
Potsdam, Germany}
\email{beckus@uni-potsdam.de}
\address{Department of Mathematics and the Henry and Marilyn Taub Faculty of
Computer Science\\
Technion - Israel Institute of Technology\\
Haifa, Israel}
\email{biber.barak@campus.technion.ac.il}
\address{Aix Marseille Univ\\
Université de Toulon\\
CNRS, CPT, Marseille, France}
\email{laurent.raymond@univ-amu.fr}
\address{Institute of Mathematics\\
University of Potsdam\\
Potsdam, Germany}
\email{yannik.thomas@uni-potsdam.de}
\begin{abstract}
We present a review of the work \cite{Raym95,Raym95-thesis}. The
review aims at making this work more accessible and offers adaptations
of some statements and proofs. In addition, this review forms an applicable
framework for the complete solution of the Dry Ten Martini Problem
for Sturmian Hamiltonians as appears in \cite{BanBecLoe_24}.

A Sturmian Hamiltonian is a one-dimensional Schr\"odinger operator
whose potential is a Sturmian sequence multiplied by a coupling constant,
$V\in\R$. The spectrum of such an operator is commonly approximated
by the spectra of designated periodic operators. If $V>4$, then the
spectral bands of the periodic operators exhibit a particular combinatorial
structure. This structure provides a formula for the integrated density
of states. Employing this, it is shown that if $V>4$, then all the
gaps, as predicted by the gap labelling theorem, are there.
\end{abstract}

\maketitle
\tableofcontents{}

\section{Introduction}

\subsection{The motivation for this review}

The starting point of this paper is the unpublished work of Raymond,
\cite{Raym95} and his PhD thesis \cite{Raym95-thesis} (see also
\cite{MFO2011_Raym}). The first two authors became aware of \cite{Raym95}
via a private communication with Damanik. Band and Beckus were influenced
by \cite{Raym95} in their joint work with Loewy, \cite{BanBecLoe_24}
and found it beneficial to refer to parts of \cite{Raym95} in \cite{BanBecLoe_24}.
Indeed, \cite{Raym95} is a very stimulating work, contains some foundational
results and is referred to numerous times (see e.g., the surveys \cite{Dam07_survey,Jit07,DaEmGo15-survey,Dam17-Survey}
and references within), in spite of being unpublished. We started
to write the current review with three goals in mind. First, it might
be worthwhile to elaborate on some of the proofs and fill in some
gaps. Second, by adapting some notations and conventions, we create
a unified framework towards providing the complete solution for the
Dry Ten Martini Problem for Sturmian Hamiltonians, \cite{BaBeLo23-MFO,BanBecLoe_24}.
Finally, we felt that the whole community might benefit from having
a published version of Raymond's work upon reaching its thirtieth
anniversary. Hence, we joined forces to produce the current review,
with Raymond joining as well after this review was already initiated.
While this review was in final stages of preparation, we became aware
that a similar publication is planned in \cite{Raym-AperiodicOrder},
as part of the book series initiated by Baake and Grimm \cite{BaaGri_book1,BaaGri_book2}.

In this review we make the connection to \cite{Raym95} as transparent
as possible. In particular, throughout the review we clarify as much
as possible where we merely rephrase statements from \cite{Raym95}
and where we elaborate or bring new statements and terminology. When
writing the current review, we were trying to provide an appropriate
balance between two objectives. On the one hand, our desire is to
reflect the original work \cite{Raym95} with no substantial changes.
On the other hand, at times we felt that the exposition may profit
by including adaptations based on later papers and recent progress
in the field.

We should emphasize that the current review covers only the first
five sections of \cite{Raym95} that form the starting point for resolving
the Dry Ten Martini Problem for Sturmian Hamiltonians in \cite{BanBecLoe_24}.
We do not treat here the last section of \cite{Raym95} about the
Hausdorff dimension of the Fibonacci Hamiltonian. This part in \cite{Raym95}
led to further progress in the study of the fractal dimensions of
the spectrum of Sturmian Hamiltonians, see e.g. \cite{KiKiLa03,LiuWen04,Damanik_Tcheremchantsev07,Damanik2008,DaGoYe16}.
Reviewing this part of \cite{Raym95} is not included here since the
focus is on the study of the integrated of states and the gap labels.

\subsection{A short historical review}

Let us start by introducing the model. We consider bounded linear
operators $\Ham:\ell^{2}(\Z)\to\ell^{2}(\Z)$, given by 
\begin{align}
(\Ham\psi)(n) & :=\psi(n+1)+\psi(n-1)+V\chi_{\left[1-\alpha,1\right)}(n\alpha\mod 1)\thinspace\psi(n),\label{eq:Hamiltonian=000020defined}
\end{align}
where $V\in\R$ is the \emph{coupling constant} and $\chi_{\left[1-\alpha,1\right)}$
is the characteristic function of the interval $\left[1-\alpha,1\right)$.
Whenever $\alpha\in\R\setminus\Q$, the operator $\Ham$ is aperiodic
(in the sense that its potential sequence is not periodic) and it
is known as a\textit{ Sturmian Hamiltonian}.

We provide a short summary on the developments for the spectral theory
of Sturmian Hamiltonians and refer the reader to the surveys \cite{Dam07_survey,Jit07,DaEmGo15-survey,Dam17-Survey}
and references therein for more details. This class of operators serves
as the guiding example for one-dimensional quasicrystals and was introduced
in \cite{KohKadTan_prl83,OstKim_phys85}. This model is also called
\emph{Kohmoto model} and a plot of the associated spectra, as they
vary with $\alpha$ - called the Kohmoto butterfly - can be found
in Figure~\ref{fig:Kohmoto=000020Butterfly}.

\begin{figure}[hbt]
\includegraphics[scale=1.4]{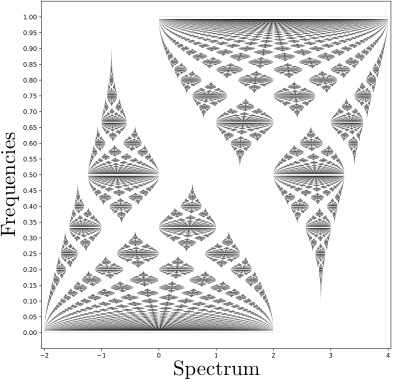}\caption{The Kohmoto buttefly for $V=2$.}
\label{fig:Kohmoto=000020Butterfly}
\end{figure}
A first mathematical analysis of the so called Fibonacci Hamiltonian
$\Ham$ with $\alpha=\frac{\sqrt{5}-1}{2}$ was developed in \cite{Casdagli1986}.
Shortly after it was proven that the Fibonacci Hamiltonian has Cantor
spectrum of Lebesgue measure zero and the spectral measure is purely
singular continuous, \cite{Sut87,Sut89}. For all Sturmian Hamiltonians
(i.e., all $\alpha\notin\Q$ and $V\neq0$), Cantor spectrum of Lebesgue
measure zero was proven in \cite{BIST89}. This result was generalized
in \cite{Len02} to a large class of one-dimensional dynamical systems.
The absence of point spectrum and upper bounds on the growth of solutions
for Sturmian Hamiltonians were thoroughly studied as well \cite{DamLen99-AbsEig,DamLen99_II,DamKilLenz00}.

Influenced by these results, one may ask whether all the spectral
gaps that are predicted by the gap labeling theorem \cite{Bell92-Gap,BelBovGhe92},
appear. This is the so-called Dry Ten Martini Problem for Sturmian
Hamiltonians. Such a question was originally asked by Kac in 1981
for the almost Mathieu operator (``are all gaps there?''), see \cite{Sim82-review}.
For large enough coupling constant, $V>4$, it was proven in \cite{Raym95}
that all gaps are there, and this is reviewed in the current paper.
For the Fibonacci Hamiltonian and small enough coupling $V$, it was
proven in \cite{DamanikGorodetski2011} that all spectral gaps are
there. This result was extended in \cite{Mei14} for $\alpha\in[0,1]\setminus\Q$
with eventually periodic continued fraction expansion and small enough
coupling constant. In a remarkable study of the Fibonacci Hamiltonian
\cite{DaGoYe16}, it was proven that all gaps are there for all $V\neq0$
and $\alpha=\frac{\sqrt{5}-1}{2}$. Finally, a complete solution of
the Dry Ten Martini Problem for Sturmian Hamiltonians for all $\alpha\in[0,1]\setminus\Q$
and all $V\neq0$ is provided in \cite{BanBecLoe_24}. Moreover, the
hierarchical structure of the periodic approximations spectra (initiated
in \cite{Raym95}) was extended in \cite{BanBecLoe_24} to all $V\neq0$.

This hierarchical structure also laid the ground to estimate the Hausdorff
dimension for the Fibonacci Hamiltonian in \cite{Raym95}. It influenced
the study of the fractal dimension and the transport exponent for
Sturmian Hamiltonians during the last decades, see e.g. \cite{KiKiLa03,LiuWen04,Damanik_Tcheremchantsev07,Damanik2008,Liu2014,Damanik2015,DaGoYe16,CaoQu_arXiv23}.

\subsection*{Organization of the paper}

The paper is structured as follows. Section~\ref{sec:=000020The=000020Sturmian=000020Potential}
discusses the Sturmian sequences and their periodic mechanical words.
In addition, we introduce there a designated space of finite continued
fraction expansions following the lines of \cite{BanBecLoe_24}. In
Section~\ref{sec:=000020Transfer=000020Matrices=000020and=000020Discriminant}
we present the standard Floquet-Bloch theory via transfer matrices
and the discriminant. Various useful identities of the discriminants
are presented there. Section~\ref{sec:=000020spectra=000020of=000020periodic=000020approximants}
describes the spectra of the periodic approximations and their special
combinatorial structure - first in general and then specializing for
the case $V>4$. Section~\ref{sec:=000020The=000020IDS} applies
the aforementioned combinatorial structure for the study of the integrated
density of states and the gap labelling for $V>4$.

\subsection*{Acknowledgments}

We are grateful for David Damanik and Michael Baake for connecting
some of the authors. First, in 2018, David Damanik introduced RB and
SB to the original work of LR, encouraging to further explore it,
and suggesting useful references along the way. Then, on December
2023, Michael Baake made the physical connection and kindly hosted
four of the authors in Bielefeld. We thank our colleague Raphael Loewy
who provided us with a critical and constructive viewpoint on this
work.

We thank Israel Institute of Technology and the University of Potsdam
for providing excellent working conditions during our mutual visits.
This work was partially supported by the Deutsche Forschungsgemeinschaft
{[}BE 6789/1-1 to S.B.{]} and the Maria-Weber Grant 2022 offered by
the Hans Böckler Stiftung. RB was supported by the Israel Science
Foundation (ISF Grant No. 844/19).

\section{The Sturmian Potential \protect\label{sec:=000020The=000020Sturmian=000020Potential}}

This section is dedicated to studying the Sturmian sequence $\chi_{\left[1-\alpha,1\right)}(n\alpha\mod 1)$,
which serves as the potential of the Sturmian Hamiltonian, (\ref{eq:Hamiltonian=000020defined}).
In particular, we will consider rational values of $\alpha$, which
give rise to periodic sequences and periodic Hamiltonians. Most of
the content of this section does not appear in \cite{Raym95} and
our main motivation is to use already in the current review some tools
and notations which are essential for \cite{BanBecLoe_24}.

\subsection{The space $\protect\Co$ of finite continued fraction expansions\protect\label{subsec:=000020space=000020of=000020continued=000020fractions}}

Let $\Nz:=\N\cup\{0\}$ and $\Nmo:=\N\cup\{-1,0\}$ and define the
space of \textit{finite continued fractions} by 
\[
\Co:=\left\{ [0],~[0,0]\right\} \cup\bigcup_{k\in\N}\set{[0,0,c_{1},\ldots,c_{k}]}{c_{1},\ldots,c_{k-1}\in\N,~c_{k}\in\N_{-1}}.
\]

This notation uses the convention that the two first entries of all
$\co\in\Co$, satisfy $c_{-1}=c_{0}=0$. As described below $\Co$
extends the common set of continued fractions, as we allow also $c_{k}\in\{-1,0\}$,
for the last entry in the continued fraction expansion.

For $\co=[0,c_{0},c_{1},\ldots,c_{k}]\in\Co$ and $m\in\Nmo$, we
will use the notation 
\[
[\co,m]:=[0,c_{0},c_{1},\ldots,c_{k},m]\in\Co,
\]
whenever it is defined. We use frequently in this work the condition
$[\co,m]\in\Co$ implicitly implying that $[0,0,c_{1},\ldots,c_{k},m]$
is actually an element of $\Co$. This puts some constraints on $c_{k}$.
For instance, if $\co=[0]$, then $m=0$ or if $k\in\N$, then $c_{k}\geq1$
must hold.

We connect the set of continued fractions with rational numbers by
introducing a map $\varphi:\Co\to\R\cup\{\infty\}$. It is defined
for all $\co\in\Co\backslash\left\{ [0],[0,0,-1]\right\} $ by 
\begin{equation}
\varphi([0,c_{0},c_{1},\dots,c_{k}]):=\begin{cases}
\varphi([0,c_{0},c_{1}\dots,c_{k-2},c_{k-1}-1]), & k\in\N\textrm{ and }c_{k}=-1,\\[0.1cm]
\varphi([0,c_{0},c_{1}\dots,c_{k-2}]), & k\in\N\textrm{ and }c_{k}=0,\\
c_{0}+\frac{1}{c_{1}+\frac{1}{\ddots+\frac{1}{c_{k}}}}, & \text{otherwise}.
\end{cases}\label{eq:=000020phi=000020map=000020from=000020C=000020to=000020Q}
\end{equation}
In addition to that we set $\varphi([0]):=\infty$ and $\varphi([0,0,-1])=-1$.

Via the definition (\ref{eq:=000020phi=000020map=000020from=000020C=000020to=000020Q}),
the space $\Co$ extends the common convention of continued fraction
expansions, which is 
\[
\frac{p}{q}=c_{0}+\frac{1}{c_{1}+\frac{1}{\ddots+\frac{1}{c_{k}}}},
\]
by allowing additionally $c_{k}\in\{-1,0\}$. The first line in the
right hand side of (\ref{eq:=000020phi=000020map=000020from=000020C=000020to=000020Q})
is equivalent to substituting $c_{k}=-1$ in the continued fraction
expansion. The second line is more delicate; if one allows taking
$c_{k}\in\R$ then one gets
\[
\lim_{c_{k}\rightarrow0}\left(c_{0}+\frac{1}{c_{1}+\frac{1}{\ddots+\frac{1}{c_{k}}}}\right)=c_{0}+\frac{1}{c_{1}+\frac{1}{\ddots+\frac{1}{c_{k-2}}}},
\]
which is the rationale standing behind the definition $\varphi([0,c_{0},c_{1}\dots,c_{k-2},c_{k-1},0]):=\varphi([0,c_{0},c_{1}\dots,c_{k-2}])$
in (\ref{eq:=000020phi=000020map=000020from=000020C=000020to=000020Q}).
\begin{rem}
From the definition of the map $\varphi$, we get $\Im(\varphi)\subseteq(\Q\cap[0,1])\cup\{-1\}\cup\{\infty\}$.
A basic yet important observation is that the map $\varphi$ is not
injective. This may be seen already from its definition in (\ref{eq:=000020phi=000020map=000020from=000020C=000020to=000020Q}).
In addition, 
\[
\varphi([0,c_{0},c_{1}\dots,c_{k-2},c_{k-1},c_{k},1])=\varphi([0,c_{0},c_{1}\dots,c_{k-2},c_{k-1},c_{k}+1]),
\]
which is a common dual representation within continued fraction expansions,
\cite[ch. I.4]{Khinchin_book64}. In addition, one can check that
\[
\varphi(\co)=\infty\quad\Leftrightarrow\quad\co\in\left\{ [0],~[0,0,0],~[0,0,1,-1]\right\} .
\]
\end{rem}

The motivation behind using continued fraction expansions is for approximating
irrational $\alpha\in\left[0,1\right]\backslash\Q$ by rational values,
which allows to approximate aperiodic Hamiltonians, (\ref{eq:Hamiltonian=000020defined}),
by periodic ones. Specifically, each $\alpha\in\left[0,1\right]\backslash\Q$
may be uniquely presented in terms of an \emph{infinite} continued
fraction expansion, 
\begin{equation}
\alpha=c_{0}+\frac{1}{c_{1}+\frac{1}{c_{2}+\frac{1}{\ddots}}},\label{eq:=000020infinite=000020continued=000020fraction=000020expansion}
\end{equation}
where $c_{0}=0$ and $c_{n}\in\N$ for all $n\in\N$. We use this
to define for each $k\in\N$, 
\begin{align*}
\cok & :=[0,0,c_{1},\ldots,c_{k}]\qquad\textrm{and}\qquad\alk & :=\varphi(\cok).
\end{align*}
The values $\alk$ offer an optimal way to approximate $\alpha$ in
the sense $\lim_{k\rightarrow\infty}\alk=\alpha$, and thus we refer
to $\alpha_{k}$ as the \textit{$k$-th convergent} of $\alpha$,
\cite[ch. I.3]{Khinchin_book64}.

We further denote $\pqk:=\alpha_{k}$, where $p_{k},q_{k}\in\N$ are
chosen to be coprime. It is useful to extend this notation so that
it includes also the values $k\in\{-1,0\}$. This is done by setting
\begin{align*}
\alpha_{-1} & :=\varphi([0])=\infty,\quad p_{-1}=1,\quad q_{-1}=0,\\
\alpha_{0} & :=\varphi([0,0])=0,\quad p_{0}=0,\quad q_{0}=1.
\end{align*}
Note that for $k=-1$, we adopt the formal convention, $\alpha_{-1}=\frac{p_{-1}}{q_{-1}}=\infty$.
. The reason for introducing $p_{-1},p_{0,}q_{-1}$ and $q_{0}$ is
given by the following recursive formulas, \cite[thm. 1]{Khinchin_book64},
for $k\in\N_{0}$
\begin{align}
p_{k+1} & =c_{k+1}p_{k}+p_{k-1}\qquad\textrm{and}\qquad q_{k+1}=c_{k+1}q_{k}+q_{k-1}.\label{eq:=000020recursion=000020for=000020p_k,=000020q_k}
\end{align}

\begin{rem}
\label{rem:=000020why=000020we=000020need=000020C=000020space} It
is beneficial to make the analogy between the notations introduced
above and the notations in \cite{Raym95}. The notation $(k,p)$,
appearing first in \cite[prop. 2.2]{Raym95}, is replaced in this
review by $[0,0,c_{1},\ldots,c_{k-1},p]=[\co_{k-1},p]$. We do so,
since we find in \cite{BanBecLoe_24} that it is essential to keep
track of all numbers in the continued fraction expansion simultaneously
and consider values of $\co\in\Co$ which correspond to different
$\alpha\notin\Q$. This matter is not raised in \cite{Raym95}, where
it is sufficient to fix a single $\alpha\notin\Q$ and for that the
notation $(k,p)$ is adequate.
\end{rem}

\subsection{Sturmian words and mechanical words\protect\label{subsec:=000020Sturmian=000020words=000020and=000020mechanical=000020words}}

We present here a brief introduction to Sturmian words and mechanical
words. For elaborate surveys see \cite{Lothaire2002,Berstel_Sturmian_survey}.

We start by denoting for $\alpha\in[0,1]$ and $n\in\Z,$ 
\begin{equation}
\Sturm(n):=\chi_{\left[1-\alpha,1\right)}(n\alpha\mod 1).\label{eq:=000020def=000020Sturmian=000020sequence}
\end{equation}
Another equivalent representation of the sequence $\omega_{\alpha}\in\left\{ 0,1\right\} ^{\Z}$
is the following.
\begin{lem}
\label{lem:=000020lower=000020mechanical=000020word} \cite[lem. 1]{BIST89},\cite[def. 2.1]{Raym95}~\\
 Let $\alpha\in[0,1]$ and $n\in\Z$. Then
\begin{equation}
\Sturm(n)=\lfloor(n+1)\alpha\rfloor-\lfloor n\alpha\rfloor,\label{eq:=000020lower=000020mechanical=000020word}
\end{equation}
where $\left\lfloor \phantom{x}\right\rfloor $ is the floor function.
\end{lem}

\begin{proof}
First, observe that $\lfloor\alpha(n+1)\rfloor-\lfloor\alpha n\rfloor\in\{0,1\}$,
for all $n\in\Z$. Using that the claim follows from 
\begin{align*}
\omega_{\alpha}(n)=1 & \iff n\alpha\mod 1\in[1-\alpha,1)\\
 & \iff\exists m\in\Z:n\alpha\in[m+1-\alpha,m+1)\\
 & \iff\lfloor\alpha(n+1)\rfloor-\lfloor\alpha n\rfloor=1.
\end{align*}
\end{proof}
We use the notation $\omega_{\alpha}$ for both rational and irrational
values of $\alpha$. The infinite words defined by $\lfloor\alpha(n+1)\rfloor-\lfloor\alpha n\rfloor$
are also called \textit{(lower) mechanical words (with slope $\alpha$)}
\cite[sec. 2.1.2]{Lothaire2002}. If $\alpha=\frac{p}{q}\in[0,1]\cap\Q$,
then it is elementary to see that $\omega_{\alpha}$ is $q$-periodic,
i.e., $\omega_{\alpha}(n+q)=\omega_{\alpha}(n)$ for all $n\in\Z$.
If $\alpha\notin\Q$ then $\omega_{\alpha}$ is called a Sturmian
sequence, which is not a periodic word. In this case, it is useful
study the ($q_{k}$-)periodic words $\omega_{\alpha_{k}}$ as approximations
of $\omega_{\alpha}$ where $\alpha_{k}=\frac{p_{k}}{q_{k}}$ are
the $k$th convergents of $\alpha$ and $p_{k},q_{k}$ are coprime.

We have seen in (\ref{eq:=000020recursion=000020for=000020p_k,=000020q_k})
that there is a recursive formula which connects the period lengths,
$q_{k}$, for three subsequent $k$ values. We show next that the
periods themselves (i.e., the finite sub-words of length $q_{k}$)
are also connected via a recursive relation. We denote these periods
by $W_{k}\in\left\{ 0,1\right\} ^{q_{k}}$, setting 
\begin{equation}
W_{k}(i):=\omega_{\alpha_{k}}(i),\qquad0\leq i\leq q_{k}-1\label{eq:=000020Sturmian=000020period=000020def}
\end{equation}
and claiming the following.
\begin{lem}
\label{lem:=000020WordRecursion} The periods of the mechanical words
satisfy the following. 
\[
W_{0}=0,\quad\quad W_{1}=\underbrace{0\ldots0}_{c_{1}-1}1.
\]
If $k\geq2$ then
\[
W_{k}=\begin{cases}
W_{k-2}W_{k-1}^{c_{k}},\quad & k\equiv0\mod 2,\\
W_{k-1}^{c_{k}}W_{k-2},\quad & k\equiv1\mod 2,
\end{cases}
\]
where the power means a concatenation of words.
\end{lem}

In fact, it will be shown in the following sections that we care of
the words $W_{k}$ only up to a cyclic shift. In this sense, the expressions
$W_{k-2}W_{k-1}^{c_{k}}$ and $W_{k-1}^{c_{k}}W_{k-2}$ are the same.
Therefore, in the literature (and in particular in \cite[eq. (2)]{Raym95})
only the expression $W_{k-1}^{c_{k}}W_{k-2}$ is used. Indeed, $W_{k}$
equals to $W_{k-1}^{c_{k}}W_{k-2}$ up to a possible cyclic shift
is proven and used in various works, see e.g., \cite{BIST89}, \cite[prop. 2.2]{DamLen99-AbsEig}
and \cite[thm. 2.15]{Dam07_survey}. Another difference between the
common viewpoint and ours, is that usually, the periods $W_{k}$ of
the mechanical words $\omega_{\alk}$ are compared to the Sturmian
sequence $\omega_{\alpha}$, whereas we wish to compare between the
various periods to themselves, $W_{k}$, $W_{k-1}$ and $W_{k-2}$.

We decided to supplement the discussion in the current review by treating
the precise sub-word $W_{k}$, as it is defined in (\ref{eq:=000020Sturmian=000020period=000020def}),
and not only up to cyclic shift. We employ this exact representation
in Appendix~\ref{sec:=000020Floquet-Bloch=000020Theory} when defining
the finite dimensional Hamiltonian matrices (\ref{eq:=000020finite-dim=000020Ham=000020matrices})
for the Floquet-Bloch theory. These matrices also play a substantial
role in \cite{BanBecLoe_24}. For these reasons we have Lemma~\ref{lem:=000020WordRecursion}
as written here (and not only up to a cyclic shift) and its proof.
Statements which are similar to Lemma~\ref{lem:=000020WordRecursion}
can be also found in \cite[eq. (2.8)]{LuckPetritis86} and \cite[problem 2.2.10]{Lothaire2002}.

The reader is referred to Appendix~\ref{sec:RecursiveRelation_Sturmian}
for the proof of Lemma~\ref{lem:=000020WordRecursion} and related
results.

\section{Transfer matrices and the Discriminant \protect\label{sec:=000020Transfer=000020Matrices=000020and=000020Discriminant}}

In this section, we study the spectrum of the operator $\Ham$ (\ref{eq:Hamiltonian=000020defined})
while our main focus lies on rational $\alpha\in[0,1]$. In this work,
we use the rational approximations $\alpha_{k}$ to study the spectrum
of $\Ham$ for $\alpha\in[0,1]\setminus\Q$. If $\alpha=\frac{p}{q}\in[0,1]$
is rational, then $\omega_{\alpha}$ is $q-$periodic. Hence, the
spectrum of $\Ham$ is given by Floquet-Bloch theory using transfer
matrices and the discriminant, as is described in the following. We
note that there is an equivalent approach to Floquet-Bloch theory
by employing $q\times q$ Hamiltonian matrices which depend on the
Bloch parameter. This equivalent approach (and its connections to
transfer matrices) is described in Appendix~\ref{sec:=000020Floquet-Bloch=000020Theory}
and extensively used in \cite{BanBecLoe_24}.

\subsection{The spectrum of periodic operators via transfer matrices and the
discriminant \protect\label{subsec:=000020spectrum=000020via=000020transfer=000020matrices}}

We briefly present here some basic Floquet-Bloch theory using the
transfer matrix formalism. We keep the exposition as short as possible
and mainly intend to set the notation and the tools to be used in
the sequel. Two good sources for a more thorough introduction to the
one-dimensional discrete Floquet-Bloch theory are \cite[ch. 5]{Simon2011}
and \cite[ch. 7]{Tes00}.

Let $V\in\R$ and $\alpha\in[0,1]$. The difference equations associated
to $\Ham$ are 
\begin{align}
Eu(n)=u(n-1)+u(n+1)+V\omega_{\alpha}(n)u(n),\quad E\ \in\R,\,n\in\Z.\label{eq:EVEquation}
\end{align}

Solutions of this equation are studied via the so-called one-step
transfer matrices
\[
A_{\alpha}(n)(E,V):=\begin{pmatrix}E-V\omega_{\alpha}(n) & -1\\
1 & 0
\end{pmatrix},\qquad E\in\R,\,n\in\Z.
\]
Writing the difference equations in a matrix form, we obtain the following.
\begin{lem}
\label{lem:DefTransMat} Let $V\in\R$, $\alpha\in[0,1]$, $u:\Z\to\C$
and $E\in\R$ be such that the equations (\ref{eq:EVEquation}) are
satisfied for all $n\in\Z$. Then we have for all $n\in\Z$ 
\[
A_{\alpha}(n)(E,V)\begin{pmatrix}u(n)\\
u(n-1)
\end{pmatrix}=\begin{pmatrix}u(n+1)\\
u(n)
\end{pmatrix}\qquad\textrm{and}\qquad\det\left(A_{\alpha}(n)(E,V)\right)=1.
\]
\end{lem}

\begin{proof}
This follows by a short computation.
\end{proof}
Let $\co\in\Co$ with $\frac{p}{q}:=\varphi(\co)\in[0,1]$ with $p,q$
coprime. We observed in the previous section that the potential $\omega_{\alpha}$
is $q$-periodic for rational $\alpha=\frac{p}{q}$. Thus, it is advantageous
to define the ($q$-step) transfer matrix 
\begin{equation}
\Mc:=A_{\pq}(q-1)\cdot A_{\pq}(q-2)\dots A_{\pq}(1)\cdot A_{\pq}(0)\label{eq:=000020lem-transfer=000020matrix=000020defined}
\end{equation}
\[
\]
and get the following immediate implication.
\begin{lem}
\label{lem:=000020transfer=000020matrix=000020defined} Let $V\in\R$,
$\co\in\Co$ with $\frac{p}{q}:=\varphi(\co)\in[0,1].$ Then 
\[
\Mc(E,V)\begin{pmatrix}u(0)\\
u(-1)
\end{pmatrix}=\begin{pmatrix}u(q)\\
u(q-1)
\end{pmatrix}
\]
holds for all $u:\Z\to\C$ and $E\in\R$ satisfying (\ref{eq:EVEquation}).
In addition, $\det M_{\co}=1$.
\end{lem}

\begin{proof}
This is an immediate consequence of Lemma~\ref{lem:DefTransMat}.
\end{proof}
The Lemma~\ref{lem:=000020transfer=000020matrix=000020defined} extends
to $\co=[0,0,-1]$ for which $\varphi(\co)=-1$. To do so, we set
$\frac{p}{q}=\frac{-1}{1}$ and apply the definition of the mechanical
word from lemma~\ref{lem:=000020lower=000020mechanical=000020word}
to get for all $n\in\Z,$
\[
\omega_{-1}(n):=\left\lfloor (n+1)(-1)\right\rfloor -\left\lfloor n(-1)\right\rfloor =-1,
\]
and
\begin{equation}
A_{-1}(n)(E,V)=\begin{pmatrix}E+V & -1\\
1 & 0
\end{pmatrix}=M_{[0,0,-1]}(E,V).\label{eq:M_=00005B0,0,-1=00005D}
\end{equation}

We continue by using the recursive structure of the Sturmian words,
as expressed in Lemma~\ref{lem:=000020WordRecursion}, in order to
provide the recursive relations between the transfer matrices.
\begin{lem}
\label{lem:=000020transfer=000020matrix=000020formula} Let $E\in\R$
and $V\in\R$. Denote 
\[
\Mz(E,V):=\begin{pmatrix}1 & -V\\
0 & 1
\end{pmatrix},
\]
and let $\co=[0,0,c_{1},\ldots,c_{k}]\in\Co$ with $\varphi(\co)\in[0,1]\cup\{-1\}$.
\begin{enumerate}
\item \label{enu:=000020transfer=000020matrix=000020formula=000020=00005B0,0=00005D}If
$\co=[0,0]$, then 
\[
\Mzz(E,V)=\begin{pmatrix}E & -1\\
1 & 0
\end{pmatrix}.
\]
\item \label{enu:=000020transfer=000020matrix=000020formula=000020_=000020general}If
$k\in\N$ and $c_{k}\in\N_{0}$, then 
\begin{equation}
\Mc(E,V)=\begin{cases}
M_{[0,0,c_{1},\ldots,c_{k-1}]}(E,V)^{c_{k}}\cdot M_{[0,0,c_{1}\ldots,c_{k-2}]}(E,V),\quad & k\equiv0\mod 2,\\
M_{[0,0,c_{1}\ldots,c_{k-2}]}(E,V)\cdot M_{[0,0,c_{1},\ldots,c_{k-1}]}(E,V)^{c_{k}},\quad & k\equiv1\mod 2.
\end{cases}\label{eq:=000020lem-transfer=000020matrix=000020formula}
\end{equation}
\item \label{enu:=000020transfer=000020matrix=000020formula=000020_=000020trace=000020recursion}If
$k\in\N$ and $c_{k}\in\N_{-1}$, then
\begin{equation}
\tr\left(\Mc\right)=\tr\left(M_{[0,0,c_{1}\ldots,c_{k-2}]}\cdot M_{[0,0,c_{1},\ldots,c_{k-1}]}{}^{c_{k}}\right).\label{eq:=000020lem-transfer=000020matrix=000020formula=000020-=000020trace}
\end{equation}
\end{enumerate}
\end{lem}

\begin{rem*}
~
\begin{enumerate}
\item We clarify the lower recursive relations in the Lemma~\ref{lem:=000020transfer=000020matrix=000020formula}
by explicitly writing
\[
M_{[0,0,c_{1}]}=M_{[0]}\cdot M_{[0,0]}{}^{c_{1}}\quad\mathrm{and}\quad M_{[0,0,c_{1},c_{2}]}=M_{[0,0,c_{1}]}{}^{c_{2}}\cdot M_{[0,0]}(E,V).
\]
\item In addition, we note that (\ref{eq:=000020lem-transfer=000020matrix=000020formula})
does not hold for $c_{k}=-1$ (or rather, should be appropriately
modified), whereas (\ref{eq:=000020lem-transfer=000020matrix=000020formula=000020-=000020trace})
does hold also for all $c_{k}\in\Nmo$. This property of the trace
is important and will be used in the next subsection.
\end{enumerate}
\end{rem*}
\begin{proof}
For $k=0$, we get $\co=[0,0]$ and since $\omega_{\varphi([0,0])}=\omega_{0}=0^{\infty}$
we have $q_{0}=1$ and 
\[
\Mzz=A_{\frac{p_{0}}{q_{0}}}(0)=\begin{pmatrix}E & -1\\
1 & 0
\end{pmatrix},
\]
proving (\ref{enu:=000020transfer=000020matrix=000020formula=000020=00005B0,0=00005D}).

Next, we prove (\ref{enu:=000020transfer=000020matrix=000020formula=000020_=000020general}).
If $k=1$, then $\co=[0,0,c_{1}]$ with $c_{1}\in\N$, as otherwise
(i.e., if $c_{1}=0$) $\varphi(\co)=\infty$. Hence, $\varphi(\co)=\frac{1}{c_{1}}$
and Lemma~\ref{lem:=000020WordRecursion} implies 
\[
\omega_{\frac{1}{c_{1}}}(0)\omega_{\frac{1}{c_{1}}}(1)\dots\omega_{\frac{1}{c_{1}}}(c_{1}-2)\omega_{\co}(c_{1}-1)=00\dots01.
\]
Using the definition of the transfer matrix, (\ref{eq:=000020lem-transfer=000020matrix=000020defined}),
we get 
\begin{align*}
M_{[0,0,c_{1}]} & =A_{\frac{1}{c_{1}}}(c_{1}-1)\cdot A_{\frac{1}{c_{1}}}(c_{1}-2)\dots A_{\frac{1}{c_{1}}}(1)\cdot A_{\frac{1}{c_{1}}}(0)\\
 & =\begin{pmatrix}E-V & -1\\
1 & 0
\end{pmatrix}\begin{pmatrix}E & -1\\
1 & 0
\end{pmatrix}^{c_{1}-1}\\
 & =\begin{pmatrix}1 & -V\\
0 & 1
\end{pmatrix}\begin{pmatrix}E & -1\\
1 & 0
\end{pmatrix}\begin{pmatrix}E & -1\\
1 & 0
\end{pmatrix}^{c_{1}-1}\\
 & =M_{[0]}\cdot M_{[0,0]}{}^{c_{1}},
\end{align*}
which verifies the statement for $k=1$ and $c_{k}\in\N$. Note that
the case $k=1$ and $c_{k}=0$ results in $\co=[0,0,0]$ satisfying
$\varphi(\co)=\infty$ which is excluded by assumption. If $k\geq2$
and $c_{k}\notin\{-1,0\}$, then (\ref{enu:=000020transfer=000020matrix=000020formula=000020_=000020general})
follows from Lemma~\ref{lem:=000020WordRecursion}.

Now, let $k\geq2$ and $c_{k}=0$ . We have $\co=[0,0,c_{1},\ldots,c_{k-2},c_{k-1},0]$
and by definition of the evaluation map, we get $\varphi(\co)=\varphi([0,0,c_{1},\ldots,c_{k-2}])$.
Thus, $\Mc(E,V)=M_{[0,0,c_{1}\ldots,c_{k-2}]}(E,V)$ follows since
$M_{\co}$ only depends on the evaluation $\varphi(\co)$. In particular,
(\ref{eq:=000020lem-transfer=000020matrix=000020formula}) holds also
for $c_{k}=0$ (regardless of the parity of $k$).

It is left to prove (\ref{eq:=000020lem-transfer=000020matrix=000020formula=000020-=000020trace}).
As a matter of fact, the cyclic property of the trace yields that
(\ref{eq:=000020lem-transfer=000020matrix=000020formula=000020-=000020trace})
is a direct consequence of (\ref{eq:=000020lem-transfer=000020matrix=000020formula})
if $c_{k}\neq-1$.

If $c_{k}=-1$ we have $\co=[0,0,c_{1},\ldots,c_{k-1},-1]$ and $\varphi(\co)=\varphi([0,0,c_{1},\ldots,c_{k-1}-1])$
and by definition $M_{\co}=M_{[0,0,c_{1},\ldots,c_{k-1}-1]}$. If
$k\geq2$ we get 
\begin{align*}
\tr\left(M_{\co}\right) & =\tr\left(M_{[0,0,c_{1},\ldots,c_{k-1}-1]}\right)\\
 & =\tr\left(M_{[0,0,c_{1}\ldots,c_{k-3}]}\cdot M_{[0,0,c_{1},\ldots,c_{k-2}]}{}^{c_{k-1}-1}\right)\\
 & =\tr\left(\left(M_{[0,0,c_{1}\ldots,c_{k-3}]}\cdot M_{[0,0,c_{1},\ldots,c_{k-2}]}{}^{c_{k-1}}\right)\cdot M_{[0,0,c_{1},\ldots,c_{k-2}]}{}^{-1}\right)\\
 & =\tr\left(M_{[0,0,c_{1}\ldots,c_{k-1}]}\cdot M_{[0,0,c_{1},\ldots,c_{k-2}]}{}^{-1}\right)\\
 & =\tr\left(M_{[0,0,c_{1}\ldots,c_{k-2}]}\cdot M_{[0,0,c_{1},\ldots,c_{k-1}]}{}^{-1}\right),
\end{align*}
where in the second and fourth equality we used (\ref{eq:=000020lem-transfer=000020matrix=000020formula})
together with the cyclic property of the trace (which allows not to
distinguish between even and odd $k$'s) and in the last equality
we used that $\tr(M)=\tr(M^{-1})$ whenever $\det M=1$ (and this
holds for transfer matrices by Lemma~\ref{lem:DefTransMat}). All
is left is to check the case $k=1$ and $c_{k}=-1$. In this case,
$\co=[0,0,-1]$, $\varphi(\co)=-1$ and a straightforward computation
invoking (\ref{eq:M_=00005B0,0,-1=00005D}) gives
\[
\tr\left(M_{[0]}M_{[0,0]}^{-1}\right)=V+E=\tr\left(M_{[0,0,-1]}\right).
\]
\end{proof}
By standard Floquet-Bloch theory (applied to one-dimensional Jacobi
operators), the spectrum of $\Hrat$ (for $\frac{p}{q}=\varphi(\co)$)
may be described by the trace of $M_{\co}$. Therefore, define the
discriminant $\tc$ for $\co\in\Co$ by
\[
\tc(E,V):=\tr(\Mc(E,V))
\]
and 
\[
\sigc(V):=\tc(\cdot,V)^{-1}([-2,2]).
\]

\begin{example}
\label{exa:Spectrum_=00005B0=00005D=000020and=000020=00005B0,0,-1=00005D}Observe
that if $\varphi(\co)=\infty$, then $t_{\co}(E,V)=2$ and so $\sigma_{\co}(V)=\R$
hold. If $\varphi(\co)=-1$, then (\ref{eq:M_=00005B0,0,-1=00005D})
leads to $t_{[0,0,-1]}(E,V)=E+V$ and so $\sigma_{[0,0,-1]}(V)=[-2-V,2-V]$
for all $V\in\R$.
\end{example}

If $\varphi(\co)\neq\infty$, we bring here a summary of useful properties
which may be found for example in \cite[sec. 5.4]{Simon2011}, \cite[sec. 7.1]{Tes00}.
\begin{prop}
\label{prop:=000020Floquet-Bloch=000020via=000020transfer=000020matrix}
Let $\co\in\Co$ with $\varphi(\co)\neq\infty$ and $V\in\R$.

Then the following properties hold:
\begin{enumerate}
\item $\sigma_{\co}(V)=\sigma\left(H_{\varphi(\co),V}\right).$
\item Denoting $\frac{p}{q}=\varphi(\co)$, the set $\tc(\cdot,V)^{-1}((-2,2))$
consists of exactly $q$ open intervals.
\item The discriminant $\tc$ is monotone on each connected component of
$\tc(\cdot,V)^{-1}((-2,2))$.
\end{enumerate}
\end{prop}

The connected components $\tc(\cdot,V)^{-1}((-2,2))$ mentioned in
Proposition~\ref{prop:=000020Floquet-Bloch=000020via=000020transfer=000020matrix}
are the interior of the so-called \emph{spectral bands} of $\sigc(V)$.
The spectral bands are closed intervals whose edge points are given
by $\tc(\cdot,V)^{-1}\left(\left\{ -2,2\right\} \right)$. In general,
it is possible that different spectral bands overlap at their endpoint.
However, this is not the case for the approximations of the Sturmian
Hamiltonian if $V\neq0$, see Proposition~\ref{prop:=000020there=000020are=000020q=000020spectral=000020bands}.
\begin{rem*}
We were trying to keep the notation here close to the one in \cite{Raym95}
and in particular use the notations $M$ and $t$ for the transfer
matrix and its trace (discriminant). Nevertheless, we deviate in the
subscript notation, using $M_{\co}$ instead of $M_{k}$ and $\tc$
instead of $t_{(k,p)}$. The reasons for this change are exactly the
ones which are specified in Remark~\ref{rem:=000020why=000020we=000020need=000020C=000020space}.
\end{rem*}

\subsection{Algebraic identities of the transfer matrices and their traces}

In this subsection, we develop some identities for the traces $\tc$.
These identities will be used in the following sections to derive
spectral properties of the periodic operators $H_{\varphi(\co),V}$.
Some of these identities can be found in \cite{BIST89,Raym95}. However,
we use here a slightly different notation (to fit \cite{BanBecLoe_24})
and in particular, use the mechanical word sequences $\omega_{\varphi(\co)}$
for all values of $\co\in\Co$, rather than a single fixed Sturmian
sequence $\omega_{\alpha}$ which is the approach used in \cite{BIST89,Raym95}.
For the sake of a self-contained presentation, we provide here complete
proofs of all the relevant identities.

We start by noting a basic property of the discriminant $\tc=\tr(M_{\co})$:
even though it is a function of $\co$, it depends only on the value
$\varphi(\co)$. This fundamental property deserves an explicit mention
here, as it is substantially used in the sequel.
\begin{lem}
\label{lem:=000020Trace=000020Depends=000020On=000020Value=000020Only}
Let $\co,\widetilde{\co}\in\Co$ be such that $\varphi(\co)=\varphi(\widetilde{\co})$,
then 
\[
\tc(\cdot,V)=t_{\widetilde{\co}}(\cdot,V)\textrm{\ensuremath{\quad}and\ensuremath{\quad}}\sigma_{\co}(V)=\sigma_{\widetilde{\co}}(V)\quad\textrm{for all}~V\in\R.
\]
\end{lem}

\begin{proof}
This property is immediate from the definition of $M_{\co}$, (\ref{eq:=000020lem-transfer=000020matrix=000020defined}),
which depends purely on the value of $\varphi(\co)$, if $\varphi(\co)\neq\infty$.
For $\varphi(\co)=\infty$, the matrix $M_{\co}$ does depend on $\co\in\{[0],[0,0,0],[0,0,1-1]\}$.
However, a short computation gives $\tc(E,V)=2$ if $\varphi(\co)=\infty$.
The statement $\sigma_{\co}(V)=\sigma_{\widetilde{\co}}(V)$ follows
directly from the equality of the traces.
\end{proof}
For the sake of representation, we write $\tc$ instead of $\tc(E,V)$.
As an immediate corollary we get.
\begin{cor}
\label{cor:FirstTraceIds} Let $[0,c_{0},c_{1},\dots,c_{k}]\in\Co$
with $c_{k}\in\N$. Then the following identities hold 
\begin{align*}
t_{[0,c_{0},\ldots,c_{k-1},c_{k},0]} & =t_{[0,c_{0},\ldots,c_{k-1}]}\\
t_{[0,c_{0},\ldots,c_{k-1},c_{k},-1]} & =t_{[0,c_{0},\ldots,c_{k-1},c_{k}-1]}\\
t_{[0,c_{0},\ldots,c_{k-1},c_{k},1]} & =t_{[0,c_{0},\ldots,c_{k-1},c_{k}+1]}.
\end{align*}
\end{cor}

\begin{proof}
This is an implication of Lemma~\ref{lem:=000020Trace=000020Depends=000020On=000020Value=000020Only}
together with the identities 
\[
\varphi([\co,0])=\varphi([0,c_{0},\ldots,c_{k-1}]),\qquad\varphi([\co,-1])=\varphi([0,c_{0},\ldots,c_{k-1},c_{k}-1]),
\]
and $\varphi([\co,1])=\varphi([0,c_{0},\ldots,c_{k-1},c_{k}+1])$
for $\co=[0,c_{0},\ldots,c_{k}]$.
\end{proof}

\begin{lem}
\label{lem:TraceEvolution} \cite[prop. 2.2]{Raym95} Let $\co\in\Co$
and $m\in\Nz$ such that $[\co,m]\in\Co$. Then 
\[
t_{[\co,m+1]}=\tc t_{[\co,m]}-t_{[\co,m-1]}.
\]
\end{lem}

\begin{proof}
Let $\co'\in\Co$ and $c_{k}\in\N_{0}$ be such that $\co=[\co',c_{k}]$.
Observe that $A^{2}=\tr(A)A-\det(A)\id_{2}$ for complex $2\times2$
matrices (this is actually a special case of Cayley-Hamilton theorem).
In particular, we will use this identity for the transfer matrices,
for which $\det\left(M_{\co}\right)=1$ by Lemma~\ref{lem:DefTransMat}.
With this at hand we get 
\begin{align*}
t_{[\co,m+1]}=t_{[\co',c_{k},m+1]} & =\tr(M_{[\co',c_{k},m+1]})\\
 & =\tr(M_{\co'}\Mc^{m+1})\\
 & =\tr(M_{\co'}\Mc^{m-1}\Mc^{2})\\
 & =\tr(M_{\co'}\Mc^{m-1}[\tr(\Mc)\Mc-\det(\Mc)\id_{2}])\\
 & =\tr(\Mc)\tr(M_{\co'}\Mc^{m})-\tr(M_{\co'}\Mc^{m-1})\\
 & =\tr(\Mc)\tr(M_{[\co,m]})-\tr(M_{[\co,m-1]})\\
 & =\tc t_{[\co,m]}-t_{[\co,m-1]},
\end{align*}
where in the second and sixth lines we used (\ref{eq:=000020lem-transfer=000020matrix=000020formula=000020-=000020trace})
of Lemma~\ref{lem:=000020transfer=000020matrix=000020formula}.
\end{proof}
Next, we aim at generalizing Lemma~\ref{lem:TraceEvolution}. To
do so, we introduce the \textit{dilated Chebyshev polynomials of the
second kind} $\CP_{l}:\R\to\R$ (see \cite[eq. (18.1.3)]{DigLibMathFunc}).
These polynomials are inductively defined by
\begin{equation}
\CP_{-1}(x):=0,\quad\CP_{0}(x):=1,\quad\CP_{l}(x)=x\CP_{l-1}(x)-\CP_{l-2}(x).\label{eq:=000020Chebyshev=000020recursion}
\end{equation}
Appendix~\ref{sec:=000020Cheby=000020poly} contains an elaborate
account on these polynomials, their connection to the 'usual' Chebyshev
polynomials of the second kind and various useful identities which
are used in this review as well as in \cite{BanBecLoe_24}.
\begin{lem}
\label{lem:=000020Trace=000020Cheby=000020Formula} \cite[prop. 2.2]{Raym95}
Let $\co\in\Co$ and $m\geq l\geq-1$ such that $[\co,m]\in\Co$ ,
then 
\[
t_{[\co,m+1]}=\CP_{l+1}(\tc)t_{[\co,m-l]}-\CP_{l}(\tc)t_{[\co,m-l-1]}.
\]
\end{lem}

\begin{proof}
We fix $m\in\N_{-1}$ and prove the statement by induction over $l\in\N_{-1}$.
For $l=-1$, we use Lemma~\ref{lem:TraceEvolution} to get
\[
\CP_{0}(\tc)t_{[\co,m-l]}-\CP_{-1}(\tc)t_{[\co,m-l-1]}=1\cdot t_{[\co,m+1]}+0\cdot t_{[\co,m]}=t_{[\co,m+1]}.
\]
Now assume the statement is correct for $m>l\geq-1$. We then get
\begin{align*}
t_{[\co,m+1]} & =\CP_{l+1}(\tc)t_{[\co,m-l]}-\CP_{l}(\tc)t_{[\co,m-l-1]}\\
 & =\CP_{l+1}(\tc)\left[\tc t_{[\co,m-l-1]}-t_{[\co,m-l-2]}\right]-\CP_{l}(\tc)t_{[\co,m-l-1]}\\
 & =\left[\tc\CP_{l+1}(\tc)-\CP_{l}(\tc)\right]t_{[\co,m-l-1]}-\CP_{l+1}(\tc)t_{[\co,m-l-2]}\\
 & =\CP_{l+2}(\tc)t_{[\co,m-(l+1)]}-\CP_{l+1}(\tc)t_{[\co,m-(l+1)-1]},
\end{align*}
where we used Lemma~\ref{lem:TraceEvolution} in the second equality,
and the Chebyshev polynomial recursion in the last equality.
\end{proof}
In the following, an extra parameter $\ell\in\left\{ -1,0\right\} $
is introduced. Later in this review so-called spectral bands $\Ic$
in $\sigc$ of backward type $A$ and $B$ are introduced that are
defined by adding to $\co$ the digit $0$ if $\Ic$ is backward type
$A$ and $-1$ if $\Ic$ is backward type $B$, see Definition~\ref{def:=000020backward=000020type}
below.
\begin{cor}
\label{cor:ExtChebyForm1} Let $\co\in\Co$ and $m\in\N$ be such
that $[\co,m]\in\Co$. For $\ell\in\{-1,0\}$, we have 
\[
t_{[\co,m]}=\CP_{m-\ell-1}(\tc)t_{[\co,1+\ell]}-\CP_{m-\ell-2}(\tc)t_{[\co,\ell]}.
\]
\end{cor}

\begin{proof}
This is a direct consequence of Lemma~\ref{lem:=000020Trace=000020Cheby=000020Formula}.
\end{proof}
We proceed to apply this corollary to get another useful identity
involving Chebyshev polynomials and the traces.
\begin{lem}
\label{lem:ExtChebyForm2} Let $\co\in\Co$ and $m\in\N$ be such
that $[\co,m]\in\Co$. For any $\xi\in\{0,1\}$ and $\ell\in\{-1,0\}$
we have 
\[
\CP_{m-1-\ell}(\tc)\left[t_{[\co,m-1]}+(-1)^{\xi}t_{[\co,1+\ell]}\right]=\left[\CP_{m-2-\ell}(\tc)+(-1)^{\xi}\right]\left[t_{[\co,m]}+(-1)^{\xi}t_{[\co,\ell]}\right].
\]
\end{lem}

\begin{proof}
First, we use twice the recursion relation for the Chebyshev polynomials
to get 
\begin{align*}
S_{l}(x)S_{l-2}(x)-S_{l-1}(x)^{2} & =xS_{l-1}(x)S_{l-2}(x)-S_{l-1}(x)^{2}-S_{l-2}(x)^{2}\\
 & =S_{l-1}(x)S_{l-3}(x)-S_{l-2}(x)^{2}.
\end{align*}
In particular, we get that this expression is independent of $l$
and therefore 
\[
S_{l}(x)S_{l-2}(x)-S_{l-1}(x)^{2}=S_{1}(x)S_{-1}(x)-S_{0}(x)^{2}=-1.
\]
Using this identity, Lemma~\ref{lem:=000020Trace=000020Cheby=000020Formula}
and Corollary~\ref{cor:ExtChebyForm1}, the lemma follows by straightforward
computation. For example, for $\ell=-1$, 
\begin{align*}
 & \CP_{m}(\tc)\left[t_{[\co,m-1]}+(-1)^{\xi}t_{[\co,0]}\right]\\
 & =\CP_{m}(\tc)\left[\CP_{m-1}(\tc)t_{[\co,0]}-S_{m-2}(\tc)t_{[\co,-1]}+(-1)^{\xi}t_{[\co,0]}\right]\\
 & =\CP_{m}(\tc)\CP_{m-1}(\tc)t_{[\co,0]}-\CP_{m}(\tc)S_{m-2}(\tc)t_{[\co,-1]}+(-1)^{\xi}\CP_{m}(\tc)t_{[\co,0]}\\
 & =\CP_{m}(\tc)\CP_{m-1}(\tc)t_{[\co,0]}-\left[\CP_{m-1}^{2}(\tc)-1\right]t_{[\co,-1]}+(-1)^{\xi}\CP_{m}(\tc)t_{[\co,0]}\\
 & =\CP_{m}(\tc)t_{[\co,0]}\left[\CP_{m-1}(\tc)+(-1)^{\xi}\right]-\left[\CP_{m-1}^{2}(\tc)-1\right]t_{[\co,-1]}\\
 & =\left[\CP_{m-1}(\tc)+(-1)^{\xi}\right]\left[\CP_{m}(\tc)t_{[\co,0]}-\left(\CP_{m-1}(\tc)+(-1)^{\xi+1}\right)t_{[\co,-1]}\right]\\
 & =\left[\CP_{m-1}(\tc)+(-1)^{\xi}\right]\left[\CP_{m}(\tc)t_{[\co,0]}-\CP_{m-1}(\tc)t_{[\co,-1]}+(-1)^{\xi}t_{[\co,-1]}\right]\\
 & =\left[\CP_{m-1}(\tc)+(-1)^{\xi}\right]\left[t_{[\co,m]}+(-1)^{\xi}t_{[\co,-1]}\right].
\end{align*}
The statement for $\ell=0$ follows the same lines except that the
case $m=1$ needs to be treated separately (since we used Corollary~\ref{cor:ExtChebyForm1}
moving from the first to the second line, which cannot applied if
$\ell=0$ and $m=1$).
\end{proof}

\subsection{The Fricke-Vogt invariant\protect\label{subsec:=000020Fricke-Vogt}}

The Fircke-Vogt invariant serves an important role in the spectral
analysis of Sturmian Hamiltonians. We review here this well-known
part of the theory, and rephrase it according to our convention to
use the space $\Co$.

Denote by $\left[\cdot,\cdot\right]$ the matrix commutator $\left[A,B\right]:=AB-BA$.
Note that $\tr\left(\left[A,B\right]\right)=0,$ as $\tr$ is linear
and $\tr(AB)=\tr(BA)$.
\begin{lem}
\label{lem:CommutatorID}\cite[prop. 2.3]{Raym95} Let $V\in\R$,
$\co=[c_{-1},c_{0},\ldots,c_{k}]\in\Co$ with $k\in\Nmo$ and $[\co,m,n]\in\Co$
with $m\in\N_{0}$ and $n\in\N_{-1}$. Then 
\[
\left[M_{[\co,m]},M_{\co}M_{[\co,m]}^{n}\right]^{2}=V^{2}\id_{2},
\]
where $\id_{2}$ is the $2\times2$ identity matrix.
\end{lem}

\begin{proof}
Denote $A:=\left[M_{[\co,m]},M_{\co}M_{[\co,m]}^{n}\right]$. As for
each $2\times2$ matrix (e.g., as a special case of Cayley-Hamilton
theorem) we have 
\[
A^{2}=\tr(A)A-\det(A)\id_{2}=-\det(A)\id_{2},
\]
where in the second equality we used that $\tr\left(A\right)=0$.
Hence, to validate the statement we need to show $\det(A)=-V^{2}$.
Computing the determinant gives 
\begin{align*}
\det(A) & =\det\left(M_{[\co,m]}M_{\co}M_{[\co,m]}^{n}-M_{\co}M_{[\co,m]}^{n}M_{[\co,m]}\right)\\
 & =\det\left(M_{[\co,m]}M_{\co}-M_{\co}M_{[\co,m]}\right)\det(M_{[\co,m]})^{n}\\
 & =\det\left(\left[M_{[\co,m]},M_{\co}\right]\right),
\end{align*}
where we used that the determinant of a transfer matrix is one by
Lemma~\ref{lem:=000020transfer=000020matrix=000020defined}. To finish
the proof, we use induction over $k\in\Nmo$ to show $\det\left(\left[M_{[\co,m]},M_{\co}\right]\right)=-V^{2}$
(for any $[\co,m]\in\Co$ and $\co=[c_{-1},c_{0},\ldots,c_{k}]$).
For the induction base, we observe 
\begin{align*}
\left[\Mzz,\Mz\right] & =\begin{pmatrix}E & -1\\
1 & 0
\end{pmatrix}\begin{pmatrix}1 & -V\\
0 & 1
\end{pmatrix}-\begin{pmatrix}1 & -V\\
0 & 1
\end{pmatrix}\begin{pmatrix}E & -1\\
1 & 0
\end{pmatrix}\\
 & =\begin{pmatrix}E & -VE-1\\
1 & -V
\end{pmatrix}-\begin{pmatrix}E-V & -1\\
1 & 0
\end{pmatrix}\\
 & =\begin{pmatrix}V & -VE\\
0 & -V
\end{pmatrix},
\end{align*}
and indeed $\det\left(\left[\Mzz,\Mz\right]\right)=-V^{2}$. For the
induction step, suppose the statement is true for $k\in\N_{-1}$.
Note that for $2\times2$-matrices $B$ and $C$, we have $[B,C]=-[C,B]$
and $\det(-B)=\det(B).$ Thus, the previous identity on the determinant
of the commutator yields 
\begin{align*}
\det\left(\left[M_{[\co,c_{k+1},m]},M_{[\co,c_{k+1}]}\right]\right) & =\det\left(\left[M_{[\co,c_{k+1}]},M_{[\co,c_{k+1},m]}\right]\right)\\
 & =\det\left(\left[M_{[\co,c_{k+1}]},\Mc M_{[\co,c_{k+1}]}^{m}\right]\right)\\
 & =\det\left(\left[M_{[\co,c_{k+1}]},\Mc\right]\right)=-V^{2},
\end{align*}
where in the second equality we used (\ref{eq:=000020lem-transfer=000020matrix=000020formula})
of Lemma~\ref{lem:=000020transfer=000020matrix=000020formula} assuming
$k$ is odd.If $k$ is even, as similar computation leads to the result.
\end{proof}
\begin{prop}
\cite[prop. 2.3]{Raym95} [Fricke-Vogt-Invariant]\label{prop:=000020Fricke-Vogt}
Let $V\in\R$, $\co\in\Co$ and $m\in\Nz$ such that $[\co,m-1]\in\Co$,
then 
\[
\tc^{2}+\tcm^{2}+t_{[\co,m-1]}^{2}-\tc\tcm t_{[\co,m-1]}=4+V^{2}.
\]
\end{prop}

To prove this proposition we use the following algebraic identity.
\begin{lem}
\label{lem:=000020trace(AB)} Let $A,B$ be two real $2\times2$ matrices
with $\det(A)=\det(B)=1$. Then 
\[
\tr(AB)=\tr(A)\tr(B)-\tr(A^{-1}B).
\]
\end{lem}

\begin{proof}
Since $A$ is a $2\times2$ matrix with $\det(A)=1$, we conclude
$A+A^{-1}=\tr(A)\id_{2}$. Hence, $\tr(AB)=\tr\left(\tr(A)B-A^{-1}B\right)=\tr(A)\tr(B)-\tr(A^{-1}B).$
\end{proof}
\begin{proof}
[Proof of Proposition \ref{prop:=000020Fricke-Vogt}] Let $\co'\in\Co$
and $c_{k}\in\N_{0}$ be such that $\co=[\co',c_{k}]$. . Denoting
$A:=\left[\Mc,M_{\co'}\Mc^{m}\right]$ and applying Lemma~\ref{lem:CommutatorID}
yields 
\[
\tr(A^{2})=\tr(V^{2}\id_{2})=2V^{2}.
\]
On the other hand, a direct computation of $\tr(A^{2})$ gives

\begin{align*}
\tr(A^{2}) & =\tr\left((\Mc M_{\co'}\Mc^{m}-M_{\co'}\Mc^{m+1})^{2}\right)\\
 & =\tr((\Mc M_{\co'}\Mc^{m})^{2})+\tr((M_{\co'}\Mc^{m+1})^{2})-2\tr(\Mc M_{\co'}\Mc^{m}M_{\co'}\Mc^{m+1})\\
 & =2\tr((M_{\co'}\Mc^{m+1})^{2})-2\tr(M_{\co'}\Mc^{m}M_{\co'}\Mc^{m+2}).
\end{align*}
For the first term we use the identity $B^{2}=\tr(B)B-\det(B)\id_{2}$
for the $2\times2$-matrix $B=M_{\co'}\Mc^{m+1}$ and then Lemma~\ref{lem:=000020transfer=000020matrix=000020formula}~(\ref{enu:=000020transfer=000020matrix=000020formula=000020_=000020trace=000020recursion}),
Lemma~\ref{lem:=000020transfer=000020matrix=000020formula}~(\ref{enu:=000020transfer=000020matrix=000020formula=000020_=000020trace=000020recursion})
andLemma~\ref{lem:TraceEvolution} lead to
\begin{align*}
\tr((M_{\co'}\Mc^{m+1})^{2}) & =\left(\tr(M_{\co'}\Mc^{m+1})\right)^{2}-2\\
 & =t_{[\co,m+1]}^{2}-2\\
 & =(\tc\tcm-t_{[\co,m-1]})^{2}-2\\
 & =\tc^{2}\tcm^{2}+t_{[\co,m-1]}^{2}-2\tc\tcm t_{[\co,m-1]}-2.
\end{align*}

For the second term we apply Lemma~\ref{lem:=000020trace(AB)} and
Lemma~\ref{lem:=000020transfer=000020matrix=000020formula}~(\ref{enu:=000020transfer=000020matrix=000020formula=000020_=000020trace=000020recursion})
to get 
\begin{align*}
\tr(M_{\co'}\Mc^{m}M_{\co'}\Mc^{m+2}) & =\tr(M_{\co'}\Mc^{m})\tr(M_{\co'}\Mc^{m+2})-\tr((M_{\co'}\Mc^{m})^{-1}M_{\co'}\Mc^{m+2})\\
 & =\tcm t_{[\co,m+2]}-\tr(\Mc^{2})\\
 & =\tcm t_{[\co,m+2]}-\tc^{2}+2\\
 & =\tcm(\tc t_{[\co,m+1]}-\tcm)-\tc^{2}+2\\
 & =\tc\tcm t_{[\co,m+1]}-\tcm^{2}-\tc^{2}+2\\
 & =\tc\tcm(\tc\tcm-t_{[\co,m-1]})-\tcm^{2}-\tc^{2}+2\\
 & =\tc^{2}\tcm^{2}-\tc\tcm t_{[\co,m-1]}-\tcm^{2}-\tc^{2}+2,
\end{align*}
where in the third equality we used the identity $B^{2}=\tr(B)B-\det(B)\id_{2}$
with $B=M_{\co}$, and in the fourth and sixth equality we used Lemma~\ref{lem:TraceEvolution}.

Combining the identities above provides the statement of the proposition.
\end{proof}

\section{The spectra of periodic approximations of Sturmian Hamiltonians\protect\label{sec:=000020spectra=000020of=000020periodic=000020approximants}}

We start applying the tools from the previous section in order to
study the spectral bands of the periodic approximations of the Sturmian
Hamiltonian, as is done in \cite[sec. 3.1]{Raym95}. We start by providing
general results for all Sturmian Hamiltonians (Subsection~\ref{subsec:=000020spectral=000020properties=000020V>0})
and then restrict to $V>4$ where further analysis may be obtained
(Subsection~\ref{subsec:=000020spectral=000020properties=000020V>4}).

\subsection{Basic spectral properties for all $V\protect\neq0$ \protect\label{subsec:=000020spectral=000020properties=000020V>0}
}

We provide basic properties on the spectrum of a periodic approximation
of a Sturmian Hamiltonians, i.e., $\Hrat$. To do so, we mainly use
the transfer matrices and the discriminant, as was introduced in the
previous section. The results in this subsection appeared already
in \cite{Casdagli1986,Sut87,BIST89}. . Since the results here apply
for all $V\neq0$, we tend to omit (only in this subsection) the notation
$V$ from the the proofs.

The following proposition is a refinement of Proposition~\ref{prop:=000020Floquet-Bloch=000020via=000020transfer=000020matrix}
for the operators $\Hrat$. Its first part appears in \cite[prop. 3.1,(i)]{Raym95}.
\begin{prop}
\label{prop:=000020there=000020are=000020q=000020spectral=000020bands}
Let $V\neq0$ and $\co\in\Co$ with $\pq:=\varphi(\co)\neq\infty$
and $p$ and $q$ coprime. Then the following assertions hold.
\begin{enumerate}
\item \label{enu:=000020prop-spectral=000020bands-1}The spectrum $\sigc(V)=\sigma(\Hrat)$
consists of exactly $q$ connected components which are closed intervals.\\
As usual, we call these intervals, the \emph{spectral bands} of $\Hrat$
(or of $\sigma(\Hrat)$).
\item \label{enu:=000020prop-spectral=000020bands-2}The restriction of
the discriminant $\tc$ to each of the spectral bands is strictly
monotone.
\end{enumerate}
\end{prop}

\begin{proof}
We need to prove only the first part of the proposition, as the second
part is classical (see e.g., \cite[thm. 5.4.2]{Simon2011}). By \cite[thm. 5.4.2]{Simon2011}
the spectrum of a $q$-periodic Jacobi operator (such as $\Hrat$)
consists of $q$ closed intervals, which might overlap only at their
boundaries. Assume by contradiction that $E$ is such a point where
two intervals overlap. By \cite[thm. 5.4.3]{Simon2011} this implies
that $\Mc(E)=\pm\id_{2}$, for $\co\in\Co$ such that $\varphi(\co)=\frac{p}{q}$.
Substituting this in Lemma~\ref{lem:CommutatorID} gives $V^{2}\id_{2}=\left[M_{[\co,m]},M_{\co}M_{[\co,m]}^{n}\right]^{2}=0$,
for any $m,n\in\N$. Hence, we get $V=0$ and a contradiction.
\end{proof}
Next, we rephrase a statement from \cite[prop. 4]{BIST89} and immediately
apply it to connect the spectra $\sigc$.
\begin{lem}
\label{lem:TraceStaysBig} Let $V\neq0$ and $\co=[0,c_{0},c_{1},\ldots,c_{k}]\in\Co$.
Let $E\in\R$ and $i<k$. If 
\[
\left|t_{[0,c_{0},c_{1},\ldots,c_{i-2}]}(E)\right|>2\quad\textrm{and}\quad\left|t_{[0,c_{0},c_{1},\ldots,c_{i-1}]}(E)\right|>2
\]
thenthere exists $C>1$ such that for all $i\leq j\leq k$, $\left|t_{[0,c_{0},c_{1},\ldots,c_{j}]}(E)\right|>2C^{q_{j}}$,
where $\varphi([0,c_{0},c_{1},\ldots,c_{j}]\}=\frac{p_{j}}{q_{j}}$
with $p_{j},q_{j}$ coprime.
\end{lem}

\begin{proof}
This follows from \cite[prop. 4]{BIST89} by fixing an $\alpha\in[0,1]\setminus\Q$
with where the first digits of the continuous fraction expansion of
$\alpha$ coincide with $c_{0},c_{1},\ldots,c_{k}$.
\end{proof}
\begin{lem}
[{\cite[prop. 3.1,(ii)]{Raym95}} spectral monotonicity property]\label{lem:MonotonicityProp}
Let $V\neq0$ and let $\co=[0,c_{0},c_{1},\ldots,c_{k}]\in\Co$ with
$\varphi(\co)\geq0$ and $k\in\N_{0}$. Then 
\[
\sigma_{\co}(V)\subseteq\sigma_{[0,c_{0},c_{1},\ldots,c_{k-2}]}(V)\cup\sigma_{[0,c_{0},c_{1},\ldots,c_{k-1}]}(V).
\]
In addition, if $[\co,-1]\in\Co$, then 
\[
\sigc(V)\subseteq\sigcz(V)\cup\sigma_{[\co,-1]}(V).
\]
\end{lem}

\begin{proof}
We start by proving the first inclusion. If $E\notin\sigma_{[0,c_{0},c_{1},\ldots,c_{k-2}]}\cup\sigma_{[0,c_{0},c_{1},\ldots,c_{k-1}]}$,
then Proposition~\ref{prop:=000020Floquet-Bloch=000020via=000020transfer=000020matrix}
implies 
\[
\left|t_{[0,c_{0},c_{1},\ldots,c_{k-2}]}(E)\right|>2\quad\textrm{and}\quad\left|t_{[0,c_{0},c_{1},\ldots,c_{k-1}]}(E)\right|>2.
\]
Thus, Lemma~\ref{lem:TraceStaysBig} leads to $\left|\tc(E)\right|>2$
and by Proposition~\ref{prop:=000020Floquet-Bloch=000020via=000020transfer=000020matrix},
$E\notin\sigc$, which proves the first inclusion.

To prove the second inclusion, note first that if $k=0$ and $\co=[0,0]$,
then the inclusion is trivial as $\sigma_{[\co,0]}=\sigma_{[0]}=\R$.
Suppose now $k\geq1.$Then the condition $[\co,-1]\in\Co$ implies
that $c_{k}\geq1$ (in particular $c_{k}\notin\{-1,0\}$). We assume
first that $c_{k}>1$. Then, by Corollary~\ref{cor:FirstTraceIds},
$\tc=t_{[0,c_{0},c_{1},\ldots,c_{k}-1,1]}$ and therefore $\sigc=\sigma_{[0,c_{0},c_{1},\ldots,c_{k}-1,1]}$.
Applying the first part of the lemma on $\widetilde{\co}=[0,c_{0},c_{1},\ldots,c_{k}-1,1]$
gives
\begin{align*}
\sigc=\sigma_{[0,c_{0},c_{1},\ldots,c_{k}-1,1]} & \subseteq\sigma_{[0,c_{0},c_{1},\ldots,c_{k-1}]}\cup\sigma_{[0,c_{0},c_{1},\ldots,c_{k}-1]},
\end{align*}
and this yields the second part of the lemma since $t_{[\co,0]}=t_{[0,c_{0},c_{1},\ldots,c_{k-1}]}$
and $t_{[\co,-1]}=t_{[0,c_{0},c_{1},\ldots,c_{k}-1]}$ by Corollary~\ref{cor:FirstTraceIds}
and Proposition~\ref{prop:=000020Floquet-Bloch=000020via=000020transfer=000020matrix}.
To complete the proof assume that $c_{k}=1$. In this case $t_{[\co,0]}=t_{[0,c_{0},c_{1},\ldots,c_{k-1}]}$
and $t_{[\co,-1]}=t_{[0,c_{0},c_{1},\ldots,c_{k-2}]}$ (the latter
is by applying twice Corollary~\ref{cor:FirstTraceIds}) and once
again the second part of the lemma follows from the first.
\end{proof}
We end by connecting the spectrum of an aperiodic Sturmian Hamiltonian,
$\Ham$, with $\alpha\notin\Q$ with the spectra of periodic operators
which approximate it. To do so, we apply the following result from
\cite{BIST89}.
\begin{prop}
[{\cite{BIST89}}] \label{prop:SpecBddTrace} Let $\alpha\in[0,1]\setminus\Q$
with infinite continued fraction expansion $\left(c_{i}\right)_{i=0}^{\infty}$.
Then
\[
\sigma(\Ham)=\set{E\in\R}{\{t_{[0,c_{0},c_{1},\ldots,c_{k}]}(E)\}_{k\in\N}\text{ is a bounded sequence }}.
\]
\end{prop}

\begin{proof}
This is proven in \cite{BIST89}.
\end{proof}
\begin{cor}
\label{cor:SpectrumInTwoSubsequentSpectra} Let $\alpha\in(0,1)\setminus\Q$
with infinite continued fraction expansion $\left(c_{i}\right)_{i=0}^{\infty}$.
Then we get for all $k\in\N$ 
\[
\sigma(\Ham)\subseteq\sigma_{[0,0,c_{1},\ldots,c_{k}]}(V)\cup\sigma_{[0,0,c_{1},\ldots,c_{k+1}]}(V).
\]
\begin{proof}
Let $E\in\sigma(\Ham)$ and assume by contradiction that there is
some $k\in\N$ such that $E\notin\sigma_{[0,0,c_{1},\ldots,c_{k}]}(V)\cup\sigma_{[0,0,c_{1},\ldots,c_{k+1}]}(V)$.
By Proposition~\ref{prop:=000020Floquet-Bloch=000020via=000020transfer=000020matrix},
$\left|t_{[0,0,c_{1},\ldots,c_{k}]}(E)\right|>2$ and $\left|t_{[0,0,c_{1},\ldots,c_{k+1}]}(E)\right|>2$.
Applying Lemma~\ref{lem:TraceStaysBig} we get that there exists
$C>1$ such that 
\[
|t_{[0,0,c_{1},\ldots,c_{n}]}(E)|>C^{q_{n}}\quad\text{for all }n\in\N.
\]
In particular $\{t_{[0,0,c_{1},\ldots,c_{n}]}(E)\}_{k\in\N}$ is an
unbounded sequence, but this contradicts $E\in\sigma(\Ham)$ by Proposition~\ref{prop:SpecBddTrace}.
\end{proof}
\end{cor}

Both Lemma~\ref{lem:MonotonicityProp} and Corollary~\ref{cor:SpectrumInTwoSubsequentSpectra}
provide monotonicity statements of the spectra. In addition to those,
we also have the following spectral convergence result.
\begin{prop}
\label{prop:=000020Convergence=000020to=000020aperiodic=000020spectrum}

Let $V\in\R$ and $\alpha\notin\Q$ with infinite continued fraction
expansion $\left(c_{i}\right)_{i=0}^{\infty}$. For $k\in\N_{0}$,
set $\cok=[0,0,c_{1},\ldots,c_{k}]\in\Co$. Then
\[
\sigma\left(\Ham\right)=\lim_{k\rightarrow\infty}\sigma_{[\cok,1]}(V)=\lim_{k\rightarrow\infty}\sigma_{\cok}(V)=\bigcap_{k\in\N_{0}}\left(\sigma_{\cok}(V)\cup\sigma_{[\cok,1]}(V)\right).
\]
\end{prop}

\begin{proof}
By \cite[thm. 1]{BIT91}, the spectral map $[0,1]\ni\beta\mapsto\sigma\left(H_{\beta,V}\right)$
is continuous at all irrational $\beta\in[0,1]$ and for all $V\in\R$.
Observe that $\lim_{k\to\infty}\varphi(\co_{k})=\lim_{k\to\infty}\varphi([\co_{k},1])=\alpha$
where $\varphi$ is the evaluation map. Thus, 
\[
\sigma\left(\Ham\right)=\lim_{k\rightarrow\infty}\sigma_{[\cok,1]}(V)=\lim_{k\rightarrow\infty}\sigma_{\cok}(V)
\]
follows using $\alpha\not\in\Q$ and $\sigma_{\co}(V)=\sigma\left(H_{\varphi(\co),V}\right)$
for $\co\in\Co$ proven in Proposition~\ref{prop:=000020Floquet-Bloch=000020via=000020transfer=000020matrix}.

Set $\Lambda_{k}(V):=\sigma_{\cok}(V)\cup\sigma_{[\cok,1]}(V)$ for
$k\in\N_{0}$ and $V\in\R$. Lemma~\ref{lem:MonotonicityProp} and
Corollary~\ref{cor:SpectrumInTwoSubsequentSpectra} imply $\sigma\left(\Ham\right)\subseteq\Lambda_{k+1}(V)\subseteq\Lambda_{k}(V)$.
Thus, $\sigma\left(\Ham\right)\subseteq\bigcap_{k\in\N_{0}}\Lambda_{k}(V)$
follows. By the convergence of $\sigma_{[\cok,1]}(V)$ and $\sigma_{\cok}(V)$,
we conclude that $\left\{ \Lambda_{k}(V)\right\} _{k\in\N_{0}}$ converge
monotonically in the Hausdorff metric to $\sigma\left(\Ham\right)$.
Thus, if $E\notin\sigma\left(\Ham\right)$, then there is an $\varepsilon>0$
such that $B_{\varepsilon}(E):=\set{E'\in\R}{|E-E'|<\varepsilon}$
does not intersect $\sigma\left(\Ham\right)$. Then the Hausdorff
convergence of $\left\{ \Lambda_{k}(V)\right\} _{k\in\N_{0}}$ to
$\sigma\left(\Ham\right)$ implies that there is a $k_{0}\in\N_{0}$
such that $B_{\varepsilon/2}(E)\cap\Lambda_{k}(V)=\emptyset$ for
all $k\geq k_{0}$. Hence, $E\not\in\bigcap_{k\in\N_{0}}\Lambda_{k}(V)$
is derived proving $\bigcap_{k\in\N_{0}}\Lambda_{k}(V)=\sigma\left(\Ham\right)$.
\end{proof}

\subsection{Spectral bands structure for large coupling constant, \textmd{$V>4$}
\protect\label{subsec:=000020spectral=000020properties=000020V>4}}

From this point on until the end of the paper, we specialize our discussion
for the case $V>4$. Under this assumption, one can prove quite a
few useful connections between the periodic spectra, $\{\sigc\}_{\co\in\Co}.$

We start with the three-intersection-property. This observation can
essentially be found in \cite{Casdagli1986} for the Fibonacci Hamiltonian
and was generalized in \cite[prop. 3.1,(iii)]{Raym95}. This property
starts failing if $|V|\leq4$ and this is one major obstacle to treat
the small coupling regime.
\begin{prop}
\cite{Casdagli1986,Raym95}\label{prop:ThreeConsecSpecCantIntersect}
Let $V>4$, $\co\in\Co$ and $m\in\Nz$ such that $[\co,m-1]\in\Co$.
Then 
\[
\sigma_{\co}(V)\cap\sigma_{[\co,m]}(V)\cap\sigma_{[\co,m-1]}(V)=\emptyset.
\]
\end{prop}

\begin{proof}
Assume by contradiction that there is some $E\in\sigma_{\co}(V)\cap\sigma_{[\co,m]}(V)\cap\sigma_{[\co,m-1]}(V).$
By Proposition~\ref{prop:=000020Floquet-Bloch=000020via=000020transfer=000020matrix},
we obtain 
\[
|\tc(E)|,|\tcm(E)|,|t_{[\co,m,-1]}(E)|\leq2.
\]
Substituting this in Proposition~\ref{prop:=000020Fricke-Vogt},
$V>4$ yields 
\[
20\geq\tc^{2}(E)+\tcm^{2}(E)+t_{[\co,m-1]}^{2}(E)-\tc(E)\tcm(E)t_{[\co,m-1]}(E)=4+V^{2}>20,
\]
a contradiction.
\end{proof}
\begin{cor}
\label{cor:=000020monotonicity=000020of=000020single=000020E} Let
$V>4$, and $\co\in\Co$ such that $[\co,-1]\in\Co$. If $E\in\sigc(V)$,
then either
\[
E\in\sigcz(V)\quad\text{or}\quad E\in\sigma_{[\co,-1]}(V),
\]
but not both.
\end{cor}

\begin{proof}
Assume $E\in\sigc(V)$. By Lemma~\ref{lem:MonotonicityProp}, we
get $E\in\sigcz(V)\cup\sigma_{[\co,-1]}(V)$. Now, apply Proposition~\ref{prop:ThreeConsecSpecCantIntersect}
with $m=0$ and get that either $E\in\sigcz(V)$ or $E\in\sigma_{[\co,-1]}(V),$
but not both.
\end{proof}
\begin{prop}
\label{prop:BandsInclusion} Let $V>4$, and $\co\in\Co$ such that
$[\co,-1]\in\Co$. If $I\subseteq\sigc(V)$ is a spectral band, then
$I$ is either contained in a spectral band of $\sigcz(V)$ or in
a spectral band of $\sigma_{[\co,-1]}(V)$, but not in both.
\end{prop}

\begin{proof}
Since $I$ is a spectral band, we conclude that $I$ is closed and
connected. Now both $I\cap\sigcz(V)$ and $I\cap\sigma_{[\co,-1]}(V)$
are closed too and according to Corollary~\ref{cor:=000020monotonicity=000020of=000020single=000020E},
we have the following disjoint union $I=\left(I\cap\sigma_{[\co,0]}(V)\right)\sqcup\left(I\cap\sigma_{[\co,-1]}(V)\right)$.
Since $I$ is connected, one of the closed sets $I\cap\sigma_{[\co,0]}(V)$
and $I\cap\sigma_{[\co,-1]}(V)$ must be empty and the other equals
to $I$. Hence, $I$ is contained in either $\sigma_{[\co,0]}(V)$
or $\sigma_{[\co,-1]}(V)$. Using the same argument we may conclude
that $I$ is contained in a single connected component (spectral band)
of $\sigma_{[\co,0]}(V)$ or $\sigma_{[\co,-1]}(V)$.
\end{proof}
Proposition~\ref{prop:BandsInclusion} motivates a classification
of the spectral bands into two types. We start employing such a dichotomy
of the spectral bands and see that it leads to a hierarchical structure
of the spectral bands from different spectra, $\sigc$. This structure
is developed and described in detail in the rest of this section.

Let $I=[a,b]$ and $J=[c,d]$ be two closed intervals. We say that
$I$ is \emph{strictly included} in $J$ and denote $I\strict J$
if $c<a$ and $b<d$. Note that this implies the (weaker) inclusion
$I\subseteq J$.
\begin{defn}
\noindent\begin{flushleft}
\label{def:=000020backward=000020type}Let $V\in\R$ and $\co\in\Co$
be such that $\varphi(\co)\in[0,1]$ and $[\co,0],[\co,-1]\in\Co$.
A spectral band $I(V)$ of $\sigma_{\co}(V)$ is called
\par\end{flushleft}
\begin{itemize}
\item \begin{flushleft}
\emph{backward type $\tA$} \\
if there exists a spectral band $J(V)$ in $\sigma_{[\co,0]}(V)$
such that $\mbox{\ensuremath{I(V)\strict J(V)}}$.
\par\end{flushleft}
\item \begin{flushleft}
\emph{weak backward type $\tA$} \\
if there exists a spectral band $J(V)$ in $\sigma_{[\co,0]}(V)$
such that $\mbox{\ensuremath{I(V)\subseteq J(V)}}$.
\par\end{flushleft}
\item \begin{flushleft}
\emph{backward type $\tB$} \\
if there exists a spectral band $J(V)$ in $\sigma_{[\co,-1]}(V)$
such that $\mbox{\ensuremath{I(V)\strict J(V)}}$.
\par\end{flushleft}
\item \begin{flushleft}
\emph{weak backward type $\tB$} \\
if there exists a spectral band $J(V)$ in $\sigma_{[\co,-1]}(V)$
such that $\mbox{\ensuremath{I(V)\subseteq J(V)}}$.
\par\end{flushleft}
\end{itemize}
\end{defn}

\begin{rem*}
In \cite[def. 3.2]{Raym95} a spectral band of weak backward type
$A$ is called a type $III$ band, and a spectral band of weak backward
type $B$ is called a type $II$ band. After that, the notations $A$
and $B$ (for such bands) also appeared in the literature, see e.g.,
\cite{KiKiLa03,Damanik2008,DaGoYe16}. . We prefer to use here (and
also in \cite{BanBecLoe_24}) the notations $A$,$B$ for visual reasons
and to distinguish those from the notation $G$ introduced in the
sequel for spectral gaps. In addition, only the notions of weak backward
types appear in \cite[def. 3.2]{Raym95} (though not in this name).
We introduce here also the stronger notion of (non-weak) backward
types and use them to prove slightly stronger statements, since those
are needed in order to obtain further results for $V<4$ in \cite{BanBecLoe_24}.

In \cite[def. 3.2]{Raym95} also the notion of type $I$ gap is introduced
being a spectral band $I(V)$ in $\sigma_{[\co,1]}(V)$ that is contained
in $\sigma_{[\co,1,-1]}(V)$ (so a weak backward $B$ band). By Proposition~\ref{prop:BandsInclusion},
$I(V)\cap\sigc(V)=\emptyset$ and so $I(V)$ is contained in a spectral
gap of $\sigc(V)$. As mentioned before, we omit this terminology
here but when coding the spectrum in Section~\ref{subsec:=000020Coding}
the label $G$ is rather used. These bands are a placeholder for the
corresponding $B$ band one level higher, confer Definition~\ref{def:CoveringSets}.
\end{rem*}
Using Definition~\ref{def:=000020backward=000020type} and Proposition~\ref{prop:BandsInclusion}
we conclude the following.
\begin{cor}
\label{cor:=000020Weak=000020Backward=000020Type} For all $V>4$
and $\co\in\Co$ with $\varphi(\co)\in[0,1]$, every spectral band
in $\sigma_{\co}(V)$ is either of weak backward type $A$ or weak
backward type $B$, but not both.
\end{cor}

\begin{proof}
This is just a reformulation of Proposition~\ref{prop:BandsInclusion}.
\end{proof}
We note that according to Definition~\ref{def:=000020backward=000020type},
whether a spectral band is a (weak) backward type $\tA$ or $\tB$
(or not at all) depends on the value of $V$. We see later (Theorem~\ref{Thm-V>4_AllTypeA_B})
that as long as $V>4$, the type of a spectral band does not depend
on the value of $V$. This statement is generalized in \cite[thm. 2.15]{BanBecLoe_24}
for all $V\neq0$. Note that there is no use to consider the backward
type properties for $V=0$, as in this case all spectra of all operators
$H_{\alpha,V}$ are equal to $[-2,2]$.

If $\co\in\Co$ with $\varphi(\co)\in(0,1)$ and $[\co,0],[\co,-1]\in\Co$,
then there are $c_{1},c_{2},\ldots,c_{k}\in\N$ for some $k\in\N$
such that either $\co=[0,0,c_{1},\ldots,c_{k}+1]$ or $\co=[0,0,c_{1},\ldots,c_{k},1]$.
Indeed, the rational number $\varphi(\co)$ has exactly two different
continued fraction expansion \cite[ch. I.4]{Khinchin_book64}. Then
the weak backward type of a spectral band in $\sigc$ depends on the
chosen representation. More precisely, a straightforward computation
using Corollary~\ref{cor:FirstTraceIds} (see details in \cite[prop. 2.10]{BanBecLoe_24})
yields
\begin{itemize}
\item $I(V)$ is of weak backward type $A$ in $\sigma_{[0,0,c_{1},\ldots,c_{k}+1]}(V)$
if and only if \\
$I(V)$ is of weak backward type $B$ in $\sigma_{[0,0,c_{1},\ldots,c_{k},1]}(V)$
and
\item $I(V)$ is of weak backward type $B$ in $\sigma_{[0,0,c_{1},\ldots,c_{k}+1]}(V)$
if and only if \\
$I(V)$ is of weak backward type $A$ in $\sigma_{[0,0,c_{1},\ldots,c_{k},1]}(V)$.
\end{itemize}
We note that this duality does not show up in \cite{Raym95}. There
one considers a fixed $\alpha\in[0,1]\setminus\Q$ with a fixed infinite
continued fraction expansion $(\tilde{c}_{k})_{k\in\N_{0}}$ and rational
number $\alpha_{k}=\varphi([0,0,\tilde{c}_{1},\ldots,\tilde{c}_{k}])$.
Hence, $\alpha_{k}$ has a unique finite continued fraction expansion.
Here and in \cite{BanBecLoe_24}, we consider all elements of $\Co$
and this is why this duality is evident.

We demonstrate the classification of spectral bands to weak backward
types, by explicitly computing a few spectral bands in the following.
\begin{example}
\label{exa:=000020backward=000020type=000020=00005B0,0=00005D=000020and=000020=00005B0,0,1=00005D}
According to Example~\ref{exa:Spectrum_=00005B0=00005D=000020and=000020=00005B0,0,-1=00005D}
we have for $V\in\R$,
\[
\sigma_{[0]}(V)=\sigma_{[0,0,0]}(V)=\sigma_{[0,0,1,-1]}(V)=\R\quad\textrm{and}\quad\sigma_{[0,0,-1]}(V)=[-2-V,2-V].
\]

We examine $\sigma_{[0,0]}(V)=[-2,2]$ and wish to determine its backward
type. To do so we need to examine $\sigma_{[0,0,0]}(V)=\R$ and $\sigma_{[0,0,-1]}(V)=[-2-V,2-V]$.
By Definition~\ref{def:=000020backward=000020type}, we see that
for all $V\neq0$ the spectral band $I_{[0,0]}(V):=[-2,2]$ is of
backward type $A$ but not of weak backward type $B$.

A few additional spectra are: 
\[
\sigma_{[0,0,1]}(V)=[-2+V,2+V]\quad\textrm{and}\quad\sigma_{[0,0,1,0]}(V)=\sigma_{[0,0]}(V)=[-2,2].
\]
Given these spectra, one sees that for all $V\neq0$, the spectral
band $I_{[0,0,1]}(V):=[-2+V,2+V]$ of $\sigma_{[0,0,1]}(V)$ is of
backward type $B$ but not of weak backward type $A$.
\end{example}

The spectral bands considered in this example are actually of a well-defined
backward type and not just \emph{weak} backward type. This is stronger
than what is currently proved in Corollary~\ref{cor:=000020Weak=000020Backward=000020Type}.
This stronger version indeed holds in general for all spectral bands
as we prove in in Theorem~\ref{Thm-V>4_AllTypeA_B}.

Next, we extend the classification of spectral bands into types by
adding forward types to the backward type (later we show that they
are actually the same).

Let $I=[a,b]$ and $J=[c,d]$ be two closed intervals. We say that
$I$ is \emph{to the left} of $J$ (or $J$ is \emph{to the right}
of $I$) and denote $I\prec J$ if $a<c$ and $b<d$. Moreover, we
say$I$ is \emph{strictly to the left} of $J$ (or $J$ is \emph{strictly
to the right }of $I$) and denote $I\precd J$ if $b<c$. Observe
that $I\precd J$ holds if and only if $I\prec J$ and $I\cap J=\emptyset$.
\begin{defn}
\label{def:ForwardType} Let $V\in\R\backslash\{0\}$. Let $\co\in\Co$
and $m\in\N$ be such that$[\co,m]\in\Co$. A spectral band $I_{\co}(V)$
of $\sigma_{\co}(V)$ is called of \emph{$m$-forward type $\tA$
}with $M=m-1$ (respectively\emph{ $m$-forward type $\tB$} with
$M=m$) if the following holds.
\begin{enumerate}[align=left, leftmargin={*}, label=(\Alph*)]
\item  \label{enu:=000020A=000020property} There exist $M$ spectral bands
of $\sigma_{[\co,m]}(V)$ (denoted $I_{[\co,m]}^{1}(V),\ldots,I_{[\co,m]}^{M}(V)$)
which satisfy
\begin{enumerate}[label=(A\arabic*)]
\item  \label{enu:A1Property} $I_{[\co,m]}^{i}(V)\strict I_{\co}(V)$
for all $1\leq i\leq M$.\\
In particular, these bands are of backward type $A$.\\
~
\item \label{enu:A2Property} $I_{[\co,m]}^{i}(V)$ is not of weak backward
type $B$ for all $1\leq i\leq M$.\\
~
\end{enumerate}
\item \label{enu:=000020B=000020property} For each $n\in\N$, there exist
$M+1$ spectral bands of $\sigma_{[\co,m,n]}(V)$\\
(denoted $I_{[\co,m,n]}^{1}(V),\ldots,I_{[\co,m,n]}^{M+1}(V)$) which
satisfy
\begin{enumerate}[label=(B\arabic*)]
\item  \label{enu:B1Property}$I_{[\co,m,n]}^{j}(V)\strict I_{[\co,m,n-1]}^{j}(V)$
for all $1\leq j\leq M+1$, where $I_{[\co,m,0]}^{j}(V):=I_{\co}(V)$.\\
In particular, these bands are of backward type $B$.\\
~
\item \label{enu:B2Property} $I_{[\co,m,n]}^{j}(V)$ is not of weak backward
type $A$ for all $1\leq j\leq M+1$.
\end{enumerate}
\end{enumerate}
~
\begin{enumerate}[align=left, leftmargin={*}, label=(I)]
\item \label{enu:IProperty} For each $n\in\N$, we have 
\[
I_{[\co,m,n]}^{1}(V)\prec I_{[\co,m]}^{1}(V)\prec I_{[\co,m,n]}^{2}(V)\prec I_{[\co,m]}^{2}(V)\ldots\prec I_{[\co,m]}^{M}(V)\prec I_{[\co,m,n]}^{M+1}(V).
\]
\end{enumerate}
We say $\Ic(V)$ satisfies the stronger interlacing property if \ref{enu:IProperty}
is replaced by
\begin{equation}
\tag{\ensuremath{\Istr}}I_{[\co,m,n]}^{1}(V)\precd I_{[\co,m]}^{1}(V)\precd I_{[\co,m,n]}^{2}(V)\precd\ldots\precd I_{[\co,m]}^{M}(V)\precd I_{[\co,m,n]}^{M+1}(V).\label{eq:Interlacing_strict}
\end{equation}
\end{defn}

\begin{rem*}
Definition~\ref{def:ForwardType} rephrases the content of \cite[lem. 3.3]{Raym95}.
A few notes should be made about the similarities and differences
of both. First, as is commonly done in this review, we use the notation
$\co\in\Co$ rather than $(k,p)$ as in \cite{Raym95}. Second, we
state the lemma from \cite{Raym95} as a definition here, since in
\cite{BanBecLoe_24} we need to keep the separation between backward
types and forwards type for the sake of some of the proofs (even if
at the end we realize that both concepts are equivalent). Third, Definition~\ref{def:ForwardType}
introduces a slightly stronger notion of forward type than the one
which appears explicitly\footnote{In the original paper \cite{Raym95} the strict inclusion and strict
order were implicitly assumed without an explicit proof.} in \cite[lem. 3.3]{Raym95}; the strengthening is by using everywhere
the strict inclusion $\strict$ rather than $\subset$ and also by
having in \ref{enu:B1Property} $I_{[\co,m,n]}^{j}(V)\strict I_{[\co,m,n-1]}^{j}(V)$
for all $n\in\N$ rather than just $I_{[\co,m,n]}^{j}(V)\strict I_{\co}^{j}(V)$.
This strengthening is crucial in \cite{BanBecLoe_24}, and that is
why we choose to deviate from the original exposition in \cite[lem. 3.3]{Raym95}.
\end{rem*}
Our next task is to show that indeed each spectral band has a well-defined
forward type as in Definition~\ref{def:ForwardType}. Actually, we
will see that if a spectral band is of weak backward type $A$ (respectively
$B$) then it is also of $m$-forward type $A$ (respectively $B$)
for all $m\in\N$. This will be stated in Proposition~\ref{prop:=000020backward=000020implies=000020forward}.
But before doing so, we need to prove two preparatory lemmas (Lemma~\ref{lem:=000020backward=000020implies=000020forward=000020-=000020partial}
and Lemma~\ref{lem:=000020Tower=000020property}).

\begin{lem}
\label{lem:=000020backward=000020implies=000020forward=000020-=000020partial}
Let $V>4$. Let $\co\in\Co$ and $m\in\N$ be such that $[\co,m]\in\Co$.
Let $\Ic(V)$ be a spectral band in $\sigma_{\co}(V)$ of weak backward
type $A$ with $M=m-1$ (respectively\emph{ of weak backward type
$\tB$} with $M=m$). Then the following holds (compare with Definition~\ref{def:ForwardType}):
\begin{enumerate}
\item \label{enu:=000020backward=000020implies=000020forward=000020-=000020partial=000020_=000020A=000020bands}There
exist \emph{exactly} $M$ spectral bands of $\sigma_{[\co,m]}(V)$
(denoted $I_{[\co,m]}^{1}(V),\ldots,I_{[\co,m]}^{M}(V)$) which are
contained in $\Ic(V)$.\\
These spectral bands satisfy properties \ref{enu:A1Property} and
\ref{enu:A2Property} from Definition~\ref{def:ForwardType}.
\item \label{enu:=000020backward=000020implies=000020forward=000020-=000020partial=000020_=000020B=000020bands}There
exist exactly $M+1$ spectral bands of $\sigma_{[\co,m,1]}(V)$ (denoted
$I_{[\co,m,1]}^{1}(V),\ldots,$ $I_{[\co,m,1]}^{M+1}(V)$) which are
contained in $\Ic(V)$. These spectral bands satisfy:
\begin{enumerate}
\item $I_{[\co,m,1]}^{j}(V)\subset\Ic(V)$, for all $1\leq j\leq M+1$.
In particular, these bands are of weak backward type $B$.
\item $I_{[\co,m,1]}^{j}(V)$ is not of weak backward type $A$ for all
$1\leq j\leq M+1$.
\end{enumerate}
\item \label{enu:=000020backward=000020implies=000020forward=000020-=000020partial=000020_=000020interlacing}The
following interlacing property holds: 
\[
I_{[\co,m,1]}^{1}(V)\precd I_{[\co,m]}^{1}(V)\precd I_{[\co,m,1]}^{2}(V)\precd I_{[\co,m]}^{2}(V)\ldots\precd I_{[\co,m]}^{M}(V)\precd I_{[\co,m,1]}^{M+1}(V).
\]
\end{enumerate}
\end{lem}

\begin{rem*}
We can colloquially phrase Lemma~\ref{lem:=000020backward=000020implies=000020forward=000020-=000020partial}
as follows: if the spectral band $\Ic(V)$ is of a weak backward type
$A$ (or $B$), then for all $m\in\N$ it is ``partially'' $m$-forward
type $A$ or $B$, correspondingly. By ``partially'' we mean that
$\Ic(V)$ fully satisfies properties \ref{enu:A1Property} and \ref{enu:A2Property}
(in Definition~\ref{def:ForwardType}), but it satisfies properties
\ref{enu:B1Property}, \ref{enu:B2Property} and the strong interlacing
(\ref{eq:Interlacing_strict}) only for $n=1$ and property \ref{enu:B1Property}
is satisfied only in its weak version, i.e., that all $I_{[\co,m,1]}^{j}$
are of \emph{weak} backward type $B$. Another difference between
Lemma~\ref{lem:=000020backward=000020implies=000020forward=000020-=000020partial}
and Definition~\ref{def:ForwardType} ($m$-forward type) goes in
the other direction: in this lemma we state that the spectral bands
$\left\{ I_{[\co,m]}^{i}\right\} _{i=1}^{M}$ and $\left\{ \Icmnj\right\} _{j=1}^{M+1}$
are unique, which is not part of Definition~\ref{def:ForwardType}.
\end{rem*}
\begin{proof}
First, we fix the value of $V>4$ throughout the proof, but for brevity
we omit $V$ from the various notations (for example, writing just
$I$ and $\sigc$). We fix the following auxiliary variable,
\[
\db:=\begin{cases}
0,\quad & I~\textrm{is of backward type }A,\\
1,\quad & I~\textrm{is of backward type }B,
\end{cases}
\]
which allows us to prove the lemma simultaneously for both these cases.
Note that with this notation $M=m-1+\db$, for the value $M$ which
is introduced in the statement.

We introduce two other notations which will help throughout the proof.
Given $\co\in\Co$ and a spectral band $\Ic\subseteq\sigma_{\co}$,
we know by Proposition~\ref{prop:=000020Floquet-Bloch=000020via=000020transfer=000020matrix}
that $\tc(\Ic)=\left[-2,2\right]$ and $\left.\tc\right|_{\Ic}$ is
strictly monotone. Hence, for each $x\in\left[-2,2\right]$ we may
denote by $E_{\co}^{\Ic}(x)$ the unique value in $\Ic$ such that
$\tc(E_{\co}^{\Ic}(x))=x$.

The proof consists of four steps, which we briefly summarize before
going into the details. To obtain the candidates for the spectral
bands $\Icmi$ in property \ref{enu:=000020A=000020property} and
the spectral bands $I_{[\co,m,1]}^{j}$ in property \ref{enu:=000020B=000020property},
we indicate specific energy values$\left\{ A_{i}\right\} _{i=1}^{M}$
and $\left\{ B_{i}\right\} _{i=1}^{M+1}$ in $I$ and find the spectral
bands of $\sigcm$ and $\sigma_{[\co,m,1]}$ which contain these values.
This forms the first two steps of the proof. The third step would
be to prove property \ref{enu:IProperty} by observing the order between
the aforementioned energy values $\left\{ A_{i}\right\} _{i=1}^{M}$
and $\left\{ B_{i}\right\} _{i=1}^{M+1}$. The last step is to show
that there are no other spectral bands in $\sigma_{[c,m]}$ and $\sigma_{[c,m,1]}$
which satisfy those properties, implying the uniqueness which is mentioned
in \ref{enu:=000020A=000020property} and \ref{enu:=000020B=000020property}.

\textbf{Step 1: Defining the spectral bands $\{I_{[\co,m]}^{i}\}_{i=1}^{M}$}
and proving \ref{enu:A2Property} and partially \ref{enu:A1Property}:

Define $A_{i}:=E_{\co}^{\Ic}(2\cos(\frac{i\pi}{m+\db}))$ for $i=1,\dots,m-1+\db=M$
satisfying $\tc(A_{i})=2\cos(\frac{i\pi}{m+\db})$. We use these values
to define spectral bands in $\sigma_{[\co,m]}$, and show later that
those are exactly\textbf{ $\{I_{[\co,m]}^{i}\}_{i=1}^{M}$} from Definition~\ref{def:ForwardType},\ref{enu:=000020A=000020property}.

Corollary~\ref{cor:ExtChebyForm1} (applied for $\ell=-\db$) implies

\begin{align*}
t_{[\co,m]}(A_{i}) & =\CP_{m-1+\db}(\tc(A_{i}))t_{[\co,1-\db]}(A_{i})-\CP_{m-2+\db}(\tc(A_{i}))t_{[\co,-\db]}(A_{i}).
\end{align*}

The dilated Chebyshev polynomials satisfy $S_{l}(2\cos\theta)=\frac{\sin(l+1)\theta}{\sin\theta}$,
see Lemma~\ref{lem:=000020Cheby=000020Poly=000020-=000020explicit=000020expression}.
Using this and $\tc(A_{i})=2\cos(\frac{i\pi}{m+\db})$, we evaluate
the dilated Chebyshev polynomials which appear in the last equation:
\begin{align*}
\CP_{m-1+\db}(\tc(A_{i})) & =\CP_{m-1+\db}\left(2\cos\left(\frac{i\pi}{m+\db}\right)\right)=\frac{\sin\left((m+\db)\frac{i\pi}{m+\db}\right)}{\sin\left(\frac{i\pi}{m+\db}\right)}=0,
\end{align*}
and
\begin{align*}
\CP_{m-2+\db}(\tc(A_{i})) & =\CP_{m-2+\db}\left(2\cos\left(\frac{i\pi}{m+\db}\right)\right)\\
 & =\frac{\sin\left((m-1+\db)\frac{i\pi}{m+\db}\right)}{\sin\left(\frac{i\pi}{m+\db}\right)}\\
 & =\frac{\sin\left(i\pi-\frac{i\pi}{m+\db}\right)}{\sin\left(\frac{i\pi}{m+\db}\right)}\\
 & =\frac{\sin(i\pi)\cos\left(\frac{i\pi}{m+\db}\right)-\cos(i\pi)\sin\left(\frac{i\pi}{m+\db}\right)}{\sin\left(\frac{i\pi}{m+\db}\right)}\\
 & =-\cos(i\pi)=(-1)^{i+1}.
\end{align*}

Combining the computations above gives 
\[
t_{[\co,m]}(A_{i})=(-1)^{i}t_{[\co,-\db]}(A_{i}).
\]

Since $\Ic$ has a well-defined weak backward type (either $A$ or
$B$), by the definition of $\db$ we get that $\big|\left.t_{[\co,-\db]}\right|_{\Ic}\big|\leq2$.
Since $A_{i}\in\Ic$, the equation above implies $|t_{[\co,m]}(A_{i})|\leq2.$
This means that $\left\{ A_{i}\right\} _{i=1}^{M}\subset\sigma_{[\co,m]}$.
Hence, for each $1\leq i\leq M$ we may denote by $I_{[\co,m]}^{i}$
the spectral band in $\sigma_{[\co,m]}$ which contains $A_{i}$.
At this point, we note that it could be that different $A_{i_{0}}\neq A_{i_{1}}$
give rise to the same spectral bands $I_{[\co,m]}^{i_{0}}$ and $I_{[\co,m]}^{i_{1}}$.
However, we prove below in step 3 that $A_{i_{0}}\neq A_{i_{1}}$
implies $I_{[\co,m]}^{i_{0}}\neq I_{[\co,m]}^{i_{1}}$.

We show now that for each i the spectral band $I_{[\co,m]}^{i}$ is
of weak backward type $A$ and not of weak backward type $B$. By
Corollary~\ref{cor:=000020monotonicity=000020of=000020single=000020E},
we have that either $A_{i}\in\sigma_{[\co,m,0]}=\sigc$ or $A_{i}\in\sigma_{[\co,m,-1]}$.
Since $A_{i}\in\Ic\subset\sigc$ we get that $A_{i}\notin\sigma_{[\co,m,-1]}$.
By Proposition~\ref{prop:BandsInclusion}, we have that $\Icmi$
is either contained in a spectral band of $\sigma_{[\co,m,0]}$ or
of $\sigma_{[\co,m,-1]}$. But, since $A_{i}\notin\sigma_{[\co,m,-1]}$
the former option holds and we get that $\Icmi$ is contained in $\Ic$.
In particular $\Icmi$ is of weak backward type $A$ and not of weak
backward type $B$. This shows property \ref{enu:A2Property}, but
it does not yet show property \ref{enu:A1Property} since we only
proved that $\Icmi$ is of weak backward type $A$. We will complete
the proof of property \ref{enu:A1Property} in step 3, where we also
prove that $A_{i_{0}}\neq A_{i_{1}}$ implies $I_{[\co,m]}^{i_{0}}\neq I_{[\co,m]}^{i_{1}}$
and that they satisfy \ref{enu:A1Property}.

\textbf{Step 2: Defining the spectral bands $\{I_{[\co,m,1]}^{j}\}_{j=1}^{M+1}$:}

We proceed similar as in step 1. Define $B_{j}:=E_{\co}^{\Ic}(2\cos(\frac{j\pi}{m+\db+1}))$
for $j=1,\dots,m+\db=M+1$ satisfying $\tc(B_{j})=2\cos(\frac{j\pi}{m+\db+1})$.
Similarly to the previous step in the proof, we will now use these
values to define spectral bands in $\sigma_{[\co,m,1]}$. Corollary~\ref{cor:FirstTraceIds}
and Corollary~\ref{cor:ExtChebyForm1} (applied for $\ell=-\db$)
lead to 
\begin{align*}
t_{[\co,m,1]}(B_{j}) & =t_{[\co,m+1]}(B_{j})\\
 & =\CP_{m+\db}(\tc(B_{j}))t_{[\co,1-\db]}(B_{j})-\CP_{m+\db-1}(\tc(B_{j}))t_{[\co,-\db]}(B_{j}).
\end{align*}

The dilated Chebyshev polynomials satisfy $S_{l}(2\cos\theta)=\frac{\sin(l+1)\theta}{\sin\theta}$,
see Lemma~\ref{lem:=000020Cheby=000020Poly=000020-=000020explicit=000020expression}.
Using this and $\tc(B_{j})=2\cos(\frac{j\pi}{m+\db+1})$, we evaluate
the dilated Chebyshev polynomials which appear in the last equation:

\begin{align*}
\CP_{m+\db}(\tc(B_{j})) & =\CP_{m+\db}\left(2\cos\left(\frac{j\pi}{m+\db+1}\right)\right)=\frac{\sin\left(\left(m+\db+1\right)\frac{j\pi}{m+\db+1}\right)}{\sin\left(\frac{j\pi}{m+\db+1}\right)}=0,
\end{align*}

and

\begin{align*}
\CP_{m+\db-1}(\tc(B_{j})) & =\CP_{m+\db-1}\left(2\cos\left(\frac{j\pi}{m+\db+1}\right)\right)\\
 & =\frac{\sin\left(\left(m+\db\right)\frac{j\pi}{m+\db+1}\right)}{\sin\left(\frac{j\pi}{m+\db+1}\right)}\\
 & =\frac{\sin\left(j\pi-\frac{j\pi}{m+\db+1}\right)}{\sin\left(\frac{j\pi}{m+\db+1}\right)}\\
 & =\frac{\sin(j\pi)\cos\left(\frac{j\pi}{m+\db+1}\right)-\cos(j\pi)\sin\left(\frac{j\pi}{m+\db+1}\right)}{\sin\left(\frac{j\pi}{m+\db+1}\right)}\\
 & =-\cos(j\pi)=(-1)^{j+1}.
\end{align*}

Combining the computations above gives 
\[
t_{[\co,m,1]}(B_{j})=(-1)^{j}t_{[\co,-\db]}(B_{j}).
\]

Since $\Ic$ has a well-defined weak backward type (either $A$ or
$B$), by the definition of $\db$ we get that $\big|\left.t_{[\co,-\db]}\right|_{\Ic}\big|\leq2$.
Since $B_{j}\in\Ic$, the equation above implies $|t_{[\co,m,1]}(B_{j})|\leq2.$
This means that $\left\{ B_{j}\right\} _{j=1}^{M+1}\subset\sigma_{[\co,m,1]}$.
Hence, for each $1\leq j\leq M+1$ we may denote by $I_{[\co,m,1]}^{j}$
the spectral band in $\sigma_{[\co,m,1]}$ which contains $B_{j}$.
Now, similarly to the argument in step 1, we deduce that each $I_{[\co,m,1]}^{j}$
is of weak backward type $B$ and not of weak backward type $A$.

We note that just as in the previous step, we should still prove that
$I_{[\co,m,1]}^{j_{0}}\neq I_{[\co,m,1]}^{j_{1}}$ if $j_{0}\neq j_{1}$).

\textbf{Step 3: Band interlacing}: As mentioned before, the interlacing
follows by the corresponding interlacing of $\left\{ A_{i}\right\} _{i=1}^{M}$
and $\left\{ B_{j}\right\} _{j=1}^{M+1}$. The interlacing of $\left\{ A_{i}\right\} _{i=1}^{M}$
and $\left\{ B_{j}\right\} _{j=1}^{M+1}$ results from the interlacing
of the zeros of two successive dilated Chebyshev polynomials, as these
belong to a family of orthogonal polynomials. Writing this explicitly,
we note that $0<\frac{j}{M+2}<\frac{j}{M+1}<\frac{j+1}{M+2}<\pi$,
for all $1\leq j\leq M$, so that the sets $\left\{ \frac{i}{M+1}\right\} _{i=1}^{M}$,
$\left\{ \frac{j}{M+2}\right\} _{j=1}^{M+1}$ interlace. The sets
$\left\{ A_{i}\right\} _{i=1}^{M}$ , $\left\{ B_{j}\right\} _{j=1}^{M+1}$
are obtained as a monotone function acting on these sets and hence
also interlace. Indeed, to see this recall that $A_{i}:=E_{\co}^{\Ic}\left(2\cos(\pi\frac{i}{M+1})\right)$,
$B_{j}:=E_{\co}^{\Ic}\left(2\cos(\pi\frac{j}{M+2})\right)$ and note
the strict monotonicity of the cosine on $\left[0,\pi\right]$ and
the strict monotonicity of $\tc$ on $\Ic$ together with $E_{\co}^{\Ic}=\left(\left.\tc\right|_{\Ic}\right)^{-1}$.

Hence, we get that either

\begin{equation}
B_{1}<A_{1}<B_{2}<\dots<A_{M}<B_{M+1},\label{eq:=000020proof=000020-=000020interlacing=000020of=000020points=000020-=0000201}
\end{equation}
or 
\begin{equation}
B_{1}>A_{1}>B_{2}>\dots>A_{M}>B_{M+1}.\label{eq:=000020proof=000020-=000020interlacing=000020of=000020points=000020-=0000202}
\end{equation}

Whether (\ref{eq:=000020proof=000020-=000020interlacing=000020of=000020points=000020-=0000201})
is used or (\ref{eq:=000020proof=000020-=000020interlacing=000020of=000020points=000020-=0000202})
is determined by $\sgn(\left.\tc'\right|_{\Ic})$ which is indeed
constant (see Proposition~\ref{prop:=000020Floquet-Bloch=000020via=000020transfer=000020matrix}).
We now draw a few conclusions from this interlacing. In the previous
two steps we have seen that $\left\{ A_{i}\right\} _{i=1}^{M}\subset\sigc\cap\sigma_{[\co,m]}$
and $\left\{ B_{j}\right\} _{j=1}^{M}\subset\sigc\cap\sigma_{[\co,m,1]}$.
Applying Proposition~\ref{prop:ThreeConsecSpecCantIntersect} with
$\co'=[\co,m,1]$, $m'=0$ yields $\sigc(V)\cap\sigma_{[\co,m]}(V)\cap\sigma_{[\co,m,1]}(V)=\emptyset$
(as always, for $V>4$). We get in particular that for all $i$, $A_{i}\notin\sigma_{[\co,m,1]}$
and all $j$, $B_{j}\notin\sigma_{[\co,m]}$. We use this to observe
that for some $i_{0}\neq i_{1}$, the spectral band $I_{[\co,m]}^{i_{0}}\subset\sigcm$
contains $A_{i_{0}}$, the spectral band $I_{[\co,m]}^{i_{1}}\subset\sigcm$
contains $A_{i_{1}},$ and by the interlacing ((\ref{eq:=000020proof=000020-=000020interlacing=000020of=000020points=000020-=0000201})
and (\ref{eq:=000020proof=000020-=000020interlacing=000020of=000020points=000020-=0000202}))
there is some point $B_{j}\notin\sigma_{[\co,m]}$ between $A_{i_{0}}$
and $A_{i_{1}}$. This means in particular that all the spectral bands
$\left\{ I_{[\co,m]}^{i}\right\} _{i=1}^{M}$, defined in step 1,
are distinct. In exactly the same manner we conclude that all the
spectral bands $\left\{ I_{[\co,m,1]}^{j}\right\} _{j=1}^{M+1}$,
defined in step 2, are distinct.  By Proposition~\ref{prop:ThreeConsecSpecCantIntersect},
$I_{[\co,m]}^{i}\cap I_{[\co,m,1]}^{j}=\emptyset$ holds for all $i,j$
since both are contained in $\sigc(V)$. Thus, the desired strong
interlacing property (\ref{eq:Interlacing_strict}) of Definition~\ref{def:ForwardType}
follows for $n=1$. We should just note that if the interlacing of
the sets $\left\{ A_{i}\right\} _{i=1}^{M}$ , $\left\{ B_{j}\right\} _{j=1}^{M+1}$
is as in (\ref{eq:=000020proof=000020-=000020interlacing=000020of=000020points=000020-=0000202}),
we should reshuffle the indices in order to get the interlacing as
in (\ref{eq:=000020proof=000020-=000020interlacing=000020of=000020points=000020-=0000201}).
Namely, we permute the indices of $\left\{ A_{i}\right\} _{i=1}^{M}$
by $1\leftrightarrow M$, $2\leftrightarrow M-1$, ..., and permute
the indices of $\left\{ B_{j}\right\} _{j=1}^{M+1}$ by $1\leftrightarrow M+1$,
$2\leftrightarrow M$, and so on. Obviously, this affects the indices
of the spectral bands $\left\{ I_{[\co,m]}^{i}\right\} _{i=1}^{M}$
and $\left\{ I_{[\co,m,1]}^{j}\right\} _{j=1}^{M+1}$, and we get
the strong interlacing property of Definition~\ref{def:ForwardType},~(\ref{eq:Interlacing_strict})
for $n=1$.

Using the interlacing property we may also deduce that the spectral
bands $\Icmi$ are of backward type $A$. We already obtained in step
1 that they are of weak backward type $A$. But, thanks to the interlacing
property there is another spectral band to the left and to the right
of each $\Icmi$ which is also included in $\Ic$. Hence, each $\Icmi$
is strictly included in $\Ic$ and it is of backward type $A$. Thus,
$\left\{ I_{[\co,m]}^{i}\right\} _{i=1}^{M}$ satisfy also \ref{enu:A1Property}.

\textbf{Step 4: Uniqueness of the bands}: We show now that the spectral
bands $\left\{ I_{[\co,m]}^{i}\right\} _{i=1}^{M}$ are the only spectral
bands of $\sigcm$ which are contained in $\Ic$ and that $\left\{ I_{[\co,m,1]}^{j}\right\} _{j=1}^{M+1}$
are the only spectral bands of $\sigma_{[\co,m,1]}$ which are contained
in $\Ic$.

\noindent Let $J\subset\sigma_{[\co,m]}$ such that $J\subset I$.
We will show that $A_{i}\in J$ for some $1\leq i\leq M$ and conclude
that $J=I_{[\co,m]}^{i}$, which proves the uniqueness in property
\ref{enu:=000020A=000020property}.

\noindent Due to Corollary~\ref{cor:=000020Weak=000020Backward=000020Type},
$\Ic$ has a well-defined weak backward type (either $A$ or $B$).
By definition of $\db$, we therefore get that $\Ic\subset\sigma_{\co}\cap\spectv{\co,-\db}$.
Hence, also $J\subset\sigma_{\co}\cap\spectv{\co,-\db}$. Then Proposition~\ref{prop:ThreeConsecSpecCantIntersect}
implies $\sigma_{\co}(V)\cap\sigma_{[\co,1-\db]}(V)\cap\spectv{\co,-\db}(V)=\emptyset$
for $V>4$ and so we conclude $J\cap\sigma_{[\co,1-\db]}=\emptyset$.
Thus, Proposition~\ref{prop:=000020Floquet-Bloch=000020via=000020transfer=000020matrix}
leads to $\left|t_{[\co,1-\db]}(E)\right|>2$ for all $E\in J$. Similarly,
we have $\sigma_{\co}(V)\cap\sigma_{[\co,m-1]}(V)\cap\spectv{\co,m}(V)=\emptyset$
for $V>4$ by Proposition~\ref{prop:ThreeConsecSpecCantIntersect}
and so $J\subset\sigma_{\co}\cap\spectv{\co,m}$. Thus, $\abs{\tracev{\co,m-1}\rbr E}>2$
follows for all $E\in J$.

\noindent Since $t_{[\co,1-\db]}$ and $\tracev{\co,m-1}$ are continuous
in $E$, the signs of $t_{[\co,1-\db]}(E)$ and $\tracev{\co,m-1}\rbr E$
are constant for all $E\in J$ by the previous considerations. Thus,
we can choose an $\xi\in\left\{ 0,1\right\} $ such that 
\[
\tracev{\co,m-1}\rbr E+\rbr{-1}^{\xi}t_{[\co,1-\db]}\rbr E\neq0\qquad\text{for all }E\in J.
\]
\\
Since $J\subset\Ic\subset\spectv{\co,-\db}$, we have $\abs{\tracev{\co,-\db}(E)}\leq2$
for all $E\in J$. In addition, $\left|\tracev{\co,m}(E)\right|=2$
if $E$ is the left or right edge of $J$. Note that the sign of $\tcm(E)$
changes if $E$ is the left respectively the right edge of $J$. Therefore,
by the intermediate value theorem there exists $E_{0}\in J$ such
that 
\[
\tracev{\co,m}\left(E_{0}\right)+\rbr{-1}^{\xi}\tracev{\co,-\db}\rbr{E_{0}}=0.
\]
Then Lemma~\ref{lem:ExtChebyForm2} (applied for $\ell=-\delta_{B}$)
gives
\begin{align*}
 & S_{m-1+\db}\rbr{\tc(E_{0})}\sbr{\underbrace{\tracev{\co,m-1}(E_{0})+\rbr{-1}^{\xi}\tracev{\co,1-\db}(E_{0})}_{\neq0\text{ as }E_{0}\in J}}\\
= & \sbr{S_{m-2+\db}\rbr{\tc(E_{0})}+\rbr{-1}^{\xi}}\sbr{\underbrace{\tracev{\co,m}(E_{0})+\rbr{-1}^{\xi}\tracev{\co,-\db}\rbr{E_{0}}}_{=0}}.
\end{align*}
We conclude $S_{m-1+\db}\rbr{\tc(E_{0})}=0$ for $E_{0}\in J\subset I$.

\noindent Since by definition of the dilated Chebyshev polynomails
$S_{0}\equiv1$, we get $m-1+\db\neq0$. Hence $m-1+\db\geq1$ and
we conclude $\abs{\trace{\co}\rbr{E_{0}}}<2$, since dilated Chebyshev
polynomials do not vanish outside $\left(-2,2\right)$, see Lemma~\ref{Lem-Chebyshev-Traces}~(\ref{enu:Chebyshev-Traces_x_geq_2-Sn(x)_geq_1}).
Therefore, there exists some $\theta\in\rbr{0,\pi}$ such that $\trace{\co}\rbr{E_{0}}=2\cos\theta$
and
\[
0=S_{m-1+\db}\rbr{2\cos\theta}=\frac{\sin\rbr{\rbr{m+\db}\theta}}{\sin\theta},
\]
where we used Lemma~\ref{lem:=000020Cheby=000020Poly=000020-=000020explicit=000020expression}
in the last equality. We conclude that $\theta=\frac{i\pi}{m+\db}$
for some $1\leq i\leq m-1+\db$. Therefore, $\tc(E_{0})=2\cos\left(\frac{i\pi}{m+\db}\right)$
or equivalently, $E_{0}=E_{\co}^{I}(2\cos(\frac{i\pi}{m+\db}))$.
But this is exactly the definition of $A_{i}$ in the beginning of
the proof, and so $E_{0}=A_{i}$. Thus, $J=\Icmi$ follows proving
the uniqueness in (\ref{enu:=000020backward=000020implies=000020forward=000020-=000020partial=000020_=000020A=000020bands}).

\medskip{}
In order to show the uniqueness for the spectral bands $\left\{ I_{[\co,m,1]}^{j}\right\} _{j=1}^{M+1}$,
we repeat the arguments above, mainly replacing $m$ with $m+1$ and
using $\sigma_{[\co,m,1]}=\sigma_{[\co,m+1]}$ and $\sigma_{[\co,m,1,-1]}=\sigma_{[\co,m,0]}$
by Lemma~\ref{lem:=000020Trace=000020Depends=000020On=000020Value=000020Only}.
Briefly, if we assume that $J$ is a spectral band of $\spectv{\co,m,1}$
such that $J\subset I$, we are able to conclude that there exists
$1\leq j\leq m+\db$ such that $E_{0}:=2\cos\rbr{\frac{j\pi}{m+1+\db}}\in J$.
Thus, $E_{0}=B_{j}$ for some $1\leq j\leq m+\delta_{B}$ follows
where $\left\{ B_{j}\right\} _{j=1}^{m+\db}$ where defined in step
2. Hence, $J=I_{[\co,m,1]}^{j}$ follows proving the uniqueness in
(\ref{enu:=000020backward=000020implies=000020forward=000020-=000020partial=000020_=000020B=000020bands}).
\end{proof}
In the course of proving Lemma~\ref{lem:=000020backward=000020implies=000020forward=000020-=000020partial}
we have gained some information regarding the location of the spectral
bands $I_{[\co,m]}^{i}$. We state this here as a separate corollary,
since it is useful in a proof which appears in \cite[eq. (7.11)]{BanBecLoe_24}.
\begin{cor}
\label{cor:=000020comes=000020after=000020lemma=000020backward=000020implies=000020forward}Let
$V>4$, $\co\in\Co$ and $m\in\N$ be such that $[\co,m]\in\Co$.
If $\Ic(V)$ is a spectral band in $\sigma_{\co}(V)$ of weak backward
type $B$. There exist unique $\left\{ E_{i}\right\} _{i=1}^{m}\subseteq\Ic$
such that
\[
\tc(E_{i})=2\cos\left(\frac{i\pi}{m+1}\right)\cdot\sign(\tc(L(\Ic))),
\]
 where $L(\Ic)$ is the left edge of $\Ic$ . In addition, for all
$1\leq i\leq m$, $E_{i}\in I_{[\co,m]}^{i}$, where $\Icmi$ are
the spectral bands from Lemma~\ref{lem:=000020backward=000020implies=000020forward=000020-=000020partial}
(also Definition~\ref{def:ForwardType}).
\end{cor}

In order to upgrade Lemma~\ref{lem:=000020backward=000020implies=000020forward=000020-=000020partial}~(\ref{enu:=000020backward=000020implies=000020forward=000020-=000020partial=000020_=000020B=000020bands})
to the spectral bands $\left\{ \Icmnj\right\} _{j=1}^{M+1}$ (as required
in \ref{enu:=000020B=000020property} in Definition~\ref{def:ForwardType}),
the following lemma is used.
\begin{lem}
\label{lem:=000020Tower=000020property} Let $V>4$. Let $n\in\N$,
$[\co',n]\in\Co$ and $I$ be a spectral band in $\sigma_{[\co',n]}$
of weak backward type $B$. Then there is a unique spectral band $J$
in $\sigma_{[\co',n+1]}$ such that $J\strict I$. In particular,
$J$ is of backward type $B$.
\end{lem}

\begin{proof}
The proof of this lemma follows from Lemma~\ref{lem:=000020backward=000020implies=000020forward=000020-=000020partial}.
Let $I\subset\sigma_{[\co',n]}$ be a spectral band of backward type
$B$. Applying Lemma~\ref{lem:=000020backward=000020implies=000020forward=000020-=000020partial}
for $\co:=[\co',n]$ gives that there exists a unique spectral band
$J\subseteq\sigma_{[\co',n,1]}$ such that $J\strict I$. Equivalently,
$J$ is a spectral band of backward type $A$ when $J$ is considered
a spectral band of $\sigma_{[\co',n,1]}$. Nevertheless, since $\sigma_{[\co',n,1]}=\sigma_{[\co',n+1]}$,
we may consider $J$ as a spectral band of $\sigma_{[\co',n+1]}$
since $\sigma_{[\co',n,1,0]}=\sigma_{[\co',n+1,-1]}$ by Lemma~\ref{lem:=000020Trace=000020Depends=000020On=000020Value=000020Only}
Thus, $J$ is of backward type $B$ as a spectral band of $\sigma_{[\co',n+1]}$
for the same reason, since $J\strict I$.
\end{proof}
Lemma~\ref{lem:=000020Tower=000020property} implies that every spectral
band of backward type $B$ of $\sigma_{[\co,n]}$ contains another
(unique) spectral band of backward type $B$ of $\sigma_{[\co,n+1]}$.
This construction continues indefinitely by recursion, and hence we
call it the \emph{tower property}.
\begin{cor}
[Tower-property]\label{cor:Tower=000020Property}Let $V>4$, $[\co',1]\in\Co$
and $I_{[\co',1]}$ be a spectral band of $\sigma_{[\co',1]}(V)$
of weak backward type $B$. Then there are unique spectral bands $I_{[\co',n]}$
in $\sigma_{[\co',n]}(V)$ of backward type $B$ for all $n\geq2$
such that $I_{[\co',j+1]}\strict I_{[\co',j]}$ for all $j\in\N$.
\end{cor}

\begin{proof}
This follows directly from an induction over $n\geq2$ and Lemma~\ref{lem:=000020Tower=000020property}
\end{proof}
\begin{prop}
\label{prop:=000020backward=000020implies=000020forward} Let $V>4$
and $\co\in\Co$ with $[\co,m]\in\Co$ for all $m\in\N$. Let $\Ic(V)$
be a spectral band in $\sigma_{\co}(V)$ of weak backward type $A$
(respectively weak backward type $B$). Then $\Ic(V)$ is of $m$-forward
type $A$ (respectively $m$-forward type $B$) for all $m\in\N.$
In addition,
\begin{enumerate}
\item \label{enu:=000020backward=000020implies=000020forward_=000020A=000020bands}the
spectral bands $\left\{ \Icmi(V)\right\} _{i=1}^{M}$ mentioned in
the $m$-forward definition (Definition~\ref{def:ForwardType}) are
the only spectral bands in $\sigcm(V)$ which are included in $\Ic(V)$.
\item \label{enu:=000020backward=000020implies=000020forward_=000020B=000020bands}the
spectral bands $\left\{ \Icmnj(V)\right\} _{j=1}^{M+1}$ mentioned
in the $m$-forward definition (Definition~\ref{def:ForwardType})
are the only spectral bands in $\sigcmn(V)$ which satisfy properties
\ref{enu:B1Property} and \ref{enu:B2Property} (in Definition~\ref{def:ForwardType}).
\item The strong interlacing property (\ref{eq:Interlacing_strict}) holds.
\end{enumerate}
\end{prop}

\begin{proof}
As before, we omit all the $V$ dependencies here but note that the
proof relies on various results using $V>4$. Combining Lemma~\ref{lem:=000020backward=000020implies=000020forward=000020-=000020partial}
with Lemma~\ref{lem:=000020Tower=000020property} implies the existence
and uniqueness of the spectral bands $\left\{ \Icmi\right\} _{i=1}^{M}$
and $\left\{ I_{[\co,m,1]}^{j}\right\} _{j=1}^{M+1}$, and these spectral
bands satisfy properties \ref{enu:A1Property},\ref{enu:A2Property}
and (\ref{eq:Interlacing_strict}), see Definition~\ref{def:ForwardType}.
Let $1\leq j\leq M+1$. Then Corollary~\ref{cor:Tower=000020Property}
(applied for $\co'=[\co,m]$) asserts that for all $n\in\N$, there
exist a unique spectral band $\Icmnj$ such that
\[
\Icmnj\strict I_{[\co,m,n-1]}^{j}\strict\ldots\strict I_{[\co,m,1]}^{j}\subseteq\Ic.
\]
By construction $\Icmnj$ is of backward type $B$ for all $n>1$
and not of weak backward type $A$ by Corollary~\ref{cor:=000020Weak=000020Backward=000020Type}.
Thus, for all $n>1$, the spectral bands $\left\{ I_{[\co,m,n]}^{j}\right\} _{j=1}^{M+1}$satisfy
properties \ref{enu:B1Property} and \ref{enu:B2Property}. All that
is left to show is $I_{[\co,m,1]}^{j}\strict\Ic$ for all $1\leq j\leq M+1$.

Let $E_{-}<E_{+}$ be chosen such that $\Ic=[E_{-},E_{+}]$. Thus,
$|t_{\co}(E_{\pm})|=2$ holds by Proposition~\ref{prop:=000020Floquet-Bloch=000020via=000020transfer=000020matrix}.
In addition, Lemma~\ref{lem:=000020backward=000020implies=000020forward=000020-=000020partial}
implies $I_{[\co,m,1]}^{1}\precd\ldots\precd I_{[\co,m,1]}^{M+1}$
and $I_{[\co,m,1]}^{j}\subseteq\Ic$. Therefore, it is sufficient
to prove $|t_{[\co,m,1]}(E_{\pm})|>2$ in order to conclude that $I_{[\co,m,1]}^{j}\strict\Ic$
for all $1\leq j\leq M+1$.

As $|t_{\co}(E_{\pm})|=2$ , Lemma~\ref{Lem-Chebyshev-Traces} implies
$|S_{l+1}\left(t_{\co}(E_{\pm})\right)|=l+2$. Thus, applying Lemma~\ref{lem:=000020Trace=000020Cheby=000020Formula}
and the reversed triangle inequality gives that for $m\geq l\geq-1$,
\begin{align}
|t_{[\co,m,1]}(E_{\pm})| & =|t_{[\co,m+1]}(E_{\pm})|\label{eq:LowerTrace_ProofforwardProperty}\\
 & =|S_{l+1}\left(t_{\co}(E_{\pm})\right)t_{[\co,m-l]}(E_{\pm})-S_{l}\left(t_{\co}(E_{\pm})\right)t_{[\co,m-l-1]}(E_{\pm})|\nonumber \\
 & \geq(l+2)|t_{[\co,m-l]}(E_{\pm})|-(l+1)|t_{[\co,m-l-1]}(E_{\pm})|.\nonumber 
\end{align}
We continue estimating the previous term by a suitable choice of $l$
depending whether $\Ic(V)$ is of weak backward type $A$ or $B$.

If $\Ic(V)$ is of weak backward type $A$, set $l=m-1$. Note that
$\Ic(V)\subseteq\sigma_{[\co,0]}(V)$ (since we assume now that $\Ic(V)$
is of weak backward type $A$) and therefore $|t_{[\co,0]}(E_{\pm})|\leq2$.
Then $E_{\pm}\subseteq\sigma_{\co}(V)\cap\sigcz(V)$ leads to $E_{\pm}\not\in\sigma_{[\co,1]}(V)$
by Proposition~\ref{prop:ThreeConsecSpecCantIntersect} and $V>4$.
Thus, $|t_{[\co,1]}(E_{\pm})|>2$ is concluded from Proposition~\ref{prop:=000020Floquet-Bloch=000020via=000020transfer=000020matrix}.
Substituting $|t_{[\co,0]}(E_{\pm})|\leq2$ and$|t_{[\co,1]}(E_{\pm})|>2$
in Equation~(\ref{eq:LowerTrace_ProofforwardProperty}) gives
\[
|t_{[\co,m,1]}(E_{\pm})|\geq(m+1)|t_{[\co,1]}(E_{\pm})-m|t_{[\co,0]}(E_{\pm})|>2
\]
finishing the proof in this case.

If $\Ic(V)$ is of weak backward type $B$, set $l=m$ and note that
$|t_{[\co,-1]}(E_{\pm})|\leq2$. In addition, $|t_{[\co,0]}(E_{\pm})|>2$
holds as $E_{\pm}\in\Ic(V)$ and $\Ic(V)$ is not of weak backward
type $A$ by Corollary~\ref{cor:=000020Weak=000020Backward=000020Type}
and $V>4$. Then Equation~(\ref{eq:LowerTrace_ProofforwardProperty})
leads to 
\[
|t_{[\co,m,1]}(E_{\pm})|\geq(m+2)|t_{[\co,0]}(E_{\pm})|-(m+1)|t_{[\co,-1]}(E_{\pm})|>2
\]
finishing the proof.
\end{proof}
Proposition~\ref{prop:=000020backward=000020implies=000020forward}
is comparable to \cite[lem. 3.3]{Raym95}. Nevertheless, Proposition~\ref{prop:=000020backward=000020implies=000020forward}
is slightly stronger in three aspects: using everywhere the strict
inclusion $\strict$ rather than $\subset$; stating the properties
for all $m,n\in\N$; and also by having in \ref{enu:B1Property} $I_{[\co,m,n]}^{j}(V)\strict I_{[\co,m,n-1]}^{j}(V)$
for all $n\in\N$ rather than just $I_{[\co,m,n]}^{j}(V)\strict I_{\co}^{j}(V)$.
. This strengthening is crucial in \cite{BanBecLoe_24}, and that
is why we choose to deviate from the original exposition in \cite[lem. 3.3]{Raym95}.

Proposition~\ref{prop:=000020backward=000020implies=000020forward}
shows the implication between (weak) backward type and forward type.
Therefore, we are motivated to include both in one definition (Definition~\ref{def:ABTypes})
and to prove their equivalence if $V>4$, see Theorem~\ref{Thm-V>4_AllTypeA_B}.
\begin{defn}
\label{def:ABTypes} Let $V>4$ and $m\in\N$. Let $\co\in\Co$ such
that $[\co,m]\in\Co$. A spectral band $I_{\co}$ of $\sigma_{\co}(V)$
is called of
\end{defn}

\begin{itemize}
\item {\em type $A$} if $I_{\co}$ is of backward type $A$ and it is
also of $m$-forward type $A$ for all $m\in\N$.
\item {\em type $B$} if $I_{\co}$ is of backward type $B$ and it is
also of $m$-forward type $B$ for all $m\in\N$.
\end{itemize}
Before proving the main theorem - that each spectral band is of type
$A$ or $B$, we provide a useful corollary of Proposition~\ref{prop:=000020backward=000020implies=000020forward}
for which the following example is a warm up.
\begin{example}
\label{exa:StructureBands_c_1} Let $V>4$. Then a short computation
yields that $\sigma_{[0,0]}(V)=[-2,2]=:I_{[0,0]}(V)$ is of backward
type $A$ and $\sigma_{[0,0,1]}(V)=[-2+V,2+V]=:K_{[0,0,1]}(V)$ is
of backward type $B$, see also Example~\ref{exa:=000020backward=000020type=000020=00005B0,0=00005D=000020and=000020=00005B0,0,1=00005D}.
Thus, Proposition~\ref{prop:=000020backward=000020implies=000020forward}
implies that $I_{[0,0]}(V)$ is of type $A$ and $K_{[0,0,1]}(V)$
is of type $B$. Moreover, Proposition~\ref{prop:=000020backward=000020implies=000020forward}
and Lemma~\ref{lem:=000020Tower=000020property} imply for all $n\geq2$,
that there are $\left\{ I_{[0,0,n]}^{i}(V)\right\} _{i=1}^{n-1}$
of type $A$ and one spectral band $K_{[0,0,n]}(V)$ of type $B$
such that 
\[
\sigma_{[0,0,n]}=\bigcup_{i=1}^{n-1}I_{[0,0,n]}^{i}(V)\cup K_{[0,0,n]}(V),\qquad K_{[0,0,l]}(V)\strict K_{[0,0,l-1]}(V)\,\textrm{for}\,2\leq l\leq n,
\]
and 
\[
I_{[0,0,n]}^{1}(V)\}\precd I_{[0,0,n]}^{2}(V)\precd\ldots\precd I_{[0,0,n]}^{n-1}(V)\precd K_{[0,0,n]}(V).
\]
The structure is sketched in Figure~\ref{fig:StructureBands_c_1}.

\begin{figure}[hbt]
\includegraphics[scale=0.7]{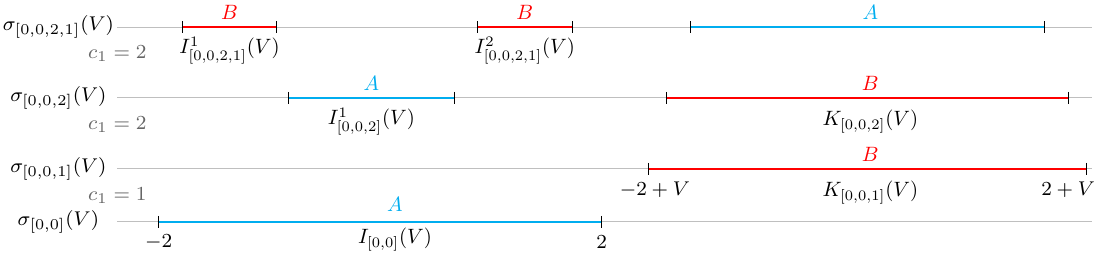}\caption{The first spectral bands defined in Example~\ref{exa:StructureBands_c_1}.}
\label{fig:StructureBands_c_1}
\end{figure}
\end{example}

\begin{cor}
\label{cor:EverySpectralBand=00003DIcmi=000020or=000020Icmnj}Let
$V>4$ and $\co'=[0,c_{0},c_{1},\ldots,c_{k}]\in\Co$ be such that
$c_{k}\geq1$ if $k\geq1$ and $\varphi(\co')\in[0,1]$. Consider
a spectral band $J$ in $\sigma_{\co'}(V)$.
\end{cor}

\begin{enumerate}
\item \label{enu:EverySpectralBand=00003DIcmi=000020or=000020Icmnj=000020_=000020type=000020A}If
$J$ is of weak backward type $A$, then $J$ is of backward type
$A$ and either
\begin{itemize}
\item $J=[-2,2]$ and $\varphi(\co')=0$, or
\item there is a unique spectral band $\Ic$ in $\sigma_{\co}(V)$ with
$\co=[0,c_{0},c_{1},\ldots,c_{k-1}]$ and a $1\leq i\leq M$(where
$M=c_{k}-1$ if $\Ic$ is of type $A$ and $M=c_{k}$ if $\Ic$ is
of type $B$) such that $J=I_{[\co,m]}^{i}$ with $m=c_{k}$ where
the the latter is the unique $i$th spectral band associated with
$\Ic$ defined in \ref{enu:=000020A=000020property}.
\end{itemize}
\item \label{enu:EverySpectralBand=00003DIcmi=000020or=000020Icmnj=000020_=000020type=000020B}If
$J$ is of weak backward type $B$, then $J$ is of backward type
$B$ and either
\begin{itemize}
\item $J=K_{[0,0,n]}$ from Example~\ref{exa:StructureBands_c_1}, $\varphi(\co)=\frac{1}{n}$
and $k=1$, or
\item there is a unique spectral band $\Ic$ in $\sigc(V)$ with $\co=[0,c_{0},c_{1},\ldots,c_{k-2}]$
and a $1\leq j\leq M+1$(where $M=c_{k-1}-1$ if $\Ic$ is of type
$A$ and $M=c_{k-1}$ if $\Ic$ is of type $B$) such that $J=I_{[\co,m,n]}^{j}$
with $m=c_{k-1},n=c_{k}$ where the latter is the unique $j$th spectral
band associated with $\Ic$ defined in \ref{enu:=000020B=000020property}.
\end{itemize}
In addition, there is a unique spectral band $I_{[\co,c_{k-1},1]}$
in $\sigma_{[\co,c_{k-1},1]}(V)$ of type $B$ such that either $J=I_{[\co,c_{k-1},1]}^{j}$
(if $c_{k}=1$) or $J\strict I_{[\co,c_{k-1},1]}^{j}$ (if $c_{k}>1$).

\end{enumerate}
\begin{proof}
Define $\co=[0,c_{0},c_{1},\ldots,c_{k-1}]$.

(\ref{enu:EverySpectralBand=00003DIcmi=000020or=000020Icmnj=000020_=000020type=000020A})
Let $J\subseteq\sigma_{\co'}(V)$ be of weak backward type $A$. Note
first that $\sigma_{[0,0]}(V)=[-2,2]$, where the corresponding spectral
band is of type $A$, see Example~\ref{exa:StructureBands_c_1}.
If $\varphi(\co')\in[0,1]$, we conclude $k\geq1$ and $\varphi([\co',0])=\varphi(\co)\geq0$.
. Then there is a spectral band $\Ic\subseteq\sigc$(V) with $J\subseteq\Ic$.
By Corollary~\ref{cor:=000020Weak=000020Backward=000020Type}, $\Ic$
is either of weak backward type $A$ or $B$. Thus, Proposition~\ref{prop:=000020backward=000020implies=000020forward}~(\ref{enu:=000020backward=000020implies=000020forward_=000020A=000020bands})
implies the statement and in particular the uniqueness of the bands
$\left\{ I_{[\co,m]}^{i}(V)\right\} _{i=1}^{M}$ for $m=c_{k}$.

(\ref{enu:EverySpectralBand=00003DIcmi=000020or=000020Icmnj=000020_=000020type=000020B})
If $k=0$, then $\co'=[0,0]$ and the spectral band in $\sigma_{[0,0]}(V)$
is of type $A$, see Example~\ref{exa:StructureBands_c_1}. If $k=1$,
then $\co'=[0,0,n]$ with $n\geq1$ by assumption. Thus, the only
spectral band $J\subseteq\sigma_{\co'}(V)$ of weak backward type
$B$ is the spectral band $K_{[0,0,n]}$ described in Example~\ref{exa:StructureBands_c_1}.
Hence, $J=K_{[0,0,n]}$ follows satisfying $K_{[0,0,n]}\subseteq K_{[0,0,1]}=[-2+V,2+V]=\sigma_{[0,0,1]}(V)$,
where $K_{[0,0,1]}$ is of type $B$. Note that $K_{[0,0,n]}\strict K_{[0,0,1]}$
if $n>1$ and else $K_{[0,0,n]}=K_{[0,0,1]}$.

Next, we treat the case $k\geq2$ with $c_{k}\geq1$. If $c_{k}=1$,
then $\sigma_{[\co',-1]}(V)=\sigc(V)$ holds by Lemma~\ref{lem:=000020Trace=000020Depends=000020On=000020Value=000020Only}.
Thus, there is a spectral band $\Ic\subseteq\sigc(V)$ such that $J\strict\Ic$.
Set $m=c_{k-1}$. Observe that $J$ is a spectral band in $\sigma_{[\co,m,1]}(V)$.
Thus, Proposition~\ref{prop:=000020backward=000020implies=000020forward}~(\ref{enu:=000020backward=000020implies=000020forward_=000020B=000020bands})
implies $J=I_{[\co,m,c_{1}]}^{j}$ for some $1\leq j\leq M+1$.

If $c_{k}\geq2$, then $\sigma_{[\co',-1]}(V)=\sigma_{[\co,c_{k-1},c_{k}-1]}(V)$
where $c_{k}-1\geq1$. Thus, there is a spectral band $J_{c_{k}-1}$
in $\sigma_{[\co,c_{k-1},c_{k}-1]}(V)=\sigma_{[\co',-1]}(V)$ with
$J\subseteq J_{c_{k}-1}$. Since $V>4$, Corollary~\ref{cor:=000020Weak=000020Backward=000020Type}
asserts that $J_{c_{k}-1}$ is either of weak backward type $A$ or
$B$. We claim that $J_{c_{k}-1}$ is of weak backward type $B$.
Therefore assume towards contradiction that $J_{c_{k}-1}$ is of weak
backward type $A$. Then $J_{c_{k}-1}\subseteq\sigma_{[\co,c_{k-1},c_{k}-1,0]}(V)=\sigma_{[\co',0]}(V)$
follows. Thus, $J\subseteq\sigma_{\co'}(V)$, $J\subseteq\sigma_{[\co',0]}(V)$
and $J\subseteq\sigma_{[\co',-1]}(V)$ while $J\neq\emptyset$, contradicting
Proposition~\ref{prop:ThreeConsecSpecCantIntersect} using $V>4$,
Hence, $J_{c_{k}-1}$ is of weak backward type $B$. Thus, we can
inductively conclude that there are spectral bands $J_{l}$ in $\sigma_{[\co,c_{k-1},l]}(V)$
of weak backward type $B$ for all $1\leq l\leq c_{k}-1$ such that
\[
J\subseteq J_{c_{k}-1}\subseteq J_{c_{k}-2}\subseteq\ldots\subseteq J_{1}.
\]
Since $J_{1}\subseteq\sigma_{[\co,c_{k-1},1]}(V)$ is of weak backward
type $B$, $J_{1}$ is included in a spectral band $\Ic$ in $\sigc(V)$.
Set $m=c_{k-1}$ and $n=c_{k}$. Observe that $J$ is a spectral band
in $\sigma_{[\co,m,n]}(V)$ and $\Ic$ is either of weak backward
type $A$ or $B$. Thus, Proposition~\ref{prop:=000020backward=000020implies=000020forward}~(\ref{enu:=000020backward=000020implies=000020forward_=000020B=000020bands})
implies $J=I_{[\co,m,n]}^{j}$ for some $1\leq j\leq M+1$. In particular,
$J=I_{[\co,c_{k-1},1]}^{j}$ if $n=c_{k}=1$ and $J\strict I_{[\co,c_{k-1},1]}^{j}$
if $n=c_{k}>1$.
\end{proof}
\begin{thm}
\label{Thm-V>4_AllTypeA_B} For all $V>4$ and $\co\in\Co$ with $[\co,m]\in\Co$
for all $m\in\N$, every spectral band in $\sigma_{\co}(V)$ is either
of type $A$ or $B$ and its type is independent of the value of $V>4$.
In addition, for every spectral band $\Ic(V)$ in $\sigma_{\co}(V)$
and all $m,n\in\N$, the spectral bands $\left\{ \Icmi(V)\right\} _{i=1}^{M}$
and $\left\{ \Icmnj(V)\right\} _{j=1}^{M+1}$ mentioned in the $m$-forward
definition (Definition~\ref{def:ForwardType}) are unique and the
strong interlacing property (\ref{eq:Interlacing_strict}) holds.
\end{thm}

\begin{rem*}
Theorem~\ref{Thm-V>4_AllTypeA_B} collects all the previous statements
of this section. This result is comparable to \cite[lem. 3.3]{Raym95},
as was discussed before. The independence of the spectral band type
in the value of $V>4$ was not explicitly mentioned in \cite[lem. 3.3]{Raym95}.
Showing this is based on combining a continuity argument and the three
intersection property (Proposition~\ref{prop:ThreeConsecSpecCantIntersect}).
Further discussions on the role of this independence may be found
in \cite{BanBecLoe_24}.
\end{rem*}
\begin{proof}
Let $V>4$ and $\co\in\Co$ be such that $[\co,m]\in\Co$ for all
$m\in\N$. Let $\Ic(V)$ be a spectral band in $\sigc(V)$. By Corollary~\ref{cor:=000020Weak=000020Backward=000020Type},
we have that $\Ic(V)$ has a well-defined weak backward type (either
$A$ or $B$) for all $V>4$. Suppose $\Ic(V)$ is of weak backward
type $A$ (respectively $B$). Then Proposition~\ref{prop:=000020backward=000020implies=000020forward}
implies that $\Ic$(V) is of $m$-forward type $A$ (respectively
$m$-forward type $B$) for all $m\in\N$. In addition, Corollary~\ref{cor:EverySpectralBand=00003DIcmi=000020or=000020Icmnj}
asserts that $\Ic(V)$ is also of backward type $A$ (respectively
$B$). Hence, $\Ic(V)$ is of type $A$ (respectively \textbf{$B$).}

According to the previous considerations, $\Ic(V)$ is either of type
$A$ or $B$ for each $V>4$. We explain now why this type is independent
on the value of $V$ (as long as $V>4$). Therefore observe that it
suffices to prove that if $\Ic(V_{0})$ is of type $A$ (respectively
$B$) for one $V_{0}>4$, then $\Ic(V)$ is of type $A$ (respectively
$B$) for all $V>4$. Assume towards contradiction this is not the
case. Then there is a $V_{0}>4$ and a sequence $\left\{ V_{n}\right\} _{n\in\N}\subseteq(4,\infty)$
such that $\lim_{n\to\infty}V_{n}=V_{0}$ and the type of $\Ic(V_{0})$
is different to the type $\Ic(V_{n})$ for all $n\in\N$. Without
loss of generality assume $\Ic(V_{0})$ is of type $A$ and $\Ic(V_{n})$
is of type $B$ for all $n\in\N$ (the other case is treated similarly).
In particular, $\Ic(V_{0})\subseteq\sigma_{[\co,0]}(V_{0})$ and $\Ic(V_{n})\subseteq\sigma_{[\co,-1]}(V_{n})$
for all $n\in\N$. In order to continue, we need the following observations.

Let $\co'\in\Co$. For $V\in\R$, the preimage $t_{\co'}(\cdot,V)^{-1}(\left\{ \pm2\right\} )$
coincides with the edges of the spectral bands, confer the discussion
at Proposition~\ref{prop:=000020Floquet-Bloch=000020via=000020transfer=000020matrix}.
Thus, if $J(V)=[a(V),b(V)]$ is a spectral band of $\sigma_{\co'}(V)$,
then $|t_{\co'}(a(V),V)|=2=|t_{\co'}(b(V),V)|$. From the definition
of $t_{\co'}$ it is immediate that $(4,\infty)\ni V\mapsto a(V)\in\R$
and $(4,\infty)\ni V\mapsto b(V)\in\R$ are continuous. Note that
indeed this edges are continuous on $V\in\R\setminus\left\{ 0\right\} $,
see also a discussion in \cite[cor. 3.2]{BanBecLoe_24}. Thus, $(4,\infty)\ni V\mapsto\sigma_{\co'}(V)$
is also continuous (as a finite union of intervals with continuous
edges) in the Hausdorff metric.

Let $\Ic(V)=[a(V),b(V)]$. By assumption, we have $a(V_{0})\in\Ic(V_{0})\subseteq\sigma_{[\co,0]}(V_{0})$
and $a(V_{n})\in\Ic(V_{n})\subseteq\sigma_{[\co,-1]}(V_{n})$ for
all $n\in\N$. By continuity of $V\mapsto a(V)$, $V\mapsto\sigma_{[\co,0]}(V)$
and $V\mapsto\sigma_{[\co,-1]}(V)$, we conclude
\[
\sigma_{[\co,0]}(V_{0})\ni a(V_{0})=\lim_{n\to\infty}a(V_{n})\in\lim_{n\to\infty}\sigma_{[\co,-1]}(V_{n})=\sigma_{[\co,-1]}(V_{0}).
\]
Thus, $a(V_{0})\in\sigc(V_{0})\cap\sigma_{[\co,0]}(V_{0})\cap\sigma_{[\co,-1]}(V_{0})$
follows contradicting Proposition~\ref{prop:ThreeConsecSpecCantIntersect}
and $V>4$.
\end{proof}

\section{The integrated density of state for Sturmian Hamiltonian \protect\label{sec:=000020The=000020IDS}}

A Sturmian Hamiltonian, $\Ham$ with $\alpha\notin\Q$ gives rise
to periodic Hamiltonians $\Hrat$ whose spectra converge to $\sigma\left(\Ham\right)$
(Proposition~\ref{prop:=000020Convergence=000020to=000020aperiodic=000020spectrum}).
The spectra of these periodic operators exhibit a special structure,
as is described in the previous section and summarized in Theorem~\ref{Thm-V>4_AllTypeA_B}.
We employ it in this section in order to study the integrated density
of states of $\Ham$ for $V>4$.

\subsection{A light introduction to the integrated density of states and its
gap labels. \protect\label{subsec:=000020Intro=000020to=000020IDS}}

We briefly introduce the integrated density of states for the Sturmian
Hamiltonian, $\Ham$. First, restricting $\Ham$ to $\ell^{2}(\{1,\ldots,n\})$,
we obtain a hermitian $n\times n$ matrix, denoted by $\Ham|_{[1,n]}$.
We denote its set of $n$ eigenvalues by $\sigma\left(\Ham|_{[1,n]}\right)$
and use it to define 
\begin{equation}
\IDS(E):=\lim_{n\to\infty}\frac{\#\set{\lambda\in\sigma\left(\left.\Ham\right|_{[1,n]}\right)}{\lambda\leq E}}{n}.\label{eq:=000020IDS=000020definition}
\end{equation}
The limit in (\ref{eq:=000020IDS=000020definition}) is known to exist
for all $\alpha\in[0,1]\backslash\Q$, $V\in\R$ and $E\in\R$, see
e.g. \cite{Sim82-review,Hof93,DaFi22-book_1}. The function $E\mapsto N_{\alpha,V}(E)$
is called the integrated density of states (IDS) of $H_{\alpha,V}$.
There are a few equivalent ways to define the IDS in our case. Here,
we choose the way which is computationally the most convenient within
the framework developed in this paper. This definition of the IDS
is common in the physics literature. Within the mathematics literature,
it is also known by the name the integrated (normalized) empirical
spectral distribution. Two fundamental properties of the IDS in our
setting are:
\begin{enumerate}[label=(IDS\arabic*)]
\item \begin{flushleft}
\label{enu:=000020IDS-property-1} The IDS, $\IDS:\R\rightarrow\left[0,1\right]$
is a monotone, non-decreasing and a continuous function.
\par\end{flushleft}
\item \begin{flushleft}
\label{enu:=000020IDS-property-2} We have $E\in\R\backslash\sigma(\Ham)$
if and only if there exists an $\varepsilon>0$ such that the restriction
$\IDS$ is constant on $\left(E-\varepsilon,E+\varepsilon\right)$.
\par\end{flushleft}

\end{enumerate}
In particular, we have that the IDS is constant on the spectral gaps,
i.e., on the connected components of $\R\backslash\sigma(\Ham)$.
The values that the IDS attains at the gaps are also called the gap
labels. The gap labelling theory is a general theory \cite{Bell92-Gap,BelBovGhe92,DamFill23-GapLabel},
which predicts the set of all possible gap labels of an operator.
Applying the gap labelling theory to $\Ham$ leads to the following
assertion.
\begin{prop}
For all $\alpha\in[0,1]\setminus\Q$ and $V\in\R\setminus\{0\}$,
\[
\set{\IDS(E)}{E\in\R\backslash\sigma(\Ham)}\subset\set{l\alpha\mod 1}{l\in\Z}\cup\{1\}.
\]
\end{prop}

The question which was raised by Mark Kac (though in the context of
the Almost-Mathieu operator) is whether there is an equality above,
or in his words, ``Are all gaps there?''. Since then this problem
was given the name ``The Dry Ten Martini Problem'' \cite{Sim82-review}.
It is shown in Theorem~\ref{thm:=000020All=000020gaps=000020are=000020there=000020V>4}
that there is indeed equality if $V>4$ (in \cite{BanBecLoe_24} this
result is extended to all $V\neq0$).

As a first step towards the proof of Theorem~\ref{thm:=000020All=000020gaps=000020are=000020there=000020V>4},
we show how the definition of the IDS in (\ref{eq:=000020IDS=000020definition})
may be restated in terms of the spectral bands of the periodic approximations
$\Hrat$, as presented previously. Therefore, note that $\left\{ E\right\} =[E,E]$
is an interval and so we can use the notation $I\precd\left\{ E\right\} $
for another interval $I$.
\begin{prop}
\label{prop:IDSBandCounting} Let $V\in\R\setminus\{0\}$ and $\alpha\in[0,1]\setminus\Q$
with continued fraction expansion $\left(c_{k}\right)_{k=0}^{\infty}$.
Consider its convergents $\varphi([0,c_{0},\ldots,c_{k}])=\frac{p_{k}}{q_{k}},k\in\N$
with $p_{k},q_{k}$ coprime. Then for all $E\in\R$, we have 
\begin{align}
\IDS(E) & =\lim_{k\to\infty}\frac{\#\set I{I\text{ is a spectral band of }\sigma(H_{\pqk,V})\text{ with }I\precd\{E\}}}{q_{k}}.\label{eq:=000020IDS=000020via=000020band=000020counting}
\end{align}
\end{prop}

\begin{proof}
We start by noting that the words $\left(\omega_{\alpha}(1),\ldots,\omega_{\alpha}(q_{k})\right)$
and $\left(\ohn(1),\ldots,\ohn(q_{k})\right)$ are equal up to a cyclic
shift. This can be deduced for example by combining Lemma~\ref{lem:=000020WordRecursion}
with \cite[prop. 2.2.24]{Lothaire2002} (using that our words $W_{k}$
are the words $s_{k}$ in \cite[prop. 2.2.24]{Lothaire2002} up to
a cyclic shift). This means that the matrices $\left.\Ham\right|_{[1,q_{k}]}$
and $\left.H_{\pqk,V}\right|_{[1,q_{k}]}$ are unitarily equivalent
(since their diagonals are equal up to a cyclic shift). Using this
observation and passing to the subsequence $n_{k}:=q_{k},\,k\in\N$,
in the limit of (\ref{eq:=000020IDS=000020definition}) yields 
\begin{align*}
\IDS(E) & =\lim_{k\to\infty}\frac{\#\set{\lambda\in\sigma\left(\left.\Ham\right|_{[1,q_{k}]}\right)}{\lambda\leq E}}{q_{k}}\\
 & =\lim_{k\to\infty}\frac{\#\set{\lambda\in\sigma\left(\left.H_{\pqk,V}\right|_{[1,q_{k}]}\right)}{\lambda\leq E}}{q_{k}}.
\end{align*}

At this point the reader is referred to Appendix~\ref{sec:=000020Floquet-Bloch=000020Theory}
and in particular Proposition~\ref{prop:=000020Floquet-Bloch} where
the Floquet-Bloch theory is summarized. Assigning to each $H_{\pqk,V}$
a $q_{k}\times q_{k}$-hermitian matrix $H_{\cok,V}(\theta)$ with
$\theta\in[0,\pi]$, the union (over $\theta\in[0,\pi]$) of these
matrices eigenvalues equals to $\sigma_{\co_{k}}(V)=\sigma(H_{\pqk,V})$,
see Proposition~\ref{prop:=000020Floquet-Bloch}. From now on, set
$\theta=0$. The matrices $\left.H_{\pqk,V}\right|_{[1,q_{k}]}$ and
$H_{\cok,V}(0)$ differ by a matrix of rank two using eq.~(\ref{eq:=000020finite-dim=000020Ham=000020matrices}).
Hence, the counting functions $\#\set{\lambda\in\sigma\left(\left.H_{\pqk,V}\right|_{[1,q_{k}]}\right)}{\lambda\leq E}$
and $\#\set{\lambda\in\sigma\left(H_{\cok,V}(0)\right)}{\lambda\leq E}$
differ by at most two\footnote{To be more precise, the difference is a traceless matrix of rank two.
By appropriately applying perturbation theory one can show that this
results in at most a difference of one in the eigenvalue counting,
see e.g. \cite[cor. III.2]{BanBecLoe_24}.}. Hence, we may replace the numerator in the limit above to get 
\[
\IDS(E)=\lim_{k\to\infty}\frac{\#\set{\lambda\in\sigma\left(H_{\cok,V}(0)\right)}{\lambda\leq E}}{q_{k}}.
\]
 According to Proposition~\ref{prop:=000020Floquet-Bloch} and in
particular eq.~(\ref{eq:=000020Floquet-Bloch=000020-=000020spectral=000020bands=000020decomposition}),
$H_{\cok,V}(0)$ has exactly one eigenvalue in each spectral band
$\sigma_{\co_{k}}(V)$. Thus, the number of spectral bands $I$ in
$\sigma_{\co_{k}}(V)=\sigma(H_{\pqk,V})$ satisfying $I\precd\{E\}$
differs at most by one from $\#\set{\lambda\in\sigma\left(H_{\cok,V}(0)\right)}{\lambda\leq E}$.
Hence, (\ref{eq:=000020IDS=000020via=000020band=000020counting})
follows.
\end{proof}
\begin{rem*}
The equivalence between (\ref{eq:=000020IDS=000020definition}) and
(\ref{eq:=000020IDS=000020via=000020band=000020counting}) was conjectured
in \cite[sec. 5]{Casdagli1986}. An explanation of this equivalence
is given at the end of section 2 in \cite{Raym95}. The proof above
contains an elaborated argument\footnote{The denominators in (\ref{eq:=000020IDS=000020definition}) and (\ref{eq:=000020IDS=000020via=000020band=000020counting})
may differ by one, even if $E$ is in a spectral gap, as opposed to
what is written in \cite{Raym95}. This was also pointed out to LR
by Mark Embree. However, this does not affect the value to which the
limit converges.}.
\end{rem*}

\subsection{Symbolic representation (coding) of the periodic spectra\protect\label{subsec:=000020Coding}}

Fix $\alpha\notin\Q$ and consider the spectrum $\sigma\left(\Ham\right)$
for $V>4$. We use the spectra of periodic operators to provide covers
of $\sigma\left(\Ham\right)$ allowing us to represent the IDS as
a power series, see eq.~(\ref{lem:LIsACover}) in Section~\ref{subsec:=000020formula=000020for=000020IDS}.
Towards this we define.
\begin{defn}
\label{def:CoveringSets} For $\co\in\Co$ with $[\co,1]\in\Co$ we
define the level $\L_{\co;V}$ by
\[
\L_{\co,V}:=\left\{ I:\begin{array}{l}
I\text{ is a spectral band of}~\sigma_{\co}(V)~\textrm{of type \ensuremath{A} or \ensuremath{B} }\text{ or}\\
I\text{ is a spectral band of}~\sigma_{[\co,1]}(V)~\textrm{of type \ensuremath{B} }
\end{array}\right\} .
\]
\end{defn}

We equip the set $\L_{\co,V}$ with the order relation $\precd$,
i.e $[a,b]\precd[c,d]$ if $b<c$, which was already introduced before
Definition~\ref{def:ForwardType}. This is in fact a total order
relation on $\L_{\co,V}$ if $V>4$, as is shown next in Lemma~\ref{lem:PrecIsTotal}.

Let us first consider some examples. The lowest level is $\L_{[0,0],V}=\{[-2,2],[V-2,V+2]\}$.
Observe that if $V>4$, then $[-2,2]\precd[V-2,V+2]$. A sketch of
$\L_{[0,0],V}$ and other sets can be found in Figure~\ref{fig:labelset=000020example}.
Observe that $L_{[0,0],V}$ and $L_{[0,0,1],V}$ both contain the
interval $[-2+V,2+V]$. Thus, these sets $L_{[0,0],V}$ and $L_{[0,0,1],V}$
are not disjoint, in general.

\begin{figure}[hbt]
\includegraphics[scale=0.74]{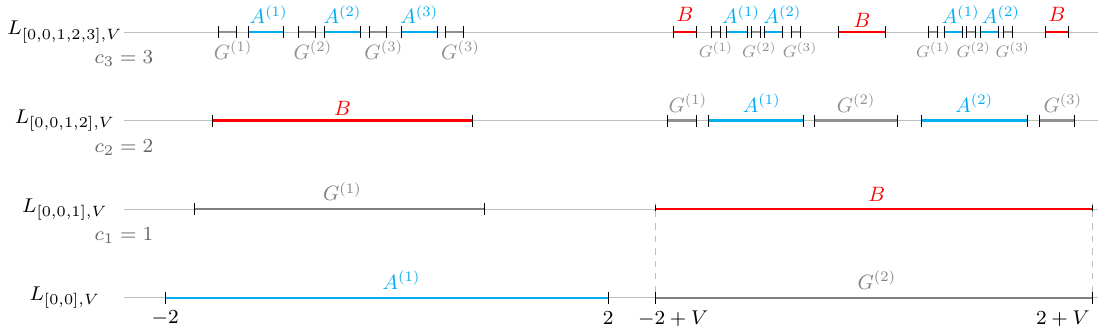}\caption{Visualisation of the sets $\protect\L_{\protect\co,V}$ for some $\protect\co\in\protect\Co$.
A label $A^{(i)},B$ or $G^{(i)}$ is assigned to each spectral band.
This label assignment describes the coding map $b_{V}$ from Proposition~\ref{prop:Coding_Spectrum}
for this particular cases.}
\label{fig:labelset=000020example}
\end{figure}

\begin{lem}
\label{lem:PrecIsTotal} Let $V>4$ and $\co,[\co,1]\in\Co$. Then,
for all $I,I'\in\L_{\co,V}$, we either have $I=I'$ or $I\cap I'=\emptyset$.
In particular, $(\L_{\co,V},\precd)$ is totally ordered.
\end{lem}

\begin{proof}
Let $I,I'\in\L_{\co,V}$. We only need to show that $I\cap I'=\emptyset$,
since then either $I\precd I'$ or $I'\precd I$ follows. If $I,I'$
are both spectral bands of $\sigma_{\co}(V)$ respectively $\sigma_{[\co,1]}(V)$,
then $I\cap I'=\emptyset$ as they are disjoint connected components
of the same spectrum, see Proposition~\ref{prop:=000020there=000020are=000020q=000020spectral=000020bands}.
Otherwise, assume $I\subseteq\sigma_{\co}(V)$ and $I'\subseteq\sigma_{[\co,1]}(V)$.
Then $I'$ is necessarily of type $B$ and in particular (by the backward
type $B$ property) there exists a spectral band $J'\subseteq\sigma_{[\co,1,-1]}(V)$
such that $I'\subseteq J'$. But $\sigma_{[\co,1,-1]}(V)=\sigma_{[\co,0]}(V)$
and by Proposition~\ref{prop:ThreeConsecSpecCantIntersect} we have
\[
\sigma_{\co}(V)\cap\sigma_{[\co,1]}(V)\cap\sigma_{[\co,0]}(V)=\emptyset,
\]
so that $I\cap I'=I\cap J'\cap I'=\emptyset$.
\end{proof}
Let $\alpha\in[0,1]\setminus\Q$ with continued fraction expansion
given by $\left(c_{k}\right)_{k=0}^{\infty}$. Define the finite continued
fraction expansions $\co_{k}:=[0,c_{0},c_{1},\ldots,c_{k}]\in\Co$
for $k\in\N_{0}$. In the following, we say $\L_{\co_{k},V}$ is a
cover of a set $A\subseteq\R$ if $A\subseteq\bigcup_{I\in\L_{\co_{k},V}}I$.
\begin{lem}
\label{lem:LIsACover} Let $V>4$ and $\alpha\in[0,1]\setminus\Q$
with continued fraction expansion $\left(c_{k}\right)_{k=0}^{\infty}$
and $\co_{k}:=[0,c_{0},c_{1},\ldots,c_{k}]\in\Co$ for $k\in\N_{0}$.
Then the following holds.
\begin{enumerate}
\item \label{enu:LIsACover_=000020of=000020k+1}For all $k\in\N_{0}$, $\L_{\co_{k},V}$
is a cover of $\L_{\co_{k+1},V}$.
\item \label{enu:LIsACover_=000020of=000020spectrum=000020alpha} For all
$k\in\N_{0}$, $\L_{\co_{k},V}$ is a cover of $\sigma(\Ham)$.
\item \label{enu:LIsACover_=000020convergence=000020sigma=000020alpha}For
all $k\in\N_{0}$, $\Lambda_{k}(V):=\sigma_{\co_{k}}(V)\cup\sigma_{[\co_{k},1]}(V)=\bigcup_{I\in\L_{\co_{k},V}}I$.
Furthermore, $\lim_{k\to\infty}\bigcup_{I\in\L_{\co_{k},V}}I=\bigcap_{k\in\N_{0}}\Lambda_{k}(V)=\sigma(\Ham)$
where the limit is taken in the Hausdorff metric.
\end{enumerate}
\end{lem}

\begin{proof}
If $I\subseteq\sigma_{[\co_{k+1},1]}(V)$ is of type $B$, then it
is contained in $\sigma_{[\co_{k+1},1,-1]}(V)=\sigma_{\co_{k}}(V)$
using Lemma~\ref{lem:=000020Trace=000020Depends=000020On=000020Value=000020Only}
and Corollary~\ref{cor:FirstTraceIds}. Thus, $\L_{\co_{k},V}$ covers
$I$ and so $\L_{\co_{k},V}$ is a cover of all spectral bands in
$\sigma_{[\co_{k+1},1]}(V)$ of type $B$. If $I\subseteq\sigma_{\co_{k+1}}(V)$
is of type $A$, then it is contained in $\sigma_{[\co_{k+1},0]}(V)=\sigma_{\co_{k}}(V)$
using Lemma~\ref{lem:=000020Trace=000020Depends=000020On=000020Value=000020Only}
and Corollary~\ref{cor:FirstTraceIds}. Thus, $\L_{\co_{k},V}$ covers
$I$. If $I\subseteq\sigma_{\co_{k+1}}(V)$ is of type $B$, then
Corollary~\ref{cor:EverySpectralBand=00003DIcmi=000020or=000020Icmnj}~(\ref{enu:EverySpectralBand=00003DIcmi=000020or=000020Icmnj=000020_=000020type=000020B})
imply that there is a $J\subseteq\sigma_{[\co_{k},1]}(V)$ of type
$B$ with $I\subseteq J$. Thus, $\L_{\co_{k},V}$ covers $I$ as
well. Combined with the previous considerations, we obtain that $\L_{\co_{k},V}$
is a cover of $\sigma_{\co_{k+1}}(V)$ and all spectral bands in $\sigma_{[\co_{k+1},1]}(V)$
of type $B$, namely $\L_{\co_{k},V}$ is a cover of $\L_{\co_{k+1},V}$.
Thus, (\ref{enu:LIsACover_=000020of=000020k+1}) is proven.

Having this, (\ref{enu:LIsACover_=000020of=000020spectrum=000020alpha})
follows from Proposition~\ref{prop:=000020Convergence=000020to=000020aperiodic=000020spectrum}.

By definition, we have $\bigcup_{I\in\L_{\co_{k},V}}I\subseteq\Lambda_{k}(V)$.
Moreover, every spectral band of $\sigma_{[\co_{k},1]}(V)$ of type
$A$ is contained in $\sigma_{\co_{k}}(V)$. Thus, $\bigcup_{I\in\L_{\co_{k},V}}I=\Lambda_{k}(V)$
and now (\ref{enu:LIsACover_=000020convergence=000020sigma=000020alpha})
follows from Proposition~\ref{prop:=000020Convergence=000020to=000020aperiodic=000020spectrum}.
\end{proof}
By the first part of the last lemma, every spectral band of $\L_{\co_{k+1},V}$
is contained in a unique spectral band of $\L_{\co_{k},V}$. We may
use this in order to construct a symbolic representation (coding)
of each spectral band in level $\L_{\co_{k},V}$ in terms of the spectral
bands in all previous levels in which it is recursively included.
\begin{defn}
\label{def:SpectralCodes} Let $\alpha\in[0,1]\setminus\Q$ with continued
fraction expansion $\left(c_{k}\right)_{k=0}^{\infty}$ and $\co_{k}:=[0,c_{0},c_{1},\ldots,c_{k}]\in\Co$
for $k\in\N_{0}$. Consider the countable alphabet $\Aa:=\{A^{(i)}:i\in\N\}\cup\{G^{(i)}:i\in\N\}\cup\{B\}$.
A (finite or infinite) \emph{spectral-$\alpha$-code} is either a
finite sequence $\wo=\left(\w(0),\w(1),\ldots,\w(k)\right)\in\Aa^{k+1}$
or an infinite sequence $\wo=\left(\w(0),\w(1),\ldots\right)\in\Aa^{\N_{0}}$
satisfying the following:
\begin{itemize}
\item[($\Sigma$1)]  $\w(0)\in\{A^{(1)},G^{(2)}\}$,
\item[($\Sigma$2)]  If $\w(j)\in\{A^{(i)}:i\in\N\}$ then $\w(j+1)\in\{A^{(i)}:1\leq i\leq c_{j+1}-1\}\cup\{G^{(i)}:1\leq i\leq c_{j+1}\}$,
\item[($\Sigma$3)]  If $\w(j)=B$ then $\w(j+1)\in\{A^{(i)}:1\leq i\leq c_{j+1}\}\cup\{G^{(i)}:1\leq i\leq c_{j+1}+1\}$,
\item[($\Sigma$4)]  If $\w(j)\in\{G^{(i)}:i\in\N\}$ then $\w(j+1)=B$.
\end{itemize}
The set of all infinite spectral-$\alpha$-codes will be denotes $\Sigma_{\alpha}$.
Similarly, the set of all spectral-$\alpha$-codes in $\Aa^{k+1}$
is denoted by $\Sigma_{\co_{k}}$. Moreover, the set $\Scks\subseteq\Sck$
is defined as those $\wo=(\w(0),\dots,\w(k))\in\Sck$, who additionally
satisfy
\begin{itemize}
\item[($\Sigma$5)] $\w(k)\in\{A^{(i)}:i\in\N\}\cup\{B\}$.
\end{itemize}
\end{defn}

A depiction of conditions $(\Sigma2)-(\Sigma4)$ in Definition~\ref{def:SpectralCodes}
appears in Figure~\ref{fig:LocalRules}. 

\begin{figure}
\includegraphics[scale=0.7]{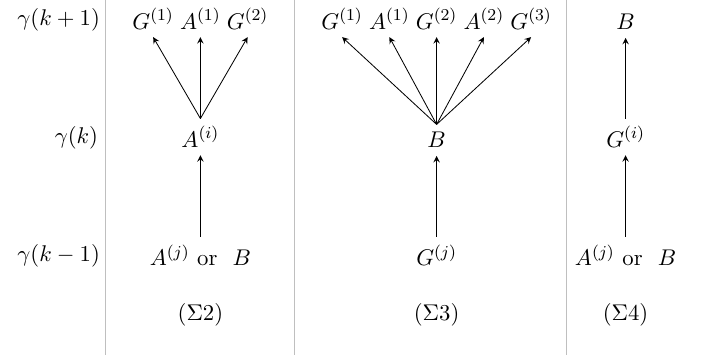}

\caption{Visualisation of the properties $(\Sigma2),(\Sigma3),(\Sigma4)$ with
$c_{k+1}=2$. Each figure shows the possible descendens $\protect\w(k+1)$
of $\protect\w(k)$ as well as from which element in $\protect\Aa$,
the element $\protect\w(k)$ could come from.}
\label{fig:LocalRules}
\end{figure}

\begin{rem}
There is a merit in embedding all the codes defined above in a tree
graph. Our depiction of the codes in Figure~\ref{fig:LocalRules}
and Figure~\ref{fig:CountingExample} uses this point of view. The
tree representation explicitly appears in \cite{BanBecLoe_24} using
a directed rooted tree with a strict (i.e., irreflexive) partial order
relation defined on its vertex set. It is called the \emph{spectral-$\alpha$-tree}
in \cite{BanBecLoe_24}. Here, we confine ourselves to the original
presentation of \cite{Raym95} using the symbolic representation of
codes (and appeal to the tree only via the figures). Finally, we note
that in \cite{BanBecLoe_24} the vertices of the tree graph are labeled
only by $A$ and $B$, as opposed to using also the label $G^{(i)}$
in the current paper.
\end{rem}

The previous definition is in close relation with the forward property
of a spectral band, see Definition~\ref{def:ForwardType}. This is
made precise in the following Proposition~\ref{prop:Coding_Spectrum}.
Before, we define a partial order $\lessdot$ on $\Aa:=\{A^{(i)}:i\in\N\}\cup\{G^{(i)}:i\in\N\}\cup\{B\}$
by setting 
\[
G^{(1)}\lessdot A^{(1)}\lessdot G^{(2)}\lessdot A^{(2)}\lessdot\ldots\,.
\]

Let $\alpha\in[0,1]\setminus\Q$ with infinite continued fraction
expansion $\left(c_{k}\right)_{k=0}^{\infty}$, convergents $\varphi(\co_{k})$
and $\co_{k}=[0,c_{0},c_{1},\ldots,c_{k}]\in\Co$ for $k\in\N_{0}$.
If $\wo,\eta\in\bigsqcup_{k\in\N_{0}}\Sigma_{\co_{k}}\cup\Sigma_{\alpha}$,
define

\[
\wo\lessdot\eta:\iff\begin{cases}
\w(0)\lessdot\eta(0),\text{ or }\\
\w(j)=\eta(j)\text{ and }\w(j+1)\lessdot\eta(j+1)\text{ for some }j\in\Nz.
\end{cases}
\]
This defines a partial order on $\bigsqcup_{k\in\N_{0}}\Sigma_{\co_{k}}$
respectively $\Sigma_{\alpha}$. We continue defining an encoding
of $\bigcup_{k\in\N_{0}}L_{\co_{k},V}$ via the spectral-$\alpha$-codes
$\bigsqcup_{k\in\N_{0}}\Sigma_{\co_{k}}$ preserving the partial order
relations, the types and inclusions. This statement deviates slightly
from \cite{Raym95} and follows the lines of \cite[prop. 7.1]{BanBecLoe_24}.
The reader is referred to Figure~\ref{fig:labelset=000020example},
where an example of some spectra are plotted together with the associated
code as described in the following proposition.
\begin{prop}
\label{prop:Coding_Spectrum}Let $\alpha\in[0,1]\setminus\Q$ with
infinite continued fraction expansion $\left(c_{k}\right)_{k=0}^{\infty}$
and $\co_{k}=[0,c_{0},c_{1},\ldots,c_{k}]\in\Co$ for $k\in\N_{0}$.
Then there exists for each $V>4$, a unique map 
\[
b_{V}:\bigsqcup_{k\in\N_{0}}\Sigma_{\co_{k}}\to\bigcup_{k\in\N_{0}}\Lck
\]
with the following properties:
\end{prop}

\begin{enumerate}
\item \label{enu:Coding_Spectrum_bijective}For each $k\in\N_{0}$, $b_{V}$
bijectively maps $\Sigma_{\co_{k}}$ onto $\Lck$.
\item \label{enu:Coding_Spectrum_Inclusion}For each $k\in\N$, we have
for all $\gamma\in\Sigma_{\co_{k-1}}$ and $\eta=(\eta(0),\ldots,\eta(k))\in\Sigma_{\co_{k}}$,
\[
\gamma=(\eta(0),\ldots,\eta(k-1))\quad\Leftrightarrow\quad b_{V}(\eta)\subseteq b_{V}(\gamma)\quad\Leftrightarrow\quad b_{V}(\gamma)\cap b_{V}(\eta)\neq\emptyset.
\]
\item \label{enu:Coding_Spectrum_order_preserving}Let $\wo,\eta\in\bigsqcup_{k\in\N_{0}}\Sigma_{\co_{k}}$.
Then $\wo\lessdot\eta$ if and only if $b_{V}(\gamma)\precd b_{V}(\eta)$.
\item \label{enu:Coding_Spectrum_Type}If $\gamma\in\Sigma_{\co_{k}}$ for
some $k\in\N_{0}$, then
\begin{enumerate}
\item $\gamma(k)\in A^{(i)}$ if and only if $b_{V}(\gamma)\subseteq\sigma_{\co_{k}}(V)$
is of type $A$.
\item $\gamma(k)\in B$ if and only if $b_{V}(\gamma)\subseteq\sigma_{\co_{k}}(V)$
is of type $B$.
\item $\gamma(k)\in G^{(i)}$ if and only if $b_{V}(\gamma)\subseteq\sigma_{[\co_{k},1]}(V)$
is of type $B$ and $b_{V}(\gamma)\cap\sigma_{\co_{k}}(V)=\emptyset$.
\end{enumerate}
\end{enumerate}
\begin{rem}
Note that Proposition~\ref{prop:Coding_Spectrum}~(\ref{enu:Coding_Spectrum_Type})
asserts that the spectral band $b_{V}(\gamma)$ for $\gamma\in\Sigma_{\co_{k}}$
is contained in a spectral gap of $\sigma_{\co_{k}}(V)$ if and only
if $\gamma(k)=G^{(i)}$ for some $i\in\N$. This in particular explains
the notation $G^{(i)}$ standing for a gap.
\end{rem}

\begin{proof}
We first note that every such map satisfying (\ref{enu:Coding_Spectrum_bijective})
and (\ref{enu:Coding_Spectrum_order_preserving}) must be unique since
$\Sigma_{\co_{k}}$ is totally ordered and $\Lck$ is totally ordered
by Lemma~\ref{lem:PrecIsTotal}.

First, suppose that such a map exists and justify that the equivalence
$b_{V}(\eta)\subseteq b_{V}(\gamma)\Leftrightarrow b_{V}(\gamma)\cap b_{V}(\eta)\neq\emptyset$
in (\ref{enu:Coding_Spectrum_Inclusion}) holds. Suppose $b_{V}(\gamma)\cap b_{V}(\eta)\neq\emptyset$
holds. Since $L_{\co_{k-1},V}$ is a cover of $L_{\co_{k},V}$ by
Lemma~\ref{lem:LIsACover}, we conclude $b_{V}(\eta)\subseteq\bigcup_{I\in L_{\co_{k-1},V}}I$.
Since the spectral bands in $L_{\co_{k-1},V}$ do not touch (Lemma~\ref{lem:PrecIsTotal})
and $b_{V}(\gamma)\cap b_{V}(\eta)\neq\emptyset$, we conclude $b_{V}(\eta)\subseteq b_{V}(\gamma)$.
The reverse implication is trivial.

We continue inductively defining the map $b_{V}$.

\uline{Induction base:} If $k=0$, we have $\co_{0}=[0,0]$, $\Sigma_{[0,0]}=\left\{ (A^{(1)}),(G^{(2)})\right\} $
and $L_{\co_{0},V}=\left\{ [-2,2],[-2+V,2+V]\right\} $, see also
Example~\ref{exa:StructureBands_c_1}. Define $b_{V}(A^{(1)})=[-2,2]$
and $b_{V}(G^{(2)})=[-2+V,2+V]$ satisfying (\ref{enu:Coding_Spectrum_bijective})-(\ref{enu:Coding_Spectrum_Type})
for $V>4$ by construction.

If $k=1$, we have $\co_{1}=[0,0,c_{1}]$ where $c_{1}\in\N$. By
Example~\ref{exa:StructureBands_c_1}, we have
\[
\sigma_{[0,0,c_{1}]}(V)=\bigcup_{j=1}^{c_{1}-1}I_{[0,0,c_{1}]}^{j}(V)\cup K_{[0,0,c_{1}]}(V)
\]
where
\[
I_{[0,0,c_{1}]}^{1}(V)\precd I_{[0,0,c_{1}]}^{2}(V)\precd\ldots\precd I_{[0,0,c_{1}]}^{c_{1}-1}(V)\precd K_{[0,0,c_{1}]}(V)
\]
and $I_{[0,0,c_{1}]}^{j}(V)\subseteq[-2,2]$ are of type $A$ and
$K_{[0,0,c_{1}]}(V)\subseteq[-2+V,2+V]$ is of type $B$. Since $\sigma_{[0,0,c_{1},1,-1]}(V)=\sigma_{[0,0]}(V)$,
every spectral band in $\sigma_{[0,0,c_{1},1]}(V)$ of type $B$ is
contained in $\sigma_{[0,0]}(V)=[-2,2]$. Applying Proposition~\ref{prop:=000020backward=000020implies=000020forward}
to the spectral band $[-2,2]$ of type $A$ with $m=c_{1}$ implies
that the spectral bands of type $B$ in $\sigma_{[0,0,c_{1},1]}(V)$
are $\left\{ J^{i}(V)\right\} _{i=1}^{c_{1}}$ with
\[
J^{1}(V)\precd I_{[0,0,c_{1}]}^{1}(V)\precd J^{2}(V)\precd\ldots\precd I_{[0,0,c_{1}]}^{c_{1}-1}(V)\precd J^{c_{1}}(V)
\]
and $J^{i}(V)\subseteq[-2,2]$. Thus, 
\[
L_{\co_{1},V}=\set{I_{[0,0,c_{1}]}^{i}(V)}{1\leq i\leq c_{1}-1}\cup\set{J^{i}(V)}{1\leq i\leq c_{1}}\cup\left\{ K_{[0,0,c_{1}]}(V)\right\} .
\]
Since $J^{c_{1}}(V)\subseteq[-2,2]$ and $K_{[0,0,c_{1}]}(V)\subseteq[-2+V,2+V]$,
we conclude $J^{c_{1}}(V)\precd K_{[0,0,c_{1}]}(V)$ using $V>4$.
In addition, we have
\[
\Sigma_{\co_{1}}=\set{(A^{(1)},A^{(i)})}{1\leq i\leq c_{1}-1}\cup\set{(A^{(1)},G^{(i)})}{1\leq i\leq c_{1}}\cup\left\{ (G^{(2)},B)\right\} .
\]
With this at hand, define $b_{V}\left((A^{(1)},A^{(i)})\right):=I_{[0,0,c_{1}]}^{i}(V)$,
$b_{V}\left((A^{(1)},G^{(i)})\right):=J^{i}(V)$ and $b_{V}\left((G^{(2)},B)\right):=K_{[0,0,c_{1}]}(V)$
satisfying (\ref{enu:Coding_Spectrum_bijective})-(\ref{enu:Coding_Spectrum_Type})
by construction.

\uline{Induction step:} Let $k\geq2$ be such that $b_{V}:\bigsqcup_{l=0}^{k}\Sigma_{\co_{k}}\to\bigcup_{l=0}^{k}\Lck$
satisfies (\ref{enu:Coding_Spectrum_bijective})-(\ref{enu:Coding_Spectrum_Type}).
We show how to extend $b_{V}:\Sigma_{\co_{k+1}}\to L_{\co_{k+1},V}$.
Let $\gamma'=(\gamma(0),\ldots,\gamma(k))\in\Sigma_{\co_{k}}$.

If $\gamma(k)=A^{(l)}$, then set $M=c_{k+1}-1$ and if $\gamma(k)=B$,
then set $M=c_{k+1}$. By the induction hypothesis and property (\ref{enu:Coding_Spectrum_Type}),
$b_{V}(\gamma')\subseteq\sigma_{\co_{k}}(V)$ is of type $A$ if $\gamma(k)=A^{(l)}$
and of type $B$ if $\gamma(k)=B$. Thus, Proposition~\ref{prop:=000020backward=000020implies=000020forward}
implies that there are exactly $\left\{ I_{\co_{k+1}}^{i}\right\} _{i=1}^{M}\subseteq\sigma_{\co_{k+1}}(V)$
of type $A$ and $\left\{ I_{[\co_{k+1},1]}^{j}\right\} _{j=1}^{M+1}\subseteq\sigma_{[\co_{k+1},1]}(V)$
of type $B$ such that $I_{\co_{k+1}}^{i},I_{[\co_{k+1},1]}^{j}\strict b_{V}(\gamma')$
and
\[
I_{[\co_{k+1},1]}^{1}\precd I_{\co_{k+1}}^{1}\precd\ldots\precd I_{\co_{k+1}}^{M}\precd I_{[\co_{k+1},1]}^{M+1}.
\]
Furthermore, ($\Sigma$2) implies that all choices for $\gamma(k+1)$
are $\set{A^{(i)}}{1\leq i\leq M}\cup\set{G^{(j)}}{1\leq j\leq M+1}$.
Then define for $\gamma=(\gamma(0),\ldots,\gamma(k+1))\in\Sigma_{\co_{k+1}}$,
\[
b_{V}\left(\gamma\right)=\begin{cases}
I_{\co_{k+1}}^{i} & \textrm{if}\,\gamma(k+1)=A^{(i)},\\
I_{[\co_{k+1},1]}^{j} & \textrm{if}\,\gamma(k+1)=G^{(j)}.
\end{cases}
\]
Thus, $b_{V}(\gamma)\subseteq b_{V}(\gamma')$ holds, namely this
definition satisfies (\ref{enu:Coding_Spectrum_Inclusion}) as well
as (\ref{enu:Coding_Spectrum_Type}). Furthermore, let $\gamma_{1}:=(\gamma(0),\ldots,\gamma(k),\gamma(k+1))$
and $\gamma_{2}:=(\gamma(0),\ldots,\gamma(k),\eta(k+1))$ for $\gamma(k)\in\left\{ A^{(l)},B\right\} $
and ${\color{blue}}\gamma(k+1),\eta(k+1)\in\set{A^{(i)}}{1\leq i\leq M}\cup\set{G^{(j)}}{1\leq j\leq M+1}$.
By construction, we conclude
\begin{equation}
\gamma_{1}\lessdot\gamma_{2}\quad\Leftrightarrow\quad b_{V}(\gamma_{1})\precd b_{V}(\gamma_{2}).\label{eq:Coding_Spectrum_Order=000020Preserving}
\end{equation}

If $\gamma(k)=G^{(l)}$, then the induction hypothesis and property
(\ref{enu:Coding_Spectrum_Type}) imply that $b_{V}(\gamma')\subseteq\sigma_{[\co_{k},1]}(V)$
is of type $B$ and $b_{V}(\gamma')\cap\sigma_{\co_{k}}(V)=\emptyset$.
Thus, Corollary~\ref{cor:Tower=000020Property} asserts that there
is a unique spectral band $J$ in $\sigma_{\co_{k+1}}(V)$ of type
$B$ with $J\subseteq b_{V}(\gamma')$. Note that if $c_{k+1}=1$,
then $J=b_{V}(\gamma')$. Define \textbf{$b_{V}(\gamma(0),\ldots,\gamma(k+1)):=J$}.
This definition satisfies (\ref{enu:Coding_Spectrum_Inclusion}) as
well as (\ref{enu:Coding_Spectrum_Type}).

By the previous considerations, we have defined the map $b_{V}:\Sigma_{\co_{k+1}}\to L_{\co_{k+1},V}$
and by construction it is injective and it satisfies (\ref{enu:Coding_Spectrum_Inclusion})
and (\ref{enu:Coding_Spectrum_Type}). Next, we prove that $b_{V}:\Sigma_{\co_{k+1}}\to L_{\co_{k+1},V}$
is also surjective. Therefore, let $J\in L_{\co_{k+1},V}$.

If $J\subseteq\sigma_{\co_{k+1}}(V)$ is of type $A$, then Corollary~\ref{cor:EverySpectralBand=00003DIcmi=000020or=000020Icmnj}~(\ref{enu:EverySpectralBand=00003DIcmi=000020or=000020Icmnj=000020_=000020type=000020A})
and $k\geq2$ imply that there is a unique spectral band $I_{\co_{k}}\subseteq\sigma_{\co_{k}}(V)$
such that $J=I_{[\co_{k},c_{k+1}]}^{i}$ for some $1\leq i\leq M$
(where $M=c_{k+1}-1$ if $I_{\co_{k}}$ is of type $A$ and $M=c_{k+1}$
if $I_{\co_{k}}$ is of type $B$). Since $b_{V}:\Sigma_{\co_{k}}\to L_{\co_{k},V}$
is bijective by induction hypothesis, there is a $\gamma=(\gamma(0),\ldots,\gamma(k))\in\Sigma_{\co_{k}}$
with $b_{V}(\gamma)=I_{\co_{k}}$. Moreover, property (\ref{enu:Coding_Spectrum_Type})
asserts $\gamma(k)=A^{(l)}$ if $I_{\co_{k}}$ is of type $A$ and
$\gamma(k)=B$ if $I_{\co_{k}}$ is of type $B$. Hence, $\gamma':=(\gamma(0),\ldots,\gamma(k),A^{(i)})\in\Sigma_{\co_{k+1}}$
by ($\Sigma$2) or ($\Sigma$3) and $b_{V}(\gamma')=J=I_{[\co_{k},c_{k+1}]}^{i}$.

If $J\subseteq\sigma_{[\co_{k+1},1]}(V)$ is of type $B$, then Corollary~~\ref{cor:EverySpectralBand=00003DIcmi=000020or=000020Icmnj}~(\ref{enu:EverySpectralBand=00003DIcmi=000020or=000020Icmnj=000020_=000020type=000020B})
and $k\geq2$ imply that there is a unique spectral band $I_{\co_{k}}\subseteq\sigma_{\co_{k}}(V)$
such that $J=I_{[\co_{k},c_{k+1},1]}^{i}$ for some $1\leq i\leq M+1$
(where $M=c_{k+1}-1$ if $I_{\co_{k}}$ is of type $A$ and $M=c_{k+1}$
if $I_{\co_{k}}$ is of type $B$). Since $b_{V}:\Sigma_{\co_{k}}\to L_{\co_{k},V}$
is bijective by induction hypothesis, there is a $\gamma=(\gamma(0),\ldots,\gamma(k))\in\Sigma_{\co_{k}}$
with $b_{V}(\gamma)=I_{\co_{k}}$. Moreover, property (\ref{enu:Coding_Spectrum_Type})
asserts $\gamma(k)=A^{(l)}$ if $I_{\co_{k}}$ is of type $A$ and
$\gamma(k)=B$ if $I_{\co_{k}}$ is of type $B$. Hence, $\gamma':=(\gamma(0),\ldots,\gamma(k),G^{(i)})\in\Sigma_{\co_{k+1}}$
by ($\Sigma$2) or ($\Sigma$3) and $b_{V}(\gamma')=J=I_{[\co_{k},c_{k+1},1]}^{i}$.

If $J\subseteq\sigma_{\co_{k+1}}(V)$ is of type $B$, then Corollary~\ref{cor:EverySpectralBand=00003DIcmi=000020or=000020Icmnj}~(\ref{enu:EverySpectralBand=00003DIcmi=000020or=000020Icmnj=000020_=000020type=000020B})
implies that there is a unique spectral band $I_{[\co_{k},1]}\subseteq\sigma_{[\co_{k},1]}(V)$
of type $B$ such that $J\subseteq I_{[\co_{k},1]}$. Since $b_{V}:\Sigma_{\co_{k}}\to L_{\co_{k},V}$
is bijective by induction hypothesis, there is a $\gamma=(\gamma(0),\ldots,\gamma(k))\in\Sigma_{\co_{k}}$
with $b_{V}(\gamma)=I_{[\co_{k},1]}$. Moreover, property (\ref{enu:Coding_Spectrum_Type})
asserts $\gamma(k)=G^{(l)}$. Thus, $\gamma'=(\gamma(0),\ldots,\gamma(k),B)\in\Sigma_{\co_{k+1}}$
by ($\Sigma$4) and $b_{V}(\gamma')=J$.

It is left to prove that $b_{V}:\Sigma_{\co_{k+1}}\to L_{\co_{k+1},V}$
satisfies (\ref{enu:Coding_Spectrum_order_preserving}). Let $\gamma,\eta\in\Sigma_{\co_{k+1}}$.
We need to treat two cases. If $\gamma(k)=\eta(k)$, then the equivalence
$\wo\lessdot\eta$ $\Leftrightarrow$ $b_{V}(\gamma)\precd b_{V}(\eta)$
follows from (\ref{eq:Coding_Spectrum_Order=000020Preserving}). If
$\gamma(k)\neq\eta(k)$, then there is a $0\leq l\leq k-1$ such that
$\gamma(j)=\eta(j)$ for all $j\leq l$ and $\gamma(l+1)\neq\eta(l+1)$.
Since $l+1\leq k$, the induction hypothesis and (\ref{enu:Coding_Spectrum_order_preserving})
yield the equivalence $\gamma'\lessdot\eta'$ $\Leftrightarrow$ $b_{V}(\gamma')\precd b_{V}(\eta')$
where $\gamma'=(\gamma(0),\ldots,\gamma(l+1))$ and $\eta'=(\eta(0),\ldots,\eta(l+1))$.
Thus, $b_{V}(\gamma')\cap b_{V}(\eta')=\emptyset$ holds. By property
(\ref{enu:Coding_Spectrum_Inclusion}), we have $b_{V}(\gamma)\subseteq b_{V}(\gamma')$
and $b_{V}(\eta)\subseteq b_{V}(\eta')$. Hence, the previous considerations
imply that $\wo\lessdot\eta$ if and only if $b_{V}(\gamma)\precd b_{V}(\eta)$
proving (\ref{enu:Coding_Spectrum_order_preserving}) for $b_{V}:\Sigma_{\co_{k+1}}\to L_{\co_{k+1},V}$.
\end{proof}
\begin{cor}
\label{cor:Image=000020b_V=000020of=000020codes=000020ending=000020with=000020A-B}Let
$\alpha\in[0,1]\setminus\Q$ with infinite continued fraction expansion
$\left(c_{k}\right)_{k=0}^{\infty}$ and $\co_{k}=[0,c_{0},c_{1},\ldots,c_{k}]\in\Co$
for $k\in\N_{0}$. For all $V>4$ and $k\in\N_{0}$, the image $b_{V}(\Scks)$
equals to $\set I{I\,\textrm{spectral band of}\,\sigma_{\co_{k}}(V)}$.
\end{cor}

\begin{proof}
This follows immediately from Proposition~\ref{prop:Coding_Spectrum}
and Theorem~\ref{Thm-V>4_AllTypeA_B} asserting that every spectral
band in $\sigma_{\co_{k}}(V)$ is either of type $A$ or $B$ if $V>4$.
\end{proof}
Now we can use the previous considerations, to assign to each infinite
code in $\Sigma_{\alpha}$ an element in $\sigma\left(\Ham\right)$.
\begin{lem}
\label{lem:Code2Energy} Let $V>4$ and $\alpha\in[0,1]\setminus\Q$.
For all $\wo\in\Sigma_{\alpha}$, the set $\bigcap_{k\in\N_{0}}b_{V}(\wo|_{[0,k]})$
contains exactly one element and $\bigcap_{k\in\N_{0}}b_{V}(\wo|_{[0,k]})\subseteq\sigma\left(\Ham\right)$.
\end{lem}

\begin{proof}
Consider the sequence $\left\{ b_{V}(\wo|_{[0,k]})\right\} _{k\in\N_{0}}$
of intervals. This is a decreasing nested sequence of non-empty closed
intervals, see Proposition~\ref{prop:Coding_Spectrum}~(\ref{enu:Coding_Spectrum_Inclusion}).
Applying Cantor intersection theorem yields that $\bigcap_{k\in\N_{0}}b_{V}(\wo|_{[0,k]})$
is non-empty. Furthermore, it must be closed and convex (as intersection
of closed and convex sets). Hence, it may be either an interval or
a single point. Lemma~\ref{lem:LIsACover} asserts $b_{V}(\wo|_{[0,k]})\subset\bigcup_{I\in L_{\co_{k},V}}I=\Sigma_{k}(V)$
and 
\[
\bigcap_{k\in\N_{0}}b_{V}(\wo|_{[0,k]})\subset\bigcap_{k\in\N_{0}}\Sigma_{k}(V)=\sigma\left(\Ham\right).
\]
 According to \cite{BIST89}, $\sigma\left(\Ham\right)$ is of Lebesgue
measure zero if $V\neq0$. Thus, $\sigma\left(\Ham\right)$ cannot
contain an interval, and therefore $\bigcap_{k\in\N_{0}}b_{V}(\wo|_{[0,k]})$
is a single point (which is contained in $\sigma\left(\Ham\right)$).
\end{proof}
A consequence of this lemma is, that we now get a well-defined map
$\emap:\Sigma_{\alpha}\to\sigma(\Ham)$ by setting $\emap(\wo)$ to
be the unique element in $\bigcap_{k\in\N_{0}}b_{V}(\wo|_{[0,k]})$,
which exists by Lemma~\ref{lem:Code2Energy}.
\begin{lem}
\label{lem:EIsABijection} Let $V>4$ and $\alpha\in[0,1]\setminus\Q$.
Then the map $\emap:\Sigma_{\alpha}\to\sigma(\Ham)$ is a bijection.
\end{lem}

\begin{proof}
Let $\varphi(\co_{k}),k\in\Nz$ be the convergents of $\alpha$. For
each $k\in\N$, we have $\sigma(\Ham)\subseteq\bigcup_{I\in\L_{\co_{k},V}}I$
by Lemma~\ref{lem:LIsACover}. Furthermore, if $E\in\sigma(\Ham)$
is fixed, then for each $k\in\N$, there exists a unique $I\in\L_{\co_{k},V}$
such that $E\in I$, since the spectral bands in $\L_{\co_{k},V}$
are disjoint. By Proposition~\ref{prop:Coding_Spectrum}, \textbf{$b_{V}:\Sigma_{\co_{k}}\to L_{\co_{k},V}$}
is a bijection for each $k\in\N$. Thus, there exists a unique $\wo_{k}\in\Sigma_{\co_{k}}$
for each $k\in\N$ such that $E\in b_{V}(\wo_{k})$. Then $E\in b_{V}(\gamma_{k})\cap b_{V}(\gamma_{k+1})\neq\emptyset$
follows. Thus, Proposition~\ref{prop:Coding_Spectrum}~(\ref{enu:Coding_Spectrum_Inclusion})
asserts $b_{V}(\gamma_{k+1})\subseteq b_{V}(\gamma_{k})$ and the
codes $\gamma_{k+1}$ and $\gamma_{k}$ coincides on the first $k+1$
digits. Hence, we inductively conclude for all $j\geq k$, $b_{V}(\gamma_{j})\subseteq b_{V}(\gamma_{k})$
and the codes $\gamma_{j}\in\Sigma_{j}$ and $\gamma_{k}\in\Sigma_{k}$
coincides on the first $k+1$ digits. Since $k$ and $j$ were arbitrary,
there is a unique $\wo\in\Sigma_{\alpha}$ such that $\wo$ and $\wo_{k}$
have the same first $k+1$ digits for all $k\in\N$. We claim $E_{\alpha,V}(\wo)=E$.
By definition of $E_{\alpha,V}$, we get $E_{\alpha,V}(\wo)\in\bigcap_{k\in\N_{0}}b_{V}(\wo_{k})$.
On the other hand, also $E\in\bigcap_{k\in\N_{0}}b_{V}(\wo_{k})$
follows from our choice of $\wo_{k}$. The uniqueness from Lemma~\ref{lem:Code2Energy}
then yields $E_{\alpha,V}(\wo)=E$.
\end{proof}
\begin{lem}
\label{lem:EIsOrderPreserving} Let $V>4$ and $\alpha\in[0,1]\setminus\Q$,
then the map $\emap:\Sigma_{\alpha}\to\sigma(\Ham)$ is order preserving,
i.e. if $\wo,\eta\in\Sigma_{\alpha}$, then 
\[
\wo\lessdot\eta\quad\Leftrightarrow\quad E_{\alpha,V}(\wo)<E_{\alpha,V}(\eta).
\]
\end{lem}

\begin{proof}
Let $\varphi(\co_{k}),k\in\Nz$ be the convergents of $\alpha$. Also
let $\wo,\eta\in\Sigma_{\alpha}$ with $\wo\lessdot\eta$. Then there
is some $k\in\N_{0}$ such that $\wo|_{[0,k]}\lessdot\eta|_{[0,k]}.$
Thus, Proposition~\ref{prop:Coding_Spectrum}~(\ref{enu:Coding_Spectrum_order_preserving})
leads to 
\[
E_{\alpha,V}(\wo)\in b_{V}(\wo|_{[0,k]})\precd b_{V}(\eta|_{[0,k]})\ni E_{\alpha,V}(\eta),
\]
implying $E_{\alpha,V}(\wo)<E_{\alpha,V}(\eta)$.

Conversely, suppose $E_{\alpha,V}(\wo)<E_{\alpha,V}(\eta)$, then
$\gamma\neq\eta$ follows by Lemma~\ref{lem:EIsABijection}. Thus,
there exists a $k_{0}\in\N_{0}$ such that for all $k<k_{0}$, $\wo|_{[0,k]}=\eta|_{[0,k]}$
and $\wo(k_{0})\neq\eta(k_{0})$. Note that $\wo(k_{0})\neq B\neq\eta(k_{0})$.
Hence, either $\wo(k_{0})\lessdot\eta(k_{0})$ or $\eta(k_{0})\lessdot\wo(k_{0})$.
Since 
\[
b_{V}\left(\wo|_{[0,k_{0}]}\right)\ni E_{\alpha,V}(\wo)<E_{\alpha,V}(\eta)\in b_{V}\left(\eta|_{[0,k_{0}]}\right),
\]
we conclude $\wo|_{[0,k_{0}]}\lessdot\eta|_{[0,k_{0}]}$ from Lemma~\ref{lem:PrecIsTotal}
Proposition~\ref{prop:Coding_Spectrum}. Hence, $\gamma\lessdot\eta$
follows.
\end{proof}
\begin{lem}
\label{lem:CodesWhichAreSpecBands} Let $V>4$ and $\alpha\in[0,1]\setminus\Q$
with convergents $\varphi(\co_{k}),k\in\N_{0}$. Furthermore, let
$\wo\in\Sigma_{\alpha}$ with $\emap(\wo)=:E$. Then for all $k\in\N$,
the image of $\set{\eta\in\Scks}{\eta\lessdot\wo|_{[0,k]}}$ under
$b_{V}$ equals $\set I{I\text{ is a spectral band of }\sigma_{\co_{k}}(V)\text{ with }I\precd\{E\}}$.
\end{lem}

\begin{proof}
Assume $\eta\in\Scks$ with $\eta\lessdot\wo|_{[0,k]}$. Then $b_{V}(\eta)$
is a spectral band of $\sigma_{\co_{k}}$ with $b_{V}(\eta)\precd b_{V}(\wo|_{[0,k]}).$
In particular, $b_{V}(\eta)\cap b_{V}(\wo|_{[0,k]})=\emptyset$ follows
from Proposition~\ref{prop:Coding_Spectrum}~(\ref{enu:Coding_Spectrum_Inclusion})
and (\ref{enu:Coding_Spectrum_order_preserving}). By construction
of the map $\emap$, we have $E\in b_{V}(\wo|_{[0,k]})$, which shows
\[
b_{V}\left(\set{\eta\in\Scks}{\eta\lessdot\wo|_{[0,k]}}\right)\subset\set I{I\text{ is a spectral band of }\sigma_{\co_{k}}\text{ with }I\precd\{E\}}.
\]
To show the other inclusion, consider a spectral band $I$ of $\sigma_{\co_{k}}$
with $I\precd\{E\}$. We apply Proposition~\ref{prop:Coding_Spectrum}
to get $\eta\in\Scks$ such that $b_{V}(\eta)=I$ and $b_{V}(\eta)\cap b_{V}(\wo|_{[0,k]})=\emptyset$.
Combining $E\in b_{V}(\wo|_{[0,k]})$ and $I\precd\{E\}$, we get
that $b_{V}(\eta)\precd b_{V}(\wo|_{[0,k]})$. Hence, by Proposition~\ref{prop:Coding_Spectrum}
$\eta\lessdot\wo|_{[0,k]}$ which finishes the proof.
\end{proof}
As a consequence of the considerations of this subsection, we conclude
with the following statement.
\begin{prop}
\label{prop:IDS_Code=000020representation}Let $V>4$ and $\alpha\in[0,1]\setminus\Q$
with convergents $\frac{p_{k}}{q_{k}}=\varphi(\co_{k}),k\in\N_{0}$,
where $p_{k},q_{k}$ are coprime. For each $E\in\sigma(\Ham)$, there
is a unique $\gamma\in\Sigma_{\alpha}$ such that $E=E_{\alpha,V}(\gamma)$
and 
\begin{equation}
N_{\alpha,V}(E)=N_{\alpha,V}(E_{\alpha,V}(\gamma))=\lim_{k\to\infty}\frac{\#\set{\eta\in\Scks}{\eta\lessdot\wo|_{[0,k]}}}{q_{k}}.\label{eq:=000020IDS=000020via=000020coding}
\end{equation}
\end{prop}

\begin{rem*}
Note that the latter statement asserts that the value $N_{\alpha,V}(E_{\alpha,V}(\gamma))$
is independent of $V>4$ as the limit on the right hand side is so.
In fact, the value is independent for all $V>0$ as proven in \cite[thm. 1.9 (d)]{BanBecLoe_24}.
\end{rem*}
\begin{proof}
Let $E\in\sigma(\Ham)$. The existence of a unique $\gamma\in\Sigma_{\alpha}$
such that $E=E_{\alpha,V}(\gamma)$ is proven in Lemma~\ref{lem:EIsABijection}.
Thus, $N_{\alpha,V}(E)=N_{\alpha,V}(E_{\alpha,V}(\gamma))$ holds.
Proposition~\ref{prop:IDSBandCounting} leads to 
\[
N_{\alpha,V}(E)=\lim_{k\to\infty}\frac{\#\set I{I\text{ is a spectral band of }\sigma(H_{\pqk,V})\text{ with }I\precd\{E\}}}{q_{k}}.
\]
First note that $\sigma(H_{\pqk,V})=\sigma_{\co_{k}}(V)$ by Proposition~\ref{prop:=000020Floquet-Bloch=000020via=000020transfer=000020matrix}.
Let $b_{V}:\bigsqcup_{k\in\N_{0}}\Sigma_{\co_{k}}\to\bigcup_{k\in\N_{0}}L_{\co_{k},V}$
be the map defined in Proposition~\ref{prop:Coding_Spectrum} satisfying
$E=E_{\alpha,V}(\gamma)\in b_{V}(\wo|_{[0,k]})$ for all $k\in\N_{0}$.
Thus, Lemma~\ref{lem:CodesWhichAreSpecBands} leads to
\begin{align*}
b_{V}\left(\set{\eta\in\Scks}{\eta\lessdot\wo|_{[0,k]}}\right) & =\set I{I\subseteq\sigma_{\co_{k}}(V)\text{ spectral band with }I\precd\left\{ E\right\} }.
\end{align*}
Hence, Proposition~\ref{prop:Coding_Spectrum}~(\ref{enu:Coding_Spectrum_bijective})
(asserting that $b_{V}$ is injective if restricted to $\Scks$) implies
\[
\sharp\set{\eta\in\Scks}{\eta\lessdot\wo|_{[0,k]}}=\sharp\set I{I\subseteq\sigma_{\co_{k}}(V)\text{ spectral band with }I\precd\left\{ E\right\} }
\]
finishing the proof.
\end{proof}

\subsection{A Formula for the IDS via the spectral coding\protect\label{subsec:=000020formula=000020for=000020IDS}}

In this subsection we use the hierarchical structure of the periodic
spectra and its coding in order to provide an explicit formula for
the IDS, $N_{\alpha,V}$. Proposition~\ref{prop:IDS_Code=000020representation}
is the starting point for the current subsection. Next, we provide
some counting arguments in order to express the numerator in (\ref{eq:=000020IDS=000020via=000020coding})
and to obtain a convenient formula for the IDS, which is eventually
proven in Proposition~\ref{prop:IDSFormula}.
\begin{example}
\label{exa:CountingCodes} We provide a guiding example to demonstrate
some of the counting arguments which are developed in this subsection.
Let $\alpha\notin\Q$ whose continued fraction expansion that starts
with $(2,1,1,2,\dots)$ and consider for example $\wo=(G^{(2)},B,G^{(2)},B,G^{(2)},\dots)\in\Sigma_{\alpha}$.
We would like to compute the IDS at $E:=E_{\alpha,V}(\wo)$ using
the sequence in Proposition~\ref{prop:IDSBandCounting}. Here we
show how to compute the $k=4$ element of this sequence. First observe
that $\frac{5}{13}=\frac{p_{k}}{q_{k}}=\varphi([0,0,2,1,1,2])$. Since
bands and codes are in a one-to-one relation by Proposition~\ref{prop:Coding_Spectrum},
we can rather think of codes and count the number of codes $\eta\in\Sigma_{\co_{4}}^{\textrm{spec}}$
of length $5$ with $\eta\lessdot\wo|_{[0,4]}$. We demonstrate this
situation in Figure~\ref{fig:CountingExample} - to better illustrate
the example, we adapt here the tree formalism from \cite{BanBecLoe_24},
even though it did not originally appear in \cite{Raym95}.

\begin{figure}[h]
\includegraphics[width=1.05\textwidth]{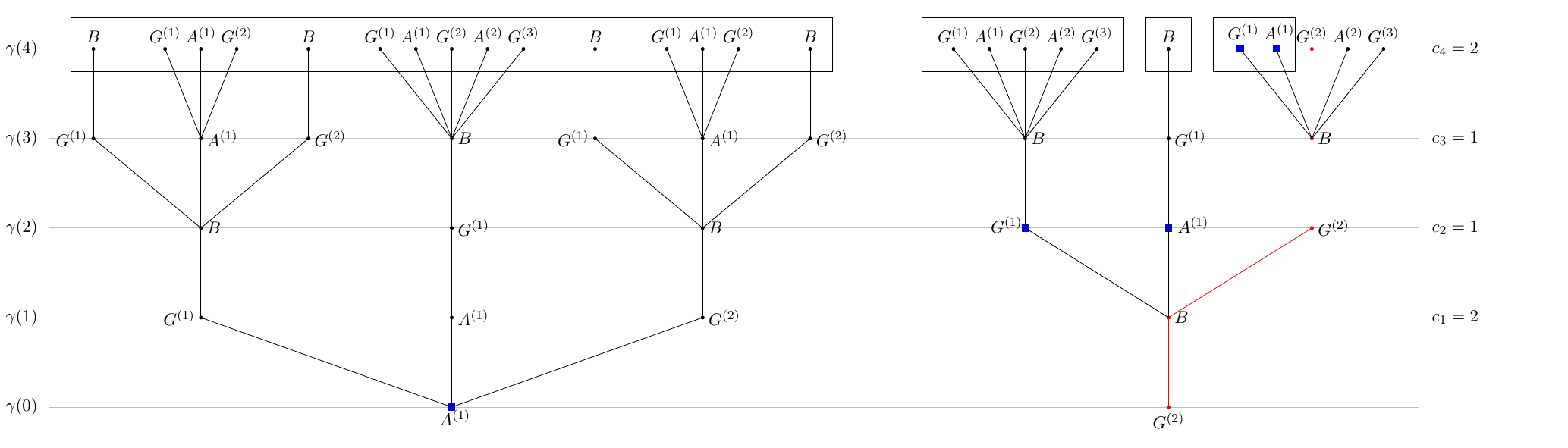} \caption{Visualization of the codes in Example~\ref{exa:CountingCodes} as
an ordered graph. The rectangles mark all the vertices in level $4$
descending from a blue marked vertex in level $j\in\left\{ 0,1,2,3,4\right\} $.
}
\label{fig:CountingExample}
\end{figure}

The beginning of the code $\wo$ is marked red and we need to count
the codes $\eta=(\eta(0),\ldots,\eta(4))\in\Sigma_{\co_{4}}^{\textrm{spec}}$
which correspond to spectral bands and which are to the left of this
path $\gamma$. Specifically, one needs to count the vertices in the
fourth level which are inside the rectangles, but only those vertices
that are labeled $A^{(i)}$ or $B$ should be counted since $\eta\in\Sigma_{\co_{4}}^{\textrm{spec}}$.
To do so, we follow the red path and at each level $j\in\left\{ 0,1,2,3,4\right\} $
we mark the vertices that branch off to the left. In Figure~\ref{fig:CountingExample},
we marked them with blue squares. The set of paths ending at a blue
square in level $j$ with label $\star=A$ if $\eta(j)=A^{(i)}$ and
$\star=G$ if $\eta(j)=G$ are denoted by $\Gamma_{j}(\gamma,\star)$.
For instance, $\Gamma_{0}(\gamma,A)=\left\{ (A^{(1)})\right\} $ and
$\Gamma_{0}(\gamma,G)=\emptyset$. Then we use the evolution laws
$(\Sigma2),(\Sigma3)$ and $(\Sigma4)$ to calculate how many codes
of length 5 are descendants of these blue squares and end with an
$A^{(i)}$ or an $B$. This number is denoted by $d_{j}^{4}(\star)$,
see Definition~\ref{def:=000020d_j^k}. For instance, $d_{0}^{4}(A)=8$
and $d_{1}^{4}(G)=3$. Then $d_{j}^{4}(\star)\cdot\sharp\paths_{j}(\w,\star)$
is the total number of codes $\eta\in\Sigma_{\co_{4}}^{\textrm{spec}}$
where $\eta(j)$ has the label $\star$ and $\eta\lessdot\wo|_{[0,4]}$
since $\eta(j)\lessdot\wo(j)$. Let $\sharp D_{j}^{4}(\gamma)=d_{j}^{4}(A)\cdot\sharp\paths_{j}(\w,A)+d_{j}^{4}(G)\cdot\sharp\paths_{j}(\w,G)$
be the sum of these numbers for the different labels $\star\in\left\{ A,G\right\} $,
see Lemma~\ref{lem:CountingFormalism} and eq. (\ref{eq:=000020Formula=000020for=000020Ds}).
For this specific example, we now can check directly in Figure~\ref{fig:CountingExample}
that
\begin{align*}
\sharp\Des_{0}^{4}(\wo) & =8\cdot1+5\cdot0 & =8,\\
\sharp\Des_{1}^{4}(\wo) & =2\cdot0+3\cdot0 & =0,\\
\sharp\Des_{2}^{4}(\wo) & =1\cdot1+2\cdot1 & =3,\\
\sharp\Des_{3}^{4}(\wo) & =1\cdot0+1\cdot0 & =0,\\
\sharp\Des_{4}^{4}(\wo) & =1\cdot1+0\cdot1 & =1,
\end{align*}
and $\sharp\set{\eta\in\Sigma_{\co_{4}}^{\textrm{spec}}}{\eta\lessdot\wo|_{[0,4]}}=12$
coinciding with the sum of $\sharp\Des_{j}^{4}(\wo)$. Thus,
\[
\frac{\#\set{\eta\in\Sigma_{\co_{4}}^{\textrm{spec}}}{\eta\lessdot\wo|_{[0,4]}}}{q_{4}}=\frac{12}{13}
\]
follows.
\end{example}

In this subsection, we perform this counting in a general manner and
the final result is given in Proposition~\ref{prop:IDSFormula}.
Let $\alpha\in[0,1]\setminus\Q$ and $\wo\in\Sigma_{\alpha}$ an infinite
spectral-$\alpha$-code. Define: 
\[
\Des_{j}^{k}(\wo):=\set{\wo'\in\Scks}{\w'(j)\lessdot\w(j)\text{ and }\forall i<j,~\wo'(i)=\wo(i)},
\]
for all $0\leq j\leq k$. Notice that in particular, 
\[
\Des_{0}^{k}(\wo)=\set{\wo'\in\Scks}{\w'(0)\lessdot\w(0)},
\]
and for example $\Des_{0}^{k}(\wo)=\emptyset$ if $\wo(0)=A^{(1)}$.

We intend to employ the sets $\Des_{j}^{k}(\wo)$ in order to compute
the numerator in (\ref{eq:=000020IDS=000020via=000020coding}), see
e.g. (\ref{enu:=000020lem-CountingFormalism=000020-=0000201}) in
the following lemma.
\begin{lem}
\label{lem:CountingFormalism} Let $V>4,\alpha\in[0,1]\setminus\Q$
with convergents $\varphi(\co_{k}),\,k\in\N_{0},$ and $\wo\in\Sigma_{\alpha}$.
Then the following assertions hold.
\end{lem}

\begin{enumerate}
\item \label{enu:=000020lem-CountingFormalism=000020-=0000201}We have 
\[
\set{\eta\in\Scks}{\eta\lessdot\wo|_{[0,k]}}=\bigsqcup_{j=0}^{k}\Des_{j}^{k}(\wo).
\]
\item \label{enu:=000020lem-CountingFormalism=000020-=0000202}Let $\res_{j}^{k}:\Sck\to\Sigma_{\co_{j}},\wo\mapsto\res_{j}^{k}(\wo)=\wo|_{[0,j]}$
be the restriction map. Then
\[
\sharp\Des_{j}^{k}(\wo)=\sum_{\eta\in\res_{j}^{k}(\Des_{j}^{k}(\wo))}\#\left((\res_{j}^{k})^{-1}(\eta)\cap\Scks\right).
\]
\item \label{enu:=000020lem-CountingFormalism=000020-=0000203}If $\eta\in\Des_{j}^{k}(\wo)$,
then $\eta(j)\neq B\neq\w(j)$. In particular $\Des_{j}^{k}=\emptyset$
if $\w(j)=B$.
\item \label{enu:=000020lem-CountingFormalism=000020-=0000204}For any $0\leq j\leq k$,
we have $\res_{j}^{k}(\Des_{j}^{k}(\wo))=\paths_{j}(\wo,A)\sqcup\paths_{j}(\wo,G),$
where 
\begin{align*}
\paths_{j}(\w,A) & :=\left\{ \eta\in\Sigma_{\co_{j}}:\eta\lessdot\wo|_{[0,j]},\,\eta|_{[0,j-1]}=\wo|_{[0,j-1]},\,\eta(j)\in\{A^{(i)}\colon i\in\N\}\right\} ,\\
\paths_{j}(\w,G) & :=\left\{ \eta\in\Sigma_{\co_{j}}:\eta\lessdot\wo|_{[0,j]},\,\eta|_{[0,j-1]}=\wo|_{[0,j-1]},\,\eta(j)\in\{G^{(i)}\colon i\in\N\}\right\} .
\end{align*}
\end{enumerate}
\begin{proof}
(\ref{enu:=000020lem-CountingFormalism=000020-=0000201}) Let $\eta\in\Des_{j}^{k}(\wo)$
for $0\leq j\leq k$. By construction of $\lessdot$ on $\Scks$,
$\eta(j)\lessdot\w(j)$ implies $\eta\lessdot\wo|_{[0,k]}$. Thus
\[
\Des_{j}^{k}(\wo)\subseteq\{\wo'\in\Sigma_{\co_{k}}^{s}:\eta\lessdot\wo|_{[0,k]}\},
\]
follows for all $0\leq j\leq k$. For the converse inclusion observe
that if $\eta\in\Scks$ satisfies $\eta\lessdot\wo|_{[0,k]},$ then
either $\eta(0)\lessdot\w(0)$, or there is some $0\leq j\leq k$
such that 
\[
\eta(j)\lessdot\w(j)\quad\text{ and }\quad\eta(i)=\w(i)\text{ for all }0\leq i<j.
\]

(\ref{enu:=000020lem-CountingFormalism=000020-=0000202}) Let $\eta\in\res_{j}^{k}(\Des_{j}^{k}(\wo))$,
so we have $\eta(j)\lessdot\w(j)$. Thus, if $\tilde{\wo}\in(\res_{j}^{k})^{-1}(\{\eta\})$,
then $\tilde{\wo}|_{[0,j]}=\eta$ holds and so $\tilde{\wo}\lessdot\wo|_{[0,k]}$.
If, additionally, $\tilde{\wo}\in\Scks$, then we $\tilde{\wo}\in\Des_{j}^{k}(\wo)$.
With this the claim follows.

(\ref{enu:=000020lem-CountingFormalism=000020-=0000203}) This follows
$(\Sigma4)$ in Definition~\ref{def:SpectralCodes} asserting that
if $\w(j)=B$, then $\w(j-1)=G^{(i)}$ for some $i\in\N$. Therefore,
there is no code $\eta\in\Sigma_{\co_{j}}$ with $\eta(j-1)=\w(j-1),\eta(j)\lessdot\w(j)$
and $\eta(j)=B$ or $\w(j)=B$.

(\ref{enu:=000020lem-CountingFormalism=000020-=0000204}) This follows
directly from (\ref{enu:=000020lem-CountingFormalism=000020-=0000203}).
\end{proof}
To continue the counting arguments, we find it useful to partition
finite codes according to their finite letter. Towards this we introduce
another notation. Let $\alpha\in[0,1]\setminus\Q$ with continued
fraction expansion $\left(c_{k}\right)_{k=0}^{\infty}$ and $\co_{k}=[0,c_{0},\dots,c_{k}]\in\Co$.
Consider a subset $\Ee\subseteq\Sigma_{\co_{k}}$. Define

\begin{align*}
\ext(\Ee,A) & :=\#\set{\eta\in\Ee}{\eta(k)=A^{(i)}\text{ for some }i\in\N},\\
\ext(\Ee,G) & :=\#\set{\eta\in\Ee}{\eta(k)=G^{(i)}\text{ for some }i\in\N},\\
\ext(\Ee,B) & :=\#\set{\eta\in\Ee}{\eta(k)=B}.
\end{align*}

For instance, $\ext(\Ee,A)$ is the number of those spectral-$\alpha$-codes
in $\Ee$ which end on an $A^{(i)}$. Let $l\in\N$. Then $\ext((\res_{k}^{k+l})^{-1}(\Ee),A)$
is the number of those codes in $\Sigma_{\co_{k+l}}$ that are extensions
of codes in $\Ee$ and end with a letter $A^{(i)}$. The next lemma
provides identities for counting the number of such code extensions.
Therefore, recall that the convergents $\pqk=\varphi(\co_{k}),\,k\in\N_{0},$
of $\alpha\in[0,1]\setminus\Q$ with $p_{k},q_{k}$ coprime satisfy
the recursive relation (\ref{eq:=000020recursion=000020for=000020p_k,=000020q_k}).
\begin{lem}
[{\cite[prop. 4.1]{Raym95}}]\label{lem:counting1} Let $\alpha\in[0,1]\setminus\Q$
with continued fraction expansion $\left(c_{k}\right)_{k=0}^{\infty}$
and convergents $\pqk=\varphi(\co_{k}),\,k\in\N_{0},$ where $p_{k},q_{k}$
are coprime.
\begin{enumerate}
\item \label{enu:counting1_a_one=000020step=000020transfer=000020matrices}Let
$\Ee\subseteq\Sigma_{\co_{k}},k\in\N_{0}$. Then 
\[
\left(\begin{array}{c}
\ext\left(\left(\res_{k}^{k+1}\right)^{-1}\left(\Ee\right),G\right)\\
\ext\left(\left(\res_{k}^{k+1}\right)^{-1}\left(\Ee\right),B\right)\\
\ext\left(\left(\res_{k}^{k+1}\right)^{-1}\left(\Ee\right),A\right)
\end{array}\right)=T_{k+1}\left(\begin{array}{c}
\ext\left(\Ee,G\right)\\
\ext\left(\Ee,B\right)\\
\ext\left(\Ee,A\right)
\end{array}\right),
\]
where 
\[
T_{k+1}:=T_{k+1}\left(\co_{k+1}\right):=\left(\begin{array}{ccc}
1 & c_{k+1}+1 & c_{k+1}\\
1 & 0 & 0\\
0 & c_{k+1} & c_{k+1}-1
\end{array}\right).
\]
\item \label{enu:counting1_b_starting=000020from=000020l=000020to=000020k}If
$\ell\in\N$, then 
\[
\left(\begin{array}{c}
\ext\left(\left(\res_{k}^{k+\ell}\right)^{-1}\left(\Ee\right),G\right)\\
\ext\left(\left(\res_{k}^{k+\ell}\right)^{-1}\left(\Ee\right),B\right)\\
\ext\left(\left(\res_{k}^{k+\ell}\right)^{-1}\left(\Ee\right),A\right)
\end{array}\right)=S_{k+\ell}S_{k}^{-1}\left(\begin{array}{c}
\ext\left(\Ee,G\right)\\
\ext\left(\Ee,B\right)\\
\ext\left(\Ee,A\right)
\end{array}\right).
\]
with $S_{k}:=S_{k}\left(\co_{k}\right):=T_{k}T_{k-1}\dots T_{1}$
and $S_{0}=\Id$ the identity matrix.
\item \label{enu:counting1_a_computing=000020S_k=000020explicitly}The matrix
$S_{k}$ is given by 
\[
S_{k}=\left(\begin{array}{ccc}
p_{k}+\left(-1\right)^{k} & q_{k}-\left(-1\right)^{k} & q_{k}-p_{k}-\left(-1\right)^{k}\\
p_{k-1}-\left(-1\right)^{k} & q_{k-1}+\left(-1\right)^{k} & q_{k-1}-p_{k-1}+\left(-1\right)^{k}\\
p_{k}-p_{k-1}+\left(-1\right)^{k} & q_{k}-q_{k-1}-\left(-1\right)^{k} & q_{k}-q_{k-1}-p_{k}+p_{k-1}-\left(-1\right)^{k}
\end{array}\right)
\]
and its inverse $S_{k}^{-1}$ equals to
\[
S_{k}^{-1}=\left(-1\right)^{k}\left(\begin{array}{ccc}
1-p_{k} & p_{k-1}-1 & p_{k-1}+p_{k}-1\\
q_{k}-1 & 1-q_{k-1} & 1-q_{k}-q_{k-1}\\
1-p_{k}-q_{k} & q_{k-1}+p_{k-1}-1 & p_{k}+p_{k-1}+q_{k}+q_{k-1}-1
\end{array}\right).
\]
\end{enumerate}
\end{lem}

\begin{proof}
(\ref{enu:counting1_a_one=000020step=000020transfer=000020matrices})
This is a direct consequence of the defining properties $(\Sigma2)$
to $(\Sigma4)$.

(\ref{enu:counting1_b_starting=000020from=000020l=000020to=000020k})
This is proven via induction over $\ell\in\N$. For $\ell=1$ this
statement is just part (\ref{enu:counting1_a_one=000020step=000020transfer=000020matrices})
of this Lemma. Now assume the statement holds for an arbitrary $\ell\in\N$.
By observing $\res_{k}^{k+\ell+1}=\res_{k}^{k+\ell}\circ\res_{k+\ell}^{k+\ell+1}$
we get 
\[
\left(\res_{k}^{k+\ell+1}\right)^{-1}\left(\Ee\right)=\left(\res_{k+\ell}^{k+\ell+1}\right)^{-1}\left(\left(\res_{k}^{k+\ell}\right)^{-1}\left(\Ee\right)\right).
\]
Using (\ref{enu:counting1_a_one=000020step=000020transfer=000020matrices})
and our induction hypothesis then yields
\begin{align*}
\vect{\ext\left(\left(\res_{k}^{k+\ell+1}\right)^{-1}\left(\Ee\right),G\right)} & =\vect{\ext\left(\left(\res_{k+\ell}^{k+\ell+1}\right)^{-1}\left(\left(\res_{k}^{k+\ell}\right)^{-1}\left(\Ee\right)\right),G\right)}\\
 & =T_{k+\ell+1}\vect{\ext\left(\left(\res_{k}^{k+\ell}\right)^{-1}\left(\Ee\right),G\right)}\\
 & =T_{k+\ell+1}S_{k+\ell}S_{k}^{-1}\vect{\ext\left(\Ee,G\right)}\\
 & =S_{k+\ell+1}S_{k}^{-1}\vect{\ext\left(\Ee,G\right)}.
\end{align*}
(\ref{enu:counting1_a_computing=000020S_k=000020explicitly}) The
stated form for $S_{k}$ can be computed inductively. Then one checks
by direct computations that $S_{k}^{-1}$ is given by the stated matrix.
Occasionally, the formula $q_{k}p_{k-1}-p_{k}q_{k-1}=(-1)^{k}$ (see
\cite[thm. 2]{Khinchin_book64}) is used. We leave the computational
details to the reader.
\end{proof}
Let $\wo\in\Sigma_{\co_{j}}$ such that $\w(j)=A^{(i)}$ for some
$i$. Then, for $k>j$, we wish to count how many $\eta\in\Scks$
there are such that $\eta|_{[0,j]}=\wo$. By the previous lemma, this
number depends only on $j$ and $k$, but it does not depend on the
particular $\wo\in\Sigma_{\co_{j}}$ satisfying $\w(j)=A^{(i)}$.
Indeed, using Lemma~\ref{lem:counting1} we define:
\begin{defn}
\label{def:=000020d_j^k}For $\alpha\in[0,1]\setminus\Q$, define
\[
d_{j}^{k}(A):=\left(\begin{array}{ccc}
0 & 1 & 1\end{array}\right)S_{k}S_{j}^{-1}\left(\begin{array}{c}
0\\
0\\
1
\end{array}\right),
\]
and 
\[
d_{j}^{k}(G):=\left(\begin{array}{ccc}
0 & 1 & 1\end{array}\right)S_{k}S_{j}^{-1}\left(\begin{array}{c}
1\\
0\\
0
\end{array}\right).
\]
\end{defn}

With this definition we may express the number of elements in the
sets $\Des_{j}^{k}(\wo)$ (which are used in Lemma~\ref{lem:CountingFormalism}~(\ref{enu:=000020lem-CountingFormalism=000020-=0000201}))
\begin{equation}
\#\Des_{j}^{k}(\wo)=d_{j}^{k}(A)\cdot\#\paths_{j}(\wo,A)+d_{j}^{k}(G)\cdot\#\paths_{j}(\wo,G),\label{eq:=000020Formula=000020for=000020Ds}
\end{equation}
where the sets $\paths_{j}(\wo,A)$, $\paths_{j}(\wo,G)$ were defined
in Lemma~\ref{lem:CountingFormalism}. To verify identity (\ref{eq:=000020Formula=000020for=000020Ds})
one first observes that the set $\Des_{j}^{k}(\wo)$ may be decomposed
into codes $\eta\in\Des_{j}^{k}(\wo)$ for which $\eta(j)=A^{(i)}$
(for some $i$) and codes $\eta\in\Des_{j}^{k}(\wo)$ for which $\eta(j)=G^{(i)}$
(for some $i$). This decomposition is thorough, since there are no
codes $\eta\in\Des_{j}^{k}(\wo)$ for which $\eta(j)=B$, see Lemma~\ref{lem:CountingFormalism}~(\ref{enu:=000020lem-CountingFormalism=000020-=0000203}).
To count the codes $\eta\in\Des_{j}^{k}(\wo)$ for which $\eta(j)=A^{(i)}$,
we notice that the prefix of each such code, $\left(\eta(0),\ldots,\eta(j-1),A^{(i)}\right)$
belongs to $\paths_{j}(\wo,A)$ and there are exactly $d_{j}^{k}(A)$
ways to extend such a prefix to get an element in $\Des_{j}^{k}(\wo)$.

Next, we provide an explicit formula for $d_{j}^{k}(A)$ and $d_{j}^{k}(G)$
using the convergents $\pqk=\varphi(\co_{k}),\,k\in\N_{0},$ of $\alpha$
with $p_{k},q_{k}$ coprime. Therefore, we like to remind the reader
on the recursive definition of $\left\{ p_{k}\right\} _{k=-1}^{\infty}$
and $\left\{ q_{k}\right\} _{k=-1}^{\infty}$ in (\ref{eq:=000020recursion=000020for=000020p_k,=000020q_k})
with initial condition
\[
p_{-1}=1,\quad p_{0}=0,\quad q_{-1}=0,\quad q_{0}=1.
\]

\begin{lem}
\label{lem:counting2} Let $\alpha\in[0,1]\setminus\Q$ with convergents
$\pqk=\varphi(\co_{k}),\,k\in\N_{0},$ with $p_{k},q_{k}$ coprime.
Consider the numbers 
\[
\Pp_{j}^{k}:=(-1)^{j}[q_{j}p_{k}-p_{j}q_{k}]\quad\quad\textrm{for }-1\leq j\leq k.
\]
Then 
\[
d_{j}^{k}(A)=\Pp_{j-1}^{k}-\Pp_{j}^{k}\quad\text{ and }\quad d_{j}^{k}(G)=\Pp_{j}^{k}
\]
hold for $0\leq j\leq k$.
\end{lem}

\begin{proof}
We sketch the computation of $d_{j}^{k}(G)$. We have 
\begin{align*}
S_{k}S_{j}^{-1}\cdot\left(\begin{array}{c}
1\\
0\\
0
\end{array}\right) & =(-1)^{j}S_{k}\cdot\left(\begin{array}{c}
1-p_{j}\\
q_{j}-1\\
1-p_{j}-q_{j}
\end{array}\right).
\end{align*}
Performing this matrix multiplication and simplifying then yields
\begin{align*}
\left(\begin{array}{ccc}
0 & 1 & 0\end{array}\right)\cdot S_{k}S_{j}^{-1}\cdot\left(\begin{array}{c}
1\\
0\\
0
\end{array}\right)= & (-1)^{j}[q_{j}p_{k-1}-p_{j}q_{k-1}-(-1)^{k}],\\
\left(\begin{array}{ccc}
0 & 0 & 1\end{array}\right)\cdot S_{k}S_{j}^{-1}\cdot\left(\begin{array}{c}
1\\
0\\
0
\end{array}\right)= & (-1)^{j}[p_{j}q_{k-1}+q_{j}p_{k}-q_{j}p_{k-1}-p_{j}q_{k}+(-1)^{k}],
\end{align*}
and hence
\begin{align*}
d_{j}^{k}(G)=(-1)^{j}[q_{j}p_{k}-p_{j}q_{k}]=\Pp_{j}^{k}.
\end{align*}
To compute $d_{j}^{k}(A)$, one proceeds analogously.
\end{proof}
Now, we have collected all the pieces in order to provide the promised
formula for the IDS $N_{\alpha,V}$.
\begin{prop}
\label{prop:IDSFormula} Let $V>4,\alpha\in[0,1]\setminus\Q$ with
convergents $\pqk=\varphi(c_{k}),k\in\N_{0},$ with $p_{k},q_{k}$
coprime. Then, for each $\wo\in\Sigma_{\alpha}$, 
\[
\IDS(E_{\alpha,V}(\wo))=\sum_{k=-1}^{\infty}(-1)^{k}\ypi_{k}(\wo)(q_{k}\alpha-p_{k})
\]
where 
\[
\ypi_{-1}(\wo):=\#\paths_{0}(\wo,A)\quad\text{and}\quad\ypi_{k}(\wo):=\#\paths_{k}(\wo,G)-\#\paths_{k}(\wo,A)+\#\paths_{k+1}(\wo,A),\,k\in\N_{0}.
\]
\end{prop}

\begin{proof}
Let $\wo\in\Sigma_{\alpha}$. Our starting point is Proposition~\ref{prop:IDS_Code=000020representation},
to which we consequently apply Lemma~\ref{lem:CountingFormalism}~(\ref{enu:=000020lem-CountingFormalism=000020-=0000201}),
(\ref{eq:=000020Formula=000020for=000020Ds}) and Lemma~\ref{lem:counting2},
\begin{align*}
\IDS(E_{\alpha,V}(\wo))= & \lim_{k\to\infty}\frac{\#\set{\eta\in\Scks}{\eta\lessdot\wo}}{q_{k}}\\
= & \lim_{k\to\infty}\frac{\sum_{j=0}^{k}\#\Des_{j}^{k}(\wo)}{q_{k}}\\
= & \lim_{k\to\infty}\frac{1}{q_{k}}\sum_{j=0}^{k}\#\paths_{j}(\wo,A)\cdot d_{j}^{k}(A)+\#\paths_{j}(\wo,G)\cdot d_{j}^{k}(G)\\
= & \lim_{k\to\infty}\frac{1}{q_{k}}\sum_{j=0}^{k}\#\paths_{j}(\wo,A)\cdot(\Pp_{j-1}^{k}-\Pp_{j}^{k})+\#\paths_{j}(\wo,G)\cdot\Pp_{j}^{k}\\
= & \lim_{k\to\infty}\frac{1}{q_{k}}\sum_{j=0}^{k}\Pp_{j}^{k}\cdot\left(\#\paths_{j}(\wo,G)-\#\paths_{j}(\wo,A)+\#\paths_{j+1}(\wo,A)\right)\\
 & \quad+\frac{1}{q_{k}}\Pp_{-1}^{k}\cdot\#\paths_{0}(\wo,A)-\frac{1}{q_{k}}\underbrace{\Pp_{k}^{k}}_{=0}\cdot\#\paths_{k+1}(\wo,A)\\
= & \lim_{k\to\infty}\frac{1}{q_{k}}\sum_{j=-1}^{k}\Pp_{j}^{k}\cdot\ypi_{j}(\wo)\\
= & \lim_{k\to\infty}\sum_{j=-1}^{k}\underbrace{(-1)^{j}\ypi_{j}(\wo)\cdot\left(q_{j}\frac{p_{k}}{q_{k}}-p_{j}\right)}_{=:f_{k}(j)}.
\end{align*}
where in the last two equalities we substitute $\ypi_{j}(\wo)$ from
this proposition statement and the $\Pp_{j}^{k}$ from Lemma~\ref{lem:counting2}.

According to \cite[thm. 4]{Khinchin_book64}, the sequence $\left\{ \frac{p_{2l}}{q_{2l}}\right\} _{l=1}^{\infty}$
is monotonically increasing, $\left\{ \frac{p_{2l-1}}{q_{2l-1}}\right\} _{l=1}^{\infty}$
is monotonically decreasing and both sequences converge to $\alpha$.
Hence, we conclude for all $k\in\N$,

\[
(-1)^{j}\left(q_{j}\frac{p_{k}}{q_{k}}-p_{j}\right)\geq0\quad\textrm{ for all }j\leq k\qquad\textrm{and}\qquad(-1)^{k}\left(q_{k}\alpha-p_{k}\right)\geq0.
\]
Since $\mu_{j}(\gamma)\geq0$ if $j\geq1$, we conclude $f_{k}(j)\geq0$
for all $j\in\N$. Thus, $\lim_{k\to\infty}\sum_{j=-1}^{k}f_{k}(j)$
converges absolutely using that the limit exists. For each $k\in\N$,
we also have $f_{2k}(2l)$ is monotone increasing in $l\in\N$ and
$f_{2k}(2l-1)$ is monotone decreasing in $l\in\N$. Note also that
$\lim_{k\to\infty}f_{k}(j)=(-1)^{j}\ypi_{j}(\wo)\cdot\left(q_{j}\alpha-p_{j}\right)$
for each $j\geq-1$. Hence, the monotone convergence theorem applied
to the following two summands separately leads to 
\[
\IDS(E_{\alpha,V}(\wo))=\lim_{k\to\infty}\left(\sum_{j=0}^{2k}f_{2k}(2j-1)+\sum_{l=0}^{2k}f_{2k}(2l)\right)=\sum_{j=-1}^{\infty}(-1)^{j}\ypi_{j}(\wo)\cdot\left(q_{j}\alpha-p_{j}\right)
\]

where we used in the first step that the limit exists and so we can
pass to the subsequence of even numbers $2k$.
\end{proof}
\begin{rem}
\label{rem:PiComputation} Recognizing the importance of the functions
$\ypi_{k}(\wo),\,k\geq-1$ for the representation of the IDS in Proposition~\ref{prop:IDSFormula},
we wish to elaborate on their possible values and their connection
to the spectral code.

We have 
\[
\ypi_{-1}(\wo)=\begin{cases}
0, & \w(0)=A^{(1)},\\
1, & \w(0)=G^{(2)}.
\end{cases}
\]
For $k\in\Nz$ one can read the value of $\ypi_{k}(\wo)$ from the
following table.

\begin{center}

\begin{tabular}{|l|l|l|l|l|l|l|l|l|}
\hline 
~ & \multicolumn{3}{c|}{$k=0$} & \multicolumn{5}{c|}{$k\geq1$}\tabularnewline
\hline 
$\w(k)$ & \multicolumn{2}{c|}{$A^{(1)}$} & \multicolumn{1}{c|}{$G^{(2)}$} & \multicolumn{2}{c|}{$A^{(i)}$} & \multicolumn{2}{c|}{$B$} & \multicolumn{1}{c|}{$G^{(i)}$}\tabularnewline
\hline 
$\w(k+1)$ & $A^{(j)}$ & $G^{(j)}$ & $B$ & $A^{(j)}$ & $G^{(j)}$ & $A^{(j)}$ & $G^{(j)}$ & $B$\tabularnewline
\hline 
$\#\paths_{k}(\wo,A)$ & $0$ & $0$ & $1$ & $i-1$ & $i-1$ & $0$ & $0$ & $i-1$\tabularnewline
$\#\paths_{k}(\wo,G)$ & $0$ & $0$ & $0$ & $i$ & $i$ & $0$ & $0$ & $i-1$\tabularnewline
$\#\paths_{k+1}(\wo,A)$ & $j-1$ & $j-1$ & $0$ & $j-1$ & $j-1$ & $j-1$ & $j-1$ & $0$\tabularnewline
\hline 
$\ypi_{k}(\wo)$ & $j-1$ & $j-1$ & $-1$ & $j$ & $j$ & $j-1$ & $j-1$ & $0$\tabularnewline
\hline 
\end{tabular}

\end{center}

Therefore we can conclude for all $k\in\Nmo$ 
\[
\ypi_{k}(\wo)=\delta_{A,\wo}(k)+\#\paths_{k+1}(\wo,A)-\delta_{k,0}
\]
where 
\[
\delta_{A,\wo}(k):=\begin{cases}
1 & \w_{k}\in\{A^{(i)}:i\in\N\}\text{ and }k\geq0,\\
0 & \text{else}
\end{cases}
\]
\end{rem}

\begin{rem}
Comparing the IDS formula in Proposition~\ref{prop:IDSFormula} to
the formula in \cite[thm. 4.7]{Raym95}, shows that they are similar
up to an additional term of $-\alpha$ which appears in \cite[thm. 4.7]{Raym95},
but not in Proposition~\ref{prop:IDSFormula}. The source for this
difference is the connection between the coefficients $\ypi_{k}(\wo)$,
we used above, and similar coefficients $\pi_{k}(\wo)$ in \cite{Raym95}.
It can be checked that the connection between both type of coefficients
is given by
\[
\pi_{k}(\wo)=\ypi_{k}(\wo)+\delta_{k,0}.
\]
\end{rem}

We conclude this subsection by making a connection between the set
of all possible infinite codes, $\wo\in\Sigma_{\alpha}$ and the set
of all possible infinite sequences, $\left(\ypi_{k}\right)_{k\in\Nmo}$.
The latter set is given by 
\[
\yPi_{\alpha}:=\left\{ (\ypi_{k})_{k\in\Nmo}\in\Nmo^{\Nmo}\;\middle|\;\begin{array}{ll}
\ypi_{-1}\in\{0,1\} & \text{ and }\ypi_{-1}=1\iff\ypi_{0}=-1,\\
\ypi_{0}\in\{-1,\dots,c_{1}-1\} & \text{ and }\ypi_{0}=c_{1}-1\implies\ypi_{1}=0,\\
\ypi_{j}\in\{0,\dots,c_{j+1}\} & \text{ and }\ypi_{j}=c_{j+1}\implies\ypi_{j+1}=0\text{ for }j\geq1
\end{array}\right\} ,
\]

where $(c_{k})_{k\in\N}$ is the continued fraction expansion of $\alpha$.
\begin{lem}
[{\cite[prop. 4.4]{Raym95}}] \label{lem:PiCode2E} Let $\alpha\in[0,1]\setminus\Q$.
Then there is a bijection between $\Sigma_{\alpha}$ and $\yPi_{\alpha}$.
The bijection is explicitly given by $\wo\mapsto\left(\ypi_{k}(\wo)\right)_{k\in\Nmo}$
where 
\[
\ypi_{-1}(\wo):=\#\paths_{0}(\wo,A)\quad\text{and}\quad\ypi_{k}(\wo):=\#\paths_{k}(\wo,G)-\#\paths_{k}(\wo,A)+\#\paths_{k+1}(\wo,A),\,k\in\N_{0}.
\]
\end{lem}

\begin{proof}
It is straightforward to verify that the map in the statement is well
defined: given $\wo\in\Sigma_{\alpha}$, one can check that $\left(\ypi_{k}(\wo)\right)_{k\in\Nmo}\in\yPi_{\alpha}$.
To see this compare the definition of $\yPi_{\alpha}$ with Remark~\ref{rem:PiComputation}
which shows a table which characterizes $(\ypi_{k})_{k\in\Nmo}$.
To show that this map is a bijection, we take $(\ypi_{k})_{k\in\Nmo}\in\yPi_{\alpha}$
and inductively compute the corresponding $\w(k)$. On the way, we
prove that $\gamma(k)$ is uniquely determined by $(\mu_{-1},\mu_{0},\ldots,\mu_{k+1})$.

To aid this computation confer the table in Remark~\ref{rem:PiComputation}
and recall that for all $k\in\Nmo$ 
\[
\ypi_{k}(\wo)=\delta_{A,\wo}(k)+\#\paths_{k+1}(\wo,A)-\delta_{k,0}.
\]

First, if $\ypi_{-1}=1$, then we set $\w(0)=G^{(2)}$ and otherwise
if $\ypi_{-1}=0$, then we set $\w(0)=A^{(1)}$.

If $\ypi_{0}=-1,$ then $\ypi_{-1}=1$ and therefore $\w(0)=G^{(2)}$.
Now property $(\Sigma4)$ from Definition~\ref{def:SpectralCodes}
implies $\wo(1)=B$. Else, if $\ypi_{0}=i\in\{0,\dots c_{1}-1\}$,
then $\w(0)=A^{(1)}$ follows. If in addition $\ypi_{1}=0$, then
we conclude $\w(1)=G^{(i)}$, and if $\ypi_{1}\neq0$, then $\w(1)=A^{(i)}$
holds. Notice that the code $(\w(0),\w(1))$ generated this way fulfills
Definition~\ref{def:SpectralCodes}. This describes the induction
base of the construction.

Now assume $\wo$ was uniquely determined by $\ypi$ up to $\w(k)$
for some $k\geq1$. If $\w(k)$ was $G^{(i)}$, then property $(\Sigma4)$
yields again $\w(k+1)=B$. If $\w(k)\in\{A^{(i)}:i\in\N\}\cup\{B\}$,
then $(\Sigma2)$ and $(\Sigma4)$ imply $\w(k+1)\in\{A^{(i)}:i\in\N\}\cup\{G^{(i)}:i\in\N\}$.
Assume $\ypi_{k}=j$. Then we get even $\w(k+1)\in\{A^{(j)},G^{(j)}\}$.
If $\ypi_{k+1}=0$, then $\w(k+1)=G^{(j)}$ and otherwise $\w(k+1)=A^{(j)}$.
Thus $\w(k),\ypi_{k}$ and $\ypi_{k+1}$ uniquely determine $\w(k+1)$.
Observe again, that the code $(\w(0),\dots,\w(k+1))$ fulfills Definition~\ref{def:SpectralCodes}.
\end{proof}

\subsection{Finding all the gap labels\protect\label{subsec:=000020all=000020gaps=000020there}}

Our aim in this subsection is to prove the following.
\begin{thm}
\label{thm:=000020All=000020gaps=000020are=000020there=000020V>4}
Let $V>4$. For $\alpha\in[0,1]\setminus\Q$ we have 
\[
\set{\IDS(E)}{E\in\R\setminus\sigma(\Ham)}=\set{\ell\alpha\mod 1}{\ell\in\Z}\cup\left\{ 1\right\} =(\Z+\Z\alpha)\cap[0,1].
\]
\end{thm}

A main tool in the proof of the theorem is Proposition~\ref{prop:IDSFormula}.
Therefore, we need the following auxiliary lemma.
\begin{lem}
[{\cite[prop. 5.2]{Raym95}}] \label{lem:=000020Integer=000020as=000020a=000020linear=000020combination=000020of=000020q's}
Let $\alpha\in[0,1]\setminus\Q$ with convergents $\pqk=\varphi(\co_{k})$.
For each $\ell\in\Z$, there is some $k_{0}\in\N$ and $\ypi=(\ypi_{j})_{j\in\Nmo}\in\yPi_{\alpha}$
such that $\ypi_{j}=0$ for $j>k_{0}$ and 
\begin{equation}
\ell=\sum_{j=-1}^{\infty}(-1)^{j}\ypi_{j}q_{j}=\sum_{j=-1}^{k_{0}}(-1)^{j}\ypi_{j}q_{j}.\label{eq:EllSum}
\end{equation}
Moreover, if $\ell\notin\{-1,0,1\}$, then $\mu_{k_{0}}\geq1$.
\end{lem}

\begin{proof}
We prove the statement inductively over $m\in\N$ that
\begin{enumerate}
\item \label{enu:=000020integer-proof=000020-=0000201} For all $\ell\in[-q_{2m},q_{2m-1})$
there is a $k_{0}\leq2m$ and a $\ypi=\left(\mu_{j}\right)\in\yPi_{\alpha}$
satisfying (\ref{eq:EllSum}) $\mu_{j}=0$ if $j>k_{0}$.
\item \label{enu:=000020integer-proof=000020-=0000202} If $\mu_{k_{0}-1}=c_{k_{0}}$,
then $\ell\in[-q_{k_{0}},-q_{k_{0}}+q_{k_{0}-1})$ if $k_{0}$ is
even \\
and $\ell\in[-q_{k_{0}-1}+q_{k_{0}},q_{k_{0}})$ if $k_{0}$ is odd.
\end{enumerate}
To do this, we check the claim in an alternating manner on the positive
and negative part of these intervals. Also recall the recursive behavior
of the sequence $(q_{k})_{k\in\Nz}$, that is 
\[
q_{-1}=0,\quad q_{0}=1\quad\text{and}\quad q_{k}=c_{k}q_{k-1}+q_{k-2}\text{ for }k\in\N.
\]
First, let $m=1$ and consider $\ell\in[-q_{2},q_{1})=[-q_{2},-q_{0})\cup\{-q_{0}\}\cup(-q_{0},q_{1})$
and we separately consider each of the cases 
\[
\ell\in\{-q_{0}\},\quad\ell\in(-q_{0},q_{1})\quad\text{ and }\quad\ell\in[-q_{2},-q_{0}).
\]
If $\ell=-q_{0}=-1$, then we choose $k_{0}=0\leq2m-1$ and $\ypi:=(\ypi_{j})_{j\in\Nmo}:=(1,-1,0,0,\dots)$.
In this case we get $\ypi\in\yPi_{\alpha}$ and 
\[
\sum_{j=-1}^{\infty}(-1)^{j}\ypi_{j}q_{j}=\ypi_{0}=-1=\ell.
\]
For the second case, if $\ell\in(-q_{0},q_{1})\cap\Z=[0,c_{1})\cap\Z$,
choose $k_{0}=0\leq2m-1$ and $\ypi:=(0,\ell,0,0,\dots)$. Again,
observe that $\ypi\in\yPi_{\alpha}$ satisfies (\ref{eq:EllSum}).
For the third case we assume $\ell\in[-q_{2},-q_{0})$. Notice that
we can decompose this interval into the $c_{2}$ intervals 
\[
[-q_{2},-q_{0})=\bigsqcup_{j=1}^{c_{2}}[-q_{0}-(c_{2}+1-j)q_{1},q_{0}-(c_{2}-j)q_{1})
\]
That is, there is a unique $\ypi_{1}\in\{1,\dots,c_{2}\}$ with 
\[
-q_{0}-\ypi_{1}q_{1}\leq\ell<-q_{0}-(\ypi_{1}-1)q_{1}.
\]
Equivalently, we can write this as 
\[
-q_{0}\leq\ell+\ypi_{1}q_{1}<-q_{0}+q_{1}=c_{1}-1.
\]
If $\ell+\ypi_{1}q_{1}$ happen to be $-q_{0}=-1$, then we can apply
the first case to it and get 
\[
\ell+\ypi_{1}q_{1}=\sum_{j=-1}^{0}(-1)^{j}\ypi_{j}q_{j},
\]
with $\ypi_{-1}=1$ and $\ypi_{0}=-1$. Then $\ypi:=(1,-1,\ypi_{1},0,0,\dots)\in\yPi_{\alpha}$
and $k_{0}=1\leq2m-1$ satisfy (\ref{eq:EllSum}) for the given $\ell$.
We proceed similarly when 
\[
-q_{0}<\ell+\ypi_{1}q_{1}<-q_{0}+q_{1}=c_{1}-1.
\]
More precisely, if this holds, the second case yields 
\[
\ell+\ypi_{1}q_{1}=\sum_{j=-1}^{0}(-1)^{j}\ypi_{j}q_{j},
\]
with $\ypi_{-1}=0$ and $\ypi_{0}=\ell+\ypi_{1}q_{1}$. Thus for $\ypi:=(0,\ypi_{0},\ypi_{1},0,0,\dots)$
we have $\ypi\in\yPi_{\alpha}$ as $\ypi_{0}\neq c_{1}-1$ (we need
this since $\ypi_{1}\neq0$) and, by construction, $\ypi$ satisfies
(\ref{eq:EllSum}) for the given $\ell$ and $k_{0}=1\leq2m-1$. This
ends the induction base.

For the induction step suppose (\ref{enu:=000020integer-proof=000020-=0000201})
and (\ref{enu:=000020integer-proof=000020-=0000202}) hold for some
$m\in\N$. Let $\ell\in[-q_{2m+2},q_{2m+1})$. Again, we separately
discuss the three cases 
\[
\ell\in[-q_{2m+2},-q_{2m}),\quad\ell\in[-q_{2m},q_{2m-1})\quad\text{ and }\quad\ell\in[q_{2m-1},q_{2m+1}).
\]
For $\ell\in[-q_{2m},q_{2m-1})$ there is nothing to do, as (\ref{eq:EllSum})
holds by the induction hypothesis.

If $\ell\in[q_{2m-1},q_{2m+1})=[q_{2m-1},q_{2m-1}+c_{2m+1}q_{2m})$,
then there exists some $\ypi_{2m}\in\{1,\dots,c_{2m+1}\}$ such that
\begin{equation}
q_{2m-1}+(\ypi_{2m}-1)q_{2m}\leq\ell<q_{2m-1}+\ypi_{2m}q_{2m},\label{eq:EllEstimate}
\end{equation}

or equivalently 
\[
-q_{2m}+q_{2m-1}\leq\ell-\ypi_{2m}q_{2m}<q_{2m-1}.
\]
In particular we observe $\ell-\ypi_{2m}q_{2m}\notin[-q_{2m},-q_{2m}+q_{2m-1})$.
Thus the induction hypothesis implies that there exists $(\ypi_{-1},\dots,\ypi_{2m-1},0,0,\dots)\in\yPi_{\alpha}$
with $\ypi_{2m-1}\neq c_{2m}$ such that 
\[
\ell-\ypi_{2m}q_{2m}=\sum_{j=-1}^{2m-1}(-1)^{j}\ypi_{j}q_{j}.
\]
Note that if $\ypi_{2m-1}=c_{2m},$ then the induction hypothesis
for (\ref{enu:=000020integer-proof=000020-=0000202}) and $k_{0}=2m$
yields $\ell-\ypi_{2m}q_{2m}\in[-q_{2m},-q_{2m}+q_{2m-1})$, a contradiction.

Therefore we set $\ypi:=(\ypi_{-1},\dots,\ypi_{2m-1},\ypi_{2m},0,0,\dots)$
and observe $\ypi\in\yPi_{\alpha}$, as $\ypi_{2m-1}\neq c_{2m}$.
Then (\ref{eq:EllSum}) hold for the given $\ell$ and $k_{0}=2m\leq2(m+1)-1$.
Note that if $\ypi_{2m}=c_{2m+1}$ i.e. $k_{0}=2m+1$ odd then (\ref{eq:EllEstimate})
implies 
\[
\ell\in[q_{2m+1}-q_{2m},q_{2m+1})=[-q_{k_{0}-1}+q_{k_{0}},q_{k_{0}}).
\]
This proves (\ref{enu:=000020integer-proof=000020-=0000202}) if $k_{0}$
is odd for $m+1$.

Finally, we suppose $\ell\in[-q_{2m+2},-q_{2m})=[-c_{2m+2}q_{2m+1}-q_{2m},-q_{2m})$.
Then there exists some $\ypi_{2m+1}\in\{1,\dots,c_{2m+2}\}$ such
that 
\begin{equation}
-\ypi_{2m+1}q_{2m+1}-q_{2m}\leq\ell<-(\ypi_{2m+1}-1)q_{2m+1}-q_{2m},\label{eq:EllEstimate2}
\end{equation}
or equivalently 
\[
-q_{2m}\leq\ell+\ypi_{2m+1}q_{2m+1}<-q_{2m}+q_{2m+1}=-q_{2m}+c_{2m+1}q_{2m}+q_{2m-1}.
\]
In particular we observe $\ell+\ypi_{2m+1}q_{2m+1}\notin[-q_{2m}+q_{2m+1},q_{2m+1})$.
Thus, the induction hypothesis implies that there some $(\ypi_{-1},\dots,\ypi_{2m},0,0,\dots)\in\yPi_{\alpha}$
with $\ypi_{2m}\neq c_{2m+1}$ such that 
\[
\ell+\ypi_{2m+1}q_{2m+1}=\sum_{j=-1}^{2m}(-1)^{j}\ypi_{j}q_{j}.
\]
Note that if $\ypi_{2m}=c_{2m+1},$ then the induction hypothesis
for (\ref{enu:=000020integer-proof=000020-=0000202}) and $k_{0}=2m+1$
yields $\ell+\ypi_{2m+1}q_{2m+1}\in[-q_{2m}+q_{2m+1},q_{2m+1})$,
a contradiction.

Therefore we set $\ypi:=(\ypi_{-1},\dots,\ypi_{2m},\ypi_{2m+1},0,0,\dots)$
and observe $\ypi\in\yPi_{\alpha}$, as $\ypi_{2m}\neq c_{2m+1}$.
Then (\ref{eq:EllSum}) hold for the given $\ell$ and $k_{0}=2m+1\leq2(m+1)-1$.
Note that if $\ypi_{2m+1}=c_{2m+2}$ i.e. $k_{0}=2m+2$ even then
(\ref{eq:EllEstimate2}) implies 
\[
\ell\in[-q_{2m+2},-q_{2m+2}-q_{2m+1}).
\]
Thus, (\ref{enu:=000020integer-proof=000020-=0000202}) holds if $k_{0}$
is even for $m+1$.

Note that $\ypi_{k_{0}}\geq0$ holds for all $k_{0}\geq1$. The cases
where $k_{0}\in\{-1,0\}$ and $\mu_{k_{0}}=0$ are exactly when $\ell\in\{-1,0,1\}$.
Therefore $\ell\in\Z\setminus\{-1,0,1\}$ and $\ypi_{k_{0}}\neq0$
imply $\ypi_{k_{0}}\geq1$.
\end{proof}
\begin{proof}
[Proof of Theorem~\ref{thm:=000020All=000020gaps=000020are=000020there=000020V>4}]
We start by recalling that the gap labelling theorem \cite[prop. 5.2.4]{Bell92-Gap}
already provides the inclusion, 
\[
\mathcal{G}:=\{\IDS(E):E\in\R\setminus\sigma(\Ham)\}\subseteq\set{\ell\alpha\mod 1}{\ell\in\Z}\cup\left\{ 1\right\} .
\]

First note, that the spectrum $\sigma(\Ham)$ is a compact subset
of $\R$. Then if $E<\inf\sigma(\Ham)$, we obtain $\IDS(E)=0$. Similarly
if $E>\sup\sigma(\Ham)$, then $\IDS(E)=1$ and so $\{0,1\}\subset\mathcal{G}$
holds. We continue proving $\ell\alpha\mod 1\in\mathcal{G}$ for $\ell\in\Z\setminus\{-1,0\}.$
The case $\ell=-1$ will be treated separately at the end.

Let $\ell\in\Z\setminus\{-1,0\}$ and let $\ypi=(\ypi_{-1},\ypi_{0},\dots)\in\yPi_{\alpha}$
be such that there is some $k_{0}\in\Nz$ with $\ell=\sum_{j=-1}^{k_{0}}(-1)^{j}\ypi_{j}q_{j}$
and $\ypi_{j}=0$ for all $j\geq k_{0}+1$, which exists by Lemma~\ref{lem:=000020Integer=000020as=000020a=000020linear=000020combination=000020of=000020q's}.
In addition, Lemma~\ref{lem:=000020Integer=000020as=000020a=000020linear=000020combination=000020of=000020q's}
asserts that $\mu_{k_{0}}\geq1$. Then define

\begin{equation}
\ypi'=(\ypi_{-1},\ypi_{0},\dots\ypi_{{k_{0}}-1},(\ypi_{k_{0}}-1),c_{{k_{0}}+2},0,c_{{k_{0}}+4},0,c_{{k_{0}}+6},\dots).
\end{equation}
Observe that $\ypi'\in\yPi_{\alpha}$ using $\mu_{k_{0}}\geq1$. Then
Lemma~\ref{lem:PiCode2E} implies that there are unique $\wo,\wo'\in\Sigma_{\alpha}$
such that $E:=E_{\alpha,V}(\wo)$ and $E':=E_{\alpha,V}(\wo')$ satisfy
\[
\IDS(E)=\sum_{j=-1}^{\infty}(-1)^{j}\ypi_{j}(q_{j}\alpha-p_{j})\quad\text{and}\quad\IDS(E')=\sum_{j=-1}^{\infty}(-1)^{j}\ypi'_{j}(q_{j}\alpha-p_{j}).
\]

With this choice we get 
\begin{align*}
[0,1]\ni\IDS(E) & =\sum_{j=-1}^{\infty}(-1)^{j}\ypi_{j}(q_{j}\alpha-p_{j})\\
 & =\sum_{j=-1}^{k_{0}}(-1)^{j}\ypi_{j}(q_{j}\alpha-p_{j})\\
 & =\underbrace{\sum_{j=-1}^{k_{0}}(-1)^{j}\ypi_{j}q_{j}\alpha}_{=\ell\alpha}-\underbrace{\sum_{j=-1}^{k_{0}}(-1)^{j}\ypi_{j}p_{j}}_{\in\Z}
\end{align*}
and therefore $\IDS(E)=\ell\alpha\mod 1$.

We also claim $E\neq E'$. Assume differently, i.e. $E=E'$, then
$\wo=\wo'$ due to Lemma~\ref{lem:EIsABijection}. Hence, $\ypi_{j}(\wo)=\ypi_{j}(\wo')$
follows for all $j\in\Nmo$ by the definition of $(\ypi_{j})_{j\in\Nmo}$
in Proposition~\ref{prop:IDSFormula}. This yields a contradiction
for $j\geq k_{0}$.

Next we observe 
\begin{align*}
|\IDS(E')-\IDS(E)|= & \,|(-1)^{{k_{0}}+1}(q_{k_{0}}\alpha-p_{k_{0}})+(-1)^{{k_{0}}+1}c_{{k_{0}}+2}(q_{{k_{0}}+1}\alpha-p_{{k_{0}}+1})\\
 & \,+\sum_{j={k_{0}}+3}^{\infty}(-1)^{j}\ypi_{j}'(q_{j}\alpha-p_{j})|\\
= & \,\left|(-1)^{{k_{0}}+1}(q_{{k_{0}}+2}\alpha-p_{{k_{0}}+2})+\sum_{j={k_{0}}+3}^{\infty}(-1)^{j}\ypi_{j}'(q_{j}\alpha-p_{j})\right|.
\end{align*}
Using the recursion formulas for $\{p_{k}\}_{k\in\Nmo}$ and $\{q_{k}\}_{k\in\Nmo}$
from eq. (\ref{eq:=000020recursion=000020for=000020p_k,=000020q_k}),
we inductively conclude 
\begin{align*}
 & \left|(-1)^{{k_{0}}+1}(q_{{k_{0}}+2n}\alpha-p_{{k_{0}}+2n})+\sum_{j={k_{0}}+2n+1}^{\infty}(-1)^{j}\ypi_{j}'(q_{j}\alpha-p_{j})\right|\\
= & \left|(-1)^{{k_{0}}+1}q_{{k_{0}}+2(n+1)}\alpha-p_{{k_{0}}+2(n+1)})+\sum_{j={k_{0}}+2(n+1)+1}^{\infty}(-1)^{j}\ypi_{j}'(q_{j}\alpha-p_{j})\right|,
\end{align*}
for all $n\in\N$. Hence, we obtain 
\[
|\IDS(E')-\IDS(E)|\leq|q_{k+2n}\alpha-p_{k+2n}|+\left|\sum_{j=k+2n+1}^{\infty}(-1)^{j}\ypi_{j}'(q_{j}\alpha-p_{j})\right|,
\]
for all $n\in\N$. Sending $n\to\infty$ and using that the sum exists,
we conclude 
\[
|\IDS(E')-\IDS(E)|\leq\lim_{n\to\infty}|q_{k+2n}\alpha-p_{k+2n}|\leq0,
\]
by properties of the Diophantine approximation \cite[thm. 9]{Khinchin_book64}.
Therefore $\ell\alpha\mod 1=\IDS(E)=\IDS(E')$ while $E\neq E'$.
Since the IDS is monotonously increasing \ref{enu:=000020IDS-property-1}
and constant on the gaps \ref{enu:=000020IDS-property-2}, we conclude
that $(E,E')$ is completely contained in $\R\setminus\sigma(\Ham)$.
That is for all $E''\in(E,E')$ we still get 
\[
\IDS(E'')=\IDS(E)=\ell\alpha\mod 1
\]
and so we conclude $\ell\alpha\mod 1\in\mathcal{G}$.

We now treat the last case $\ell=-1$ and the gap label $\ell\alpha\mod 1=1-\alpha$
for $\ell=-1$. Let $\ypi=(1,-1,0,0,\dots)$ and $\ypi'=(0,c_{1}-1,0,c_{3},0,c_{5}\dots)$.
By Lemma~\ref{lem:PiCode2E} there are again unique $\wo,\wo'\in\Sigma_{\alpha}$
such that $\ypi(\wo)=\ypi$ and $\ypi(\wo')=\ypi'$. Following similar
computations as above, we get $\IDS(E_{\alpha,V}(\wo))=1-\alpha=\IDS(E_{\alpha,V}(\wo'))$.
Since $\wo\neq\wo'$, Lemma~\ref{lem:EIsABijection} implies $E_{\alpha,V}(\wo)\neq E_{\alpha,V}(\wo')$.
Hence we also get in this case $1-\alpha\in\mathcal{G}$.
\end{proof}

\appendix
\appendix
\renewcommand{\thesection}{\Roman{section}}

\section{A recursive relation for periods of mechanical words \protect\label{sec:RecursiveRelation_Sturmian}}

This appendix is devoted to the proof of Lemma~\ref{lem:=000020WordRecursion}.
Let $\alpha\in[0,1]\setminus\Q$ with infinite continued fraction
expansion $(c_{k})_{k=0}^{\infty}$ and convergents $\alk=\varphi([0,c_{0},c_{1},\ldots,c_{k}])$
for $k\in\N_{0}$. Recall from (\ref{eq:=000020Sturmian=000020period=000020def})
the definition 
\[
W_{k}(i):=\omega_{\alpha_{k}}(i),\qquad0\leq i\leq q_{k}-1,
\]
for the period of those mechanical words $\omega_{\alk}$. Further
recall the statement of Lemma~\ref{lem:=000020WordRecursion}:
\[
W_{0}=0,\quad\quad W_{1}=\underbrace{0\ldots0}_{c_{1}-1}1,
\]
and if $k\geq2$, then
\[
W_{k}=\begin{cases}
W_{k-2}W_{k-1}^{c_{k}},\quad & k\equiv0\mod 2,\\
W_{k-1}^{c_{k}}W_{k-2},\quad & k\equiv1\mod 2,
\end{cases}
\]
where the power means a concatenation of words. We now bring three
auxiliary lemmas (Lemma~\ref{lem:=000020period=000020prefixes},
Lemma~\ref{lem:=000020period=000020suffixes} and Lemma~\ref{lem:=000020sub-periods=000020of=000020periods})
which are needed to prove the recursion property of the Sturmian periods
as stated in Lemma~\ref{lem:=000020WordRecursion}.
\begin{lem}
\label{lem:=000020period=000020prefixes} [Period prefixes] Let $\alpha\in[0,1]\setminus\Q$
and $k\geq2$.
\begin{enumerate}
\item \label{enu:=000020period=000020prefixes_a}If $k\equiv0\mod 2$ then
\begin{equation}
W_{k}(i)=W_{k-1}(i)\qquad\textrm{for all }0\leq i\leq q_{k-1}-2\label{eq:=000020lem-period=000020initials=000020-=000020even=000020k=000020-=0000201}
\end{equation}
and 
\begin{equation}
W_{k}(i)=W_{k-2}(i\mod{q_{k-2}})\qquad\textrm{for all }0\leq i\leq q_{k-1}-2.\label{eq:=000020lem-period=000020initials=000020-=000020even=000020k=000020-=0000202}
\end{equation}
\item \label{enu:=000020period=000020prefixes_b}If $k\equiv1\mod 2$ then
\begin{equation}
W_{k}(i)=W_{k-1}(i\mod{q_{k-1}})\qquad\textrm{for all }0\leq i\leq q_{k}-2\label{eq:=000020lem-period=000020initials=000020-=000020odd=000020k=000020-=0000201}
\end{equation}
and
\begin{equation}
W_{k}(i)=W_{k-2}(i)\qquad\textrm{for all }0\leq i\leq q_{k-2}-2.\label{eq:=000020lem-period=000020initials=000020-=000020odd=000020k=000020-=0000202}
\end{equation}
\end{enumerate}
\end{lem}

\begin{proof}
(\ref{enu:=000020period=000020prefixes_a}) Start by treating the
case $k\equiv0\mod 2$. Since $k$ is even, standard theory of rational
convergents (\cite[thm. 4]{Khinchin_book64}) imply $\frac{p_{k-2}}{q_{k-2}}<\frac{p_{k}}{q_{k}}<\frac{p_{k-1}}{q_{k-1}}$.

We start by showing that for all $0\leq i\leq q_{k-1}-1$ , $\fl{\frac{p_{k}}{q_{k}}\thinspace i}=\fl{\frac{p_{k-1}}{q_{k-1}}\thinspace i}$
from which (\ref{eq:=000020lem-period=000020initials=000020-=000020even=000020k=000020-=0000201})
of the Lemma follows when using Lemma~\ref{lem:=000020lower=000020mechanical=000020word}
and (\ref{eq:=000020Sturmian=000020period=000020def}).

Assume towards contradiction that there exists $0\leq i\leq q_{k-1}-1$
such that $\fl{\frac{p_{k}}{q_{k}}\thinspace i}\neq\fl{\frac{p_{k-1}}{q_{k-1}}\thinspace i}$.
Clearly, $i>0$ must hold. Using $\frac{p_{k}}{q_{k}}<\frac{p_{k-1}}{q_{k-1}}$,
we infer that there exists an $m\in\N$ such that 
\[
\frac{p_{k}}{q_{k}}\thinspace i<m\leq\frac{p_{k-1}}{q_{k-1}}\thinspace i,
\]
or equivalently, 
\begin{equation}
\frac{p_{k}}{q_{k}}<\frac{m}{i}\leq\frac{p_{k-1}}{q_{k-1}}.\label{eq:=000020lem-potentials=000020of=000020consequative=000020aprroximants=000020-=0000202-2}
\end{equation}
Since $k$ is even, \cite[thm. 2]{Khinchin_book64} implies
\begin{equation}
\frac{p_{k-1}}{q_{k-1}}-\frac{p_{k}}{q_{k}}=\frac{1}{q_{k-1}q_{k}}\label{eq:=000020lem-potentials=000020of=000020consequative=000020aprroximants=000020-=0000203-2}
\end{equation}
By (\ref{eq:=000020lem-potentials=000020of=000020consequative=000020aprroximants=000020-=0000202-2}),
$mq_{k}-ip_{k}>0$ holds and so $mq_{k}-ip_{k}\geq1$. Thus, (\ref{eq:=000020lem-potentials=000020of=000020consequative=000020aprroximants=000020-=0000203-2})
and $i\leq q_{k-1}-1$ lead to 
\[
\frac{1}{q_{k-1}q_{k}}=\frac{p_{k-1}}{q_{k-1}}-\frac{p_{k}}{q_{k}}\geq\frac{m}{i}-\frac{p_{k}}{q_{k}}=\frac{mq_{k}-ip_{k}}{iq_{k}}>\frac{mq_{k}-ip_{k}}{q_{k-1}q_{k}}\geq\frac{1}{q_{k-1}q_{k}},
\]
a contradiction.

Next, we show that for all $0\leq i\leq q_{k-1}-1$ , $\fl{\frac{p_{k-2}}{q_{k-2}}\thinspace i}=\fl{\frac{p_{k}}{q_{k}}\thinspace i}$
from which (\ref{eq:=000020lem-period=000020initials=000020-=000020even=000020k=000020-=0000202})
of the Lemma follows when using Lemma~\ref{lem:=000020lower=000020mechanical=000020word}
and (\ref{eq:=000020Sturmian=000020period=000020def}).\\
 Assume towards contradiction that there exists $0\leq i\leq q_{k-1}-1$
such that $\fl{\frac{p_{k-2}}{q_{k-2}}\thinspace i}\neq\fl{\frac{p_{k}}{q_{k}}\thinspace i}$.
Clearly, $i>0$ must hold. Using $\frac{p_{k-2}}{q_{k-2}}<\frac{p_{k}}{q_{k}}$,
we infer that there exists an $m\in\N$ such that 
\[
\frac{p_{k-2}}{q_{k-2}}\thinspace i<m\leq\frac{p_{k}}{q_{k}}\thinspace i,
\]
or equivalently, 
\begin{equation}
\frac{p_{k-2}}{q_{k-2}}<\frac{m}{i}\leq\frac{p_{k}}{q_{k}}.\label{eq:=000020lem-potentials=000020of=000020consequative=000020aprroximants=000020-=0000202-1-1}
\end{equation}
 Let $(c_{i})_{i\in\N_{0}}$be the infinite continued fraction expansion
of $\alpha$. Since $k$ is even, \cite[thm. 3]{Khinchin_book64}
implies
\begin{equation}
\frac{p_{k}}{q_{k}}-\frac{p_{k-2}}{q_{k-2}}=\frac{c_{k}}{q_{k}q_{k-2}}\label{eq:=000020lem-potentials=000020of=000020consequative=000020aprroximants=000020-=0000203-1-1}
\end{equation}
By (\ref{eq:=000020lem-potentials=000020of=000020consequative=000020aprroximants=000020-=0000202-1-1}),
$mq_{k-2}-ip_{k-2}>0$ holds and so $mq_{k-2}-ip_{k-2}\geq1$. Thus,
(\ref{eq:=000020lem-potentials=000020of=000020consequative=000020aprroximants=000020-=0000202-1-1})
and (\ref{eq:=000020lem-potentials=000020of=000020consequative=000020aprroximants=000020-=0000203-1-1})
lead to
\[
\frac{c_{k}}{q_{k}q_{k-2}}=\frac{p_{k}}{q_{k}}-\frac{p_{k-2}}{q_{k-2}}\geq\frac{m}{i}-\frac{p_{k-2}}{q_{k-2}}=\frac{mq_{k-2}-ip_{k-2}}{iq_{k-2}}>\frac{mq_{k-2}-ip_{k-2}}{q_{k-1}q_{k-2}}.
\]
Since $k\geq2$, we have $q_{k-2}>0$ and so (\ref{eq:=000020recursion=000020for=000020p_k,=000020q_k})
and the previous estimate lead to
\[
1\leq mq_{k-2}-ip_{k-2}<\frac{c_{k}q_{k-1}}{q_{k}}=\frac{c_{k}q_{k-1}}{c_{k}q_{k-1}+q_{k-2}}<1,
\]
a contradiction.

(\ref{enu:=000020period=000020prefixes_b}) For the case $k\equiv1\mod 2$,
we get by standard theory of rational convergents (\cite[thm. 4]{Khinchin_book64})
that $\frac{p_{k-1}}{q_{k-1}}<\frac{p_{k}}{q_{k}}<\frac{p_{k-2}}{q_{k-2}}$.

To prove (\ref{eq:=000020lem-period=000020initials=000020-=000020odd=000020k=000020-=0000201})
we need to show that for all $0\leq i\leq q_{k}-1$ , $\fl{\frac{p_{k-1}}{q_{k-1}}\thinspace i}=\fl{\frac{p_{k}}{q_{k}}\thinspace i}$,
which can be done similarly to the way that (\ref{eq:=000020lem-period=000020initials=000020-=000020even=000020k=000020-=0000201})
was proven above (but exchanging the roles of $\frac{p_{k}}{q_{k}}$
and $\frac{p_{k-1}}{q_{k-1}}$). Note that statement (\ref{eq:=000020lem-period=000020initials=000020-=000020odd=000020k=000020-=0000201})
is stronger than (\ref{eq:=000020lem-period=000020initials=000020-=000020even=000020k=000020-=0000201})
in the sense that we show equality for a longer subset (up to $q_{k}-2$
rather than up to $q_{k-1}-2$).

Assume towards contradiction that there exists $0\leq i\leq q_{k}-1$
such that $\fl{\frac{p_{k-1}}{q_{k-1}}\thinspace i}\neq\fl{\frac{p_{k}}{q_{k}}\thinspace i}$.
Clearly, $i>0$ must hold. Using $\frac{p_{k-1}}{q_{k-1}}<\frac{p_{k}}{q_{k}}$,
we infer that there exists an $m\in\N$ such that
\[
\frac{p_{k-1}}{q_{k-1}}<\frac{m}{i}\leq\frac{p_{k}}{q_{k}}.
\]
Thus, $q_{k-1}m-p_{k-1}i\geq1$ follows. Since $k$ is odd, \cite[thm. 2]{Khinchin_book64}
implies
\[
\frac{p_{k}}{q_{k}}-\frac{p_{k-1}}{q_{k-1}}=\frac{1}{q_{k-1}q_{k}}.
\]
Hence, $i\leq q_{k}-1$ lead to 
\[
\frac{1}{q_{k-1}q_{k}}=\frac{p_{k}}{q_{k}}-\frac{p_{k-1}}{q_{k-1}}\geq\frac{m}{i}-\frac{p_{k-1}}{q_{k-1}}=\frac{q_{k-1}m-p_{k-1}i}{iq_{k-1}}>\frac{q_{k-1}m-p_{k-1}i}{q_{k-1}q_{k}}\geq\frac{1}{q_{k-1}q_{k}},
\]
a contradiction.

To prove (\ref{eq:=000020lem-period=000020initials=000020-=000020odd=000020k=000020-=0000202})
we need to show that for all $0\leq i\leq q_{k-2}-1$ , $\fl{\frac{p_{k}}{q_{k}}\thinspace i}=\fl{\frac{p_{k-2}}{q_{k-2}}\thinspace i}$.
Next, we show that for all $0\leq i\leq q_{k-2}-1$ , $\fl{\frac{p_{k-2}}{q_{k-2}}\thinspace i}=\fl{\frac{p_{k}}{q_{k}}\thinspace i}$
from which (\ref{eq:=000020lem-period=000020initials=000020-=000020even=000020k=000020-=0000202})
of the Lemma follows when using Lemma~\ref{lem:=000020lower=000020mechanical=000020word}
and (\ref{eq:=000020Sturmian=000020period=000020def}).

Assume towards contradiction that there exists $0\leq i\leq q_{k-2}-1$
such that $\fl{\frac{p_{k-2}}{q_{k-2}}\thinspace i}\neq\fl{\frac{p_{k}}{q_{k}}\thinspace i}$.
Clearly, $i>0$ must hold. Using $\frac{p_{k}}{q_{k}}<\frac{p_{k-2}}{q_{k-2}}$,
we infer that there exists an $m\in\N$ such that 
\[
\frac{p_{k}}{q_{k}}<\frac{m}{i}\leq\frac{p_{k-2}}{q_{k-2}},
\]
Let $(c_{i})_{i\in\N_{0}}$be the infinite continued fraction expansion
of $\alpha$. Since $k$ is odd, \cite[thm. 3]{Khinchin_book64} implies
\[
\frac{c_{k}}{q_{k}q_{k-2}}=\frac{p_{k-2}}{q_{k-2}}-\frac{p_{k}}{q_{k}}>\frac{p_{k-2}}{q_{k-2}}-\frac{m}{i}=\frac{p_{k-2}i-mq_{k-2}}{iq_{k-2}}.
\]
Hence, $i\leq q_{k-1}-1$ and the recursive relation (\ref{eq:=000020recursion=000020for=000020p_k,=000020q_k})
\[
1\geq\frac{c_{k}q_{k-1}}{q_{k}}>\frac{c_{k}i}{q_{k}}>p_{k-2}i-mq_{k-2}\geq0.
\]
Thus, $p_{k-2}i-mq_{k-2}=0$ follows or equivalently $\frac{m}{i}=\frac{p_{k-2}}{q_{k-2}}$.
This contradicts $i\leq q_{k-2}-1$ and $p_{k-1},q_{k-1}$ are coprime.
\end{proof}
\begin{lem}
\label{lem:=000020period=000020suffixes}[Period suffixes] Let $\alpha\in[0,1]\setminus\Q$
and $k\geq2$.
\begin{enumerate}
\item If $k\equiv0\mod 2$, then 
\begin{equation}
W_{k}(q_{k}-i)=W_{k-1}(q_{k-1}-i)\qquad\textrm{for all }1\leq i\leq q_{k-1}.\label{eq:=000020lem-period=000020suffixes=000020-=000020even=000020k}
\end{equation}
\item If $k\equiv1\mod 2$, then
\begin{equation}
W_{k}(q_{k}-i)=W_{k-2}(q_{k-2}-i)\qquad\textrm{for all }1\leq i\leq q_{k-2}.\label{eq:=000020lem-period=000020suffixes=000020-=000020odd=000020k}
\end{equation}
\end{enumerate}
\end{lem}

\begin{proof}
Using (\ref{eq:=000020lower=000020mechanical=000020word}) in Lemma~\ref{lem:=000020lower=000020mechanical=000020word}
gives for all $k\geq1$, 
\begin{equation}
W_{k}(q_{k}-1)=\lfloor(q_{k}-1+1)\frac{p_{k}}{q_{k}}\rfloor-\lfloor(q_{k}-1)\frac{p_{k}}{q_{k}}\rfloor=1,\label{eq:=000020lem-period=000020suffixes=000020-=0000201}
\end{equation}
which proves the case $i=1$ in (\ref{eq:=000020lem-period=000020suffixes=000020-=000020even=000020k})
and (\ref{eq:=000020lem-period=000020suffixes=000020-=000020odd=000020k}).
If $c_{1}=1$ and $k\in\left\{ 2,3\right\} $, we conclude $q_{1}=1$
from (\ref{eq:=000020recursion=000020for=000020p_k,=000020q_k}),
and hence the statement holds in this case (as we have shown that
it holds for $i=1$). Thus, in the sequel of the proof, when treating
$k\in\{2,3\}$, we will assume $c_{1}>1$.

We show another auxiliary statement which aids in the the proof -
that the subword $\left.W_{k}\right|_{\left\{ 1,\ldots,q_{k}-2\right\} }$
is a palindrome, i.e., 
\begin{equation}
W_{k}(i)=W_{k}(q_{k}-(i+1))\qquad\textrm{for all }1\leq i\leq q_{k}-2.\label{eq:=000020lem-period=000020suffixes=000020-=0000202}
\end{equation}
To prove this identity, observe for $1\leq i\leq q-2$ and $p,q$
coprime,
\begin{align*}
\left\lfloor \left(q-i\right)\frac{p}{q}\right\rfloor -\left\lfloor \left(q-1-i\right)\frac{p}{q}\right\rfloor  & =\left\lfloor -i\frac{p}{q}\right\rfloor -\left\lfloor -\left(i+1\right)\frac{p}{q}\right\rfloor \\
 & =-\left(\left\lfloor i\frac{p}{q}\right\rfloor +1\right)+\left(\left\lfloor \left(i+1\right)\frac{p}{q}\right\rfloor +1\right)\\
 & =\left\lfloor \left(i+1\right)\frac{p}{q}\right\rfloor -\left\lfloor i\frac{p}{q}\right\rfloor .
\end{align*}
Thus, (\ref{eq:=000020lem-period=000020suffixes=000020-=0000202})
follows from (\ref{eq:=000020lower=000020mechanical=000020word})
in Lemma~\ref{lem:=000020lower=000020mechanical=000020word}. We
now proceed to prove the lemma using the above.
\begin{enumerate}
\item Assume that $k\equiv0\mod 2$. For $2\leq i\leq q_{k-1}-1$, we have
\begin{align*}
W_{k}(q_{k}-i) & =W_{k}(i-1)=W_{k-1}(i-1)=W_{k-1}(q_{k-1}-i),
\end{align*}
where the first and third equalities follow from (\ref{eq:=000020lem-period=000020suffixes=000020-=0000202}),
and the second equality follows from (\ref{eq:=000020lem-period=000020initials=000020-=000020even=000020k=000020-=0000201})
in Lemma~\ref{lem:=000020period=000020prefixes}. To finish this
part of the proof we only need to show that (\ref{eq:=000020lem-period=000020suffixes=000020-=000020even=000020k})
holds for $i=q_{k-1}$, i.e. that $W_{k}(q_{k}-q_{k-1})=W_{k-1}(0)=0$.\\
Using \cite[thm. 2]{Khinchin_book64} we calculate 
\begin{align*}
q_{k-1}\frac{p_{k}}{q_{k}} & =\frac{1}{q_{k}}\left(q_{k-1}p_{k}-q_{k}p_{k-1}\right)+p_{k-1}=-\frac{1}{q_{k}}+p_{k-1}.
\end{align*}
Hence,
\begin{align*}
W_{k}(q_{k}-q_{k-1}) & =\left\lfloor (q_{k}-q_{k-1}+1)\frac{p_{k}}{q_{k}}\right\rfloor -\left\lfloor (q_{k}-q_{k-1})\frac{p_{k}}{q_{k}}\right\rfloor \\
 & =\left\lfloor p_{k}-p_{k-1}+\frac{1}{q_{k}}+\frac{p_{k}}{q_{k}}\right\rfloor -\left\lfloor p_{k}-p_{k-1}+\frac{1}{q_{k}}\right\rfloor \\
 & =\left\lfloor \frac{p_{k}+1}{q_{k}}\right\rfloor -\left\lfloor \frac{1}{q_{k}}\right\rfloor =0,
\end{align*}
follows where in the last equality we used that $q_{k}>p_{k}+1$,
which holds if $k>2$ or if $k=2$ and $c_{1}>1$ (which we can assume
since we have already dealt the case $k=2$, $c_{1}=1$ in the beginning
of the proof).
\item Assume that $k\equiv1\mod 2$. For $2\leq i\leq q_{k-2}-1$ we have
\begin{align*}
W_{k}(q_{k}-i) & =W_{k}(i-1)=W_{k-2}(i-1)=W_{k-2}(q_{k-2}-i),
\end{align*}
where the first and third equalities follow from (\ref{eq:=000020lem-period=000020suffixes=000020-=0000202}),
and the second equality follows from (\ref{eq:=000020lem-period=000020initials=000020-=000020odd=000020k=000020-=0000202})
in Lemma~\ref{lem:=000020period=000020prefixes}. To finish this
part of the proof we only need to show that (\ref{eq:=000020lem-period=000020suffixes=000020-=000020odd=000020k})
holds for $i=q_{k-2}$, i.e. that $W_{k}(q_{k}-q_{k-2})=W_{k-2}(0)=0$.\\
Using \cite[thm. 3]{Khinchin_book64} we calculate 
\begin{align*}
q_{k-2}\frac{p_{k}}{q_{k}} & =\frac{1}{q_{k}}\left(q_{k-2}p_{k}-q_{k}p_{k-2}\right)+p_{k-2}=-\frac{c_{k}}{q_{k}}+p_{k-2}.
\end{align*}
Hence, we conclude
\begin{align*}
W_{k}(q_{k}-q_{k-2}) & =\left\lfloor (q_{k}-q_{k-2}+1)\frac{p_{k}}{q_{k}}\right\rfloor -\left\lfloor (q_{k}-q_{k-2})\frac{p_{k}}{q_{k}}\right\rfloor \\
 & =\left\lfloor p_{k}-p_{k-2}+\frac{c_{k}}{q_{k}}+\frac{p_{k}}{q_{k}}\right\rfloor -\left\lfloor p_{k}-p_{k-2}+\frac{c_{k}}{q_{k}}\right\rfloor \\
 & =\left\lfloor \frac{p_{k}+c_{k}}{q_{k}}\right\rfloor -\left\lfloor \frac{c_{k}}{q_{k}}\right\rfloor .
\end{align*}
To conclude $W_{k}(q_{k}-q_{k-2})=0$, we now show $\frac{c_{k}}{q_{k}}<1$
and $\frac{p_{k}+c_{k}}{q_{k}}<1$. The recursions (\ref{eq:=000020recursion=000020for=000020p_k,=000020q_k})
lead to
\[
\frac{c_{k}}{q_{k}}=\frac{1}{q_{k}}\frac{q_{k}-q_{k-2}}{q_{k-1}}<1
\]
and 
\[
\frac{p_{k}+c_{k}}{q_{k}}=\frac{c_{k}\left(p_{k-1}+1\right)+p_{k-2}}{c_{k}q_{k-1}+q_{k-2}}<1,
\]
where to get the last inequality we observe that for $k\geq3$ (recalling
that $k\geq2$ and we consider now odd $k$ values) $p_{k-2}\leq q_{k-2}$
and $p_{k-1}+1\leq q_{k-1}$, and equality in both of these may be
achieved only if $k=3$ and $c_{1}=1$ (which yields $p_{1}=q_{1}=1$
and $p_{2}+1=c_{2}=q_{2}$), but we have already dealt with this case
in the beginning of the proof.
\end{enumerate}
\end{proof}
\begin{lem}
\label{lem:=000020sub-periods=000020of=000020periods}[Sub-periods of the period]
Let $\alpha\in[0,1]\setminus\Q$ and $k\geq2$.
\begin{enumerate}
\item if $k\equiv0\mod 2$, then 
\begin{equation}
W_{k}(i)=W_{k}(i+q_{k-1}\mod{q_{k}})\qquad\textrm{for all }1\leq i\leq q_{k}-2.\label{eq:=000020lem-sub-periods=000020of=000020periods=000020-=000020k=000020even}
\end{equation}
\item if $k\equiv1\mod 2$, then 
\begin{equation}
W_{k}(i\mod{q_{k}})=W_{k}(i+q_{k-1}\mod{q_{k}})\qquad\textrm{for all }-q_{k-1}+1\leq i\leq q_{k}-q_{k-1}-2.\label{eq:=000020lem-sub-periods=000020of=000020periods=000020-=000020k=000020odd}
\end{equation}
\end{enumerate}
\end{lem}

\begin{proof}
As before we use frequently (\ref{eq:=000020lower=000020mechanical=000020word})
in Lemma~\ref{lem:=000020lower=000020mechanical=000020word}.
\begin{enumerate}
\item Let $k\equiv0\mod 2$. Let $1\leq i\leq q_{k}-2$. Since $p_{k},q_{k}$
are coprime, we conclude $j\frac{p_{k}}{q_{k}}\notin\Z$ and $\fl{j\frac{p_{k}}{q_{k}}-\frac{1}{q_{k}}}=\fl{j\frac{p_{k}}{q_{k}}}$
for all $1\leq j\leq q_{k}-1$. Then
\begin{align*}
W_{k}(i+q_{k-1}\mod{q_{k}}) & =\fl{(i+1+q_{k-1})\frac{p_{k}}{q_{k}}}-\fl{(i+q_{k-1})\frac{p_{k}}{q_{k}}}\\
 & =\fl{(i+1)\frac{p_{k}}{q_{k}}+\left(q_{k-1}\frac{p_{k}}{q_{k}}-p_{k-1}\right)}-\fl{i\frac{p_{k}}{q_{k}}+\left(q_{k-1}\frac{p_{k}}{q_{k}}-p_{k-1}\right)}\\
 & =\fl{(i+1)\frac{p_{k}}{q_{k}}-\frac{1}{q_{k}}}-\fl{i\frac{p_{k}}{q_{k}}-\frac{1}{q_{k}}}\\
 & =\fl{(i+1)\frac{p_{k}}{q_{k}}}-\fl{i\frac{p_{k}}{q_{k}}}=W_{k}(i)
\end{align*}
follows where we used \cite[thm. 2]{Khinchin_book64} in the third
equality.
\item Let $k\equiv1\mod 2$. Let $-q_{k-1}+1\leq i\leq q_{k}-q_{k-1}-2$.
\begin{align}
W_{k}(i+q_{k-1}) & =\fl{(i+1+q_{k-1})\frac{p_{k}}{q_{k}}}-\fl{(i+q_{k-1})\frac{p_{k}}{q_{k}}}\nonumber \\
 & =\fl{(i+1)\frac{p_{k}}{q_{k}}+\left(q_{k-1}\frac{p_{k}}{q_{k}}-p_{k-1}\right)}-\fl{i\frac{p_{k}}{q_{k}}+\left(q_{k-1}\frac{p_{k}}{q_{k}}-p_{k-1}\right)}\nonumber \\
 & =\fl{(i+1)\frac{p_{k}}{q_{k}}+\frac{1}{q_{k}}}-\fl{i\frac{p_{k}}{q_{k}}+\frac{1}{q_{k}}}.\label{eq:=000020lem-sub-periods=000020of=000020periods=000020-=0000201}
\end{align}
Thus, it suffices to prove (quite similarly to the proof of the first
part) that for all $-q_{k-1}+1\leq j\leq q_{k}-q_{k-1}-1$,
\[
\fl{j\frac{p_{k}}{q_{k}}+\frac{1}{q_{k}}}=\fl{j\frac{p_{k}}{q_{k}}}.
\]
 Assume towards contradiction that $\fl{j\frac{p_{k}}{q_{k}}+\frac{1}{q_{k}}}\neq\fl{j\frac{p_{k}}{q_{k}}}$
for some $-q_{k-1}+1\leq j\leq q_{k}-q_{k-1}-1$. This means that
there exists $m\in\Z$ such that $j\frac{p_{k}}{q_{k}}+\frac{1}{q_{k}}=m$.
Therefore, 
\begin{equation}
jp_{k}\equiv-1\mod q_{k}.\label{eq:=000020lem-sub-periods=000020of=000020periods=000020-=0000203}
\end{equation}
By \cite[thm. 2]{Khinchin_book64} $-q_{k-1}p_{k}+p_{k-1}q_{k}=-1$.
Thus, $j=-q_{k-1}$ is a solution to (\ref{eq:=000020lem-sub-periods=000020of=000020periods=000020-=0000203}).
In fact, since $p_{k},q_{k}$ are coprime any solution to (\ref{eq:=000020lem-sub-periods=000020of=000020periods=000020-=0000203})
satisfies $j\equiv-q_{k-1}\mod q_{k}$. We now obtain a contradiction
since there is no such value in the range $j\in\{-q_{k-1}+1,\ldots,q_{k}-q_{k-1}-1\}$.
\end{enumerate}
\end{proof}
We now combine the last three lemmas to prove Lemma~\ref{lem:=000020WordRecursion}.
\begin{proof}
[Proof of Lemma \ref{lem:=000020WordRecursion}]A short computation
leads to $W_{0}=0$ with $\alpha_{0}=\frac{0}{1}$ and $W_{1}=\underbrace{0\ldots0}_{c_{1}-1}1$
with $\alpha_{1}=\frac{1}{c_{1}}$.

\begin{enumerate}
\item Let $k\equiv0\mod 2$. We first treat the case $k=2$ and $c_{1}=1$.
In this case, we have $W_{0}=0$ and $W_{1}=1$.. Then the recursion
relation (\ref{eq:=000020recursion=000020for=000020p_k,=000020q_k})
asserts $\alpha_{1}=\frac{p_{1}}{q_{1}}=\frac{1}{1}$ and $\alpha_{2}=\frac{p_{2}}{q_{2}}=\frac{c_{2}}{c_{2}+1}$.
Thus, $W_{2}=0\underbrace{1\ldots1}_{c_{2}}=W_{0}W_{1}^{c_{2}}$ follows
as claimed by (\ref{eq:=000020def=000020Sturmian=000020sequence})
and (\ref{eq:=000020Sturmian=000020period=000020def}). Therefore,
we can from now on assume that if $k=2$, then $c_{1}>1$. \\
Next, observe that $q_{k-1}=q_{k-2}$ can happen only for $k=2$ and
$c_{1}=1$ (as can be verified from (\ref{eq:=000020recursion=000020for=000020p_k,=000020q_k})).
Therefore, we may continue the proof assuming that $q_{k-1}>q_{k-2}$.\\
Applying (\ref{eq:=000020lem-period=000020initials=000020-=000020even=000020k=000020-=0000202})
in Lemma~\ref{lem:=000020period=000020prefixes} establishes the
required statement for the prefix of $W_{k}$, i.e.
\[
W_{k-2}=\left.W_{k}\right|_{\left\{ 0,\ldots,q_{k-2}-1\right\} }=\left.W_{k}\right|_{\left\{ 0,\ldots,q_{k}-c_{k}q_{k-1}-1\right\} },
\]
where we used $q_{k-1}>q_{k-2}$ and the recursive relation (\ref{eq:=000020recursion=000020for=000020p_k,=000020q_k})
of $\left\{ q_{k}\right\} $.\\
Applying (\ref{eq:=000020lem-period=000020suffixes=000020-=000020even=000020k})
in Lemma~\ref{lem:=000020period=000020suffixes} gives
\[
\left.W_{k}\right|_{\left\{ q_{k}-q_{k-1},\ldots,q_{k}-1\right\} }=W_{k-1}.
\]
The last equality together with (\ref{eq:=000020lem-sub-periods=000020of=000020periods=000020-=000020k=000020even})
in Lemma~\ref{lem:=000020sub-periods=000020of=000020periods} yields
\[
\left.W_{k}\right|_{\left\{ q_{k}-c_{k}q_{k-1},\ldots,q_{k}-1\right\} }=W_{k-1}^{c_{k}},
\]
and completes the proof of this part.
\item Let $k\equiv1\mod 2$. \\
Applying (\ref{eq:=000020lem-period=000020suffixes=000020-=000020odd=000020k})
in Lemma~\ref{lem:=000020period=000020suffixes} establishes the
required statement for the suffix of $W_{k}$, i.e.
\[
W_{k-2}=\left.W_{k}\right|_{\left\{ q_{k}-q_{k-2},\ldots,q_{k}-1\right\} }=\left.W_{k}\right|_{\left\{ c_{k}q_{k-1},\ldots,q_{k}-1\right\} },
\]
where we used the recursive relation (\ref{eq:=000020recursion=000020for=000020p_k,=000020q_k})
of $\left\{ q_{k}\right\} $. Applying (\ref{eq:=000020lem-period=000020initials=000020-=000020odd=000020k=000020-=0000201})
in Lemma~\ref{lem:=000020period=000020prefixes} gives
\[
\left.W_{k}\right|_{\left\{ 0,\ldots,q_{k-1}-1\right\} }=W_{k-1}.
\]
We observe that in this case we have $k\geq3$ and $q_{k-2}\geq1$.
Hence, the last equality together with (\ref{eq:=000020lem-sub-periods=000020of=000020periods=000020-=000020k=000020odd})
in Lemma~\ref{lem:=000020sub-periods=000020of=000020periods} yields
\[
\left.W_{k}\right|_{\left\{ 0,\ldots,c_{k}q_{k-1}-1\right\} }=W_{k-1}^{c_{k}},
\]
and completes the proof.
\end{enumerate}
\end{proof}

\section{Floquet-Bloch theory via finite-dimensional Hamiltonian matrices
\protect\label{sec:=000020Floquet-Bloch=000020Theory}}

This appendix complements Section~\ref{sec:=000020Transfer=000020Matrices=000020and=000020Discriminant}
by providing an alternative approach for the spectral analysis of
the periodic operators, $H_{\frac{p}{q},V}$. In Section~\ref{subsec:=000020spectrum=000020via=000020transfer=000020matrices}
the Floquet-Bloch theory is described in terms of transfer matrices
and the discriminant, whereas here we make use of finite Hamiltonian
matrices $\Hc(\theta)$. These matrices $\Hc(\theta)$ play a crucial
role in \cite{BanBecLoe_24} and henceforth it is advantageous to
introduce them already here and make the appropriate connection to
the transfer matrices.

We use here the continued fraction notation, $\co\in\Co$ and denote
$\frac{p}{q}:=\varphi(\co)$. The corresponding operator is 
\[
(\Hrat\psi)(n):=\psi(n+1)+\psi(n-1)+V\ohrat(n)\thinspace\psi(n),
\]
where the potential is given by the mechanical word,
\[
\ohrat(n):=\chi_{\left[1-\frac{p}{q},1\right)}\left(n\frac{p}{q}\mod 1\right),
\]
which is $q$ periodic (see Section~\ref{subsec:=000020Sturmian=000020words=000020and=000020mechanical=000020words}).
To describe the relevant Floquet-Bloch theory, we define the following
finite-dimensional auxiliary matrix

\begin{equation}
\Hc(\theta):=\begin{cases}
2\cos(\theta)+V\ohrat(0), & q=1,\\
\\\begin{pmatrix}V\ohrat(0) & 1+e^{-i\theta}\\
1+e^{i\theta} & V\ohrat(1)
\end{pmatrix}, & q=2,\\
\\\begin{pmatrix}V\ohrat(0) & 1 & 0 & \ldots &  & e^{-i\theta}\\
1 & V\ohrat(1) & 1 & \ldots &  & 0\\
0 & 1 & \ddots &  &  & \vdots\\
\vdots & \ddots &  & \ddots &  & 0\\
0 &  &  &  &  & 1\\
e^{i\theta} & 0 & \cdots & 0 & 1 & V\ohrat(q-1)
\end{pmatrix}, & q>2.
\end{cases}\label{eq:=000020finite-dim=000020Ham=000020matrices}
\end{equation}
 The characteristic polynomial of the matrix above is denoted by 
\[
\polyc(\theta;E):=\det(E-\Hc(\theta)).
\]
Using the auxiliary matrices defined above, we get the following by
standard Floquet-Bloch theory (see e.g., \cite[sec. 7.2]{Tes00},
\cite[sec. 5.3]{Simon2011}).
\begin{prop}
\label{prop:=000020Floquet-Bloch} Let $V\in\R$ and $\co\in\Co$
with $\frac{p}{q}=\varphi(\co)\neq\infty$.
\begin{enumerate}
\item The spectrum of $\Hrat$ is given by
\[
\sigma(\Hrat)=\bigcup_{\theta\in[0,\pi]}\sigma(\Hc(\theta)).
\]
\item Suppose $p$ and $q$ are coprime. Denoting the roots of $\polyc(\theta;\cdot)$
by $\left\{ \lambda_{i}^{\left(\theta\right)}\right\} _{i=1}^{q}$
we have 
\begin{equation}
\lambda_{q-1}^{\left(0\right)}>\lambda_{q-1}^{\left(\pi\right)}>\lambda_{q-2}^{\left(\pi\right)}>\lambda_{q-2}^{\left(0\right)}>\lambda_{q-3}^{\left(0\right)}>\lambda_{q-3}^{\left(\pi\right)}>\ldots\label{eq:=000020Floquet-Bloch=000020-=000020evalues=000020inequalities}
\end{equation}
and get that $\sigma(\Hrat)$ is the following union of $q$ disjoint
closed intervals 
\begin{equation}
\sigma(\Hrat)=\ldots\sqcup\left[\lambda_{q-3}^{\left(\pi\right)},\lambda_{q-3}^{\left(0\right)}\right]\sqcup\left[\lambda_{q-2}^{\left(0\right)},\lambda_{q-2}^{\left(\pi\right)}\right]\sqcup\left[\lambda_{q-1}^{\left(\pi\right)},\lambda_{q-1}^{\left(0\right)}\right],\label{eq:=000020Floquet-Bloch=000020-=000020spectral=000020bands=000020decomposition}
\end{equation}
which are commonly called \emph{spectral bands}.
\end{enumerate}
\end{prop}

The general statement of Proposition~\ref{prop:=000020Floquet-Bloch}
within Floquet-Bloch theory is with weak inequalities in (\ref{eq:=000020Floquet-Bloch=000020-=000020evalues=000020inequalities})
and possible intersections of the spectral bands in (\ref{eq:=000020Floquet-Bloch=000020-=000020spectral=000020bands=000020decomposition})
at their edges. Specifically, in our case where the potential is given
by $\ohrat(n)$, this slightly stronger version holds since $p,q$
are coprime - a proof is found in Proposition~\ref{prop:=000020there=000020are=000020q=000020spectral=000020bands}
using transfer matrices.

Since Floquet-Bloch theory may be described either in terms of transfer
matrices (as in Section~\ref{subsec:=000020spectrum=000020via=000020transfer=000020matrices})
and in terms of finite-dimensional Hamiltonian matrices (as in this
appendix), it makes sense to draw a direct connection between both.
Hence, we explicitly state the connection between the trace $\tc$
of the transfer matrix (i.e., the discriminant) and the characteristic
polynomial $\polyc$:

\begin{lem}
\label{lem:PropCharacPoly} For all $\theta\in[0,2\pi]$, 
\[
\polyc(\theta;E)=\tc(E,V)-2\cos(\theta).
\]
\end{lem}

A standard way to prove the identity in the lemma is to develop the
Floquet-Bloch theory using both the discriminant $\tc$ and the characteristic
polynomial $\polyc$ and note that these two polynomials have common
roots. See for example \cite[thm. 5.4.1,(iii)]{Simon2011}. Nevertheless,
we bring here a direct computational proof\footnote{An idea towards such a proof is also found in remark 3 after \cite[thm. 5.4.1]{Simon2011}.}
which exploits the structure of the matrix $\Hc(\theta)$.
\begin{proof}
As usual, denote $\frac{p}{q}:=\varphi(\co)$, with coprime $p,q$.
We first prove the statement assuming $q\geq3$ and at the end check
that it holds also for the cases $q=1$ and $q=2$.Start by examining
$\polyc(\theta;E)+2\cos(\theta)=\det\left(E\id-\Hc(\theta)\right)+2\cos(\theta)$
and decomposing it into summands. We use the Leibniz formula for determinants
to get 
\begin{equation}
\polyc(\theta;E)=\sum_{\sigma\in S_{q}}\sgn(\sigma)\prod_{n=1}^{q}\left[E\id-\Hc(\theta)\right]{}_{n,\sigma(n)},\label{eq:=000020Char=000020Poly=000020-=000020Leibnitz=000020formula}
\end{equation}
where $S_{q}$ is the set of all permutations on $\left[q\right]:=\{1,2,\ldots,q\}$.
We examine only permutations with a non-vanishing contribution to
the sum above. Let $\sigma\in S_{q}$ be such permutation and $n\in\left[q\right]$.
We have that $\left[E\id-\Hc(\theta)\right]{}_{n,\sigma(n)}\neq0$
only if $\sigma(n)\in\{n-1,n,n+1\}$ (noting that we consider a cyclic
ordering of the indices in the set $\left[q\right]$, such that if
$n=1$ then $n-1:=q$ and if $n=q$ then $n+1=1$). If $\sigma(n)=n+1$,
then we can have either $\sigma(n+1)=n$ or $\sigma(n+1)=n+2$ (so
that the corresponding product in (\ref{eq:=000020Char=000020Poly=000020-=000020Leibnitz=000020formula})
differs than zero). In the first case, we see that the permutation
$\sigma$ contains an involution, $\left(n\quad n+1\right)$. The
second case imposes that $\sigma$ is the cyclic permutation, $\sigma_{\mathrm{cyc}}^{+}=(1\quad2\quad\ldots\quad q-1\quad q)$,
as all other permutations which satisfy both $\sigma(n)=n+1$ and
$\sigma(n+1)=\sigma(n+2)$ have a vanishing contribution to (\ref{eq:=000020Char=000020Poly=000020-=000020Leibnitz=000020formula}).
Explicitly the contribution of $\sigma_{\mathrm{cyc}}^{+}$ to this
sum is
\[
\sign(\sigma_{\mathrm{cyc}}^{+})\left(\prod_{n=1}^{q-1}\left[E\id-\Hc(\theta)\right]{}_{n,n+1}\right)\left[-\Hc(\theta)\right]_{q,1}=(-1)^{q+1}(-1)^{q-1}(-\ue^{\ui\theta})=-\ue^{\ui\theta}.
\]
 If we repeat the arguments above for the case $\sigma(n)=n-1$ we
get that either $\sigma$ contains the involution $\left(n-1\quad n\right)$
or that it is the cyclic permutation $\sigma_{\mathrm{cyc}}^{-}=(q\quad q-1\quad\ldots\quad2\quad1)$
whose contribution to (\ref{eq:=000020Char=000020Poly=000020-=000020Leibnitz=000020formula})
is $-\ue^{-\ui\theta}$. Hence, the contribution of both $\sigma_{\mathrm{cyc}}^{+}$
and $\sigma_{\mathrm{cyc}}^{-}$ sums to $-2\cos(\theta)$. All other
permutations with non-vanishing contribution to (\ref{eq:=000020Char=000020Poly=000020-=000020Leibnitz=000020formula})
contain only involutions of the form $\left(n-1\quad n\right)$ or
fixed points $\left(n\right)$. We denote the set of such permutations
by $\widetilde{S_{q}}$ and summarize the discussion so far by writing
\begin{equation}
\polyc(\theta;E)+2\cos(\theta)=\sum_{\sigma\in\widetilde{S}_{q}}(-1)^{|I(\sigma)|}\prod_{i\in F(\sigma)}\left(E-V\ohrat(n)\right),\label{eq:=000020char=000020poly=000020expansion}
\end{equation}
where $I(\sigma)$ is the set of involutions $(i\quad i+1)$ of $\sigma$
and $F(\sigma)$ is the set of fixed points of $\sigma$.

Now, we consider $\tc(E,V)$ and decompose it into summands. To do
so, we recall (see Subsection~\ref{subsec:=000020spectrum=000020via=000020transfer=000020matrices})
the definition of $\Mc(E,V)$ as the product of one step transfer
matrices,
\[
A_{\alpha}(n)(E,V):=\begin{pmatrix}E-V\omega_{\alpha}(n) & -1\\
1 & 0
\end{pmatrix},
\]
and write
\begin{align}
\tc(E,V) & =\text{tr}\left(\Mc(E,V)\right)\\
 & =\text{tr}\left(\prod_{n=1}^{q}A_{\varphi(\co)}(n)(E,V)\right)\\
 & =\sum_{\nu\in\left\{ 1,2\right\} ^{q}}\prod_{n=1}^{q}\left[A_{\varphi(\co)}(n)(E,V)\right]_{\nu_{n},\nu_{n+1}},
\end{align}
where we have the interpretation $\nu_{q+1}:=\nu_{1}$ due to the
cyclic property of the trace. In the sum above, a summand which corresponds
to $\nu\in\left\{ 1,2\right\} ^{q}$ is non-zero if and only if there
is no $n$ such that $\nu_{n}=\nu_{n+1}=2$. We denote the set of
all such $\nu\in\left\{ 1,2\right\} ^{q}$ with non-vanishing contribution
by $\mathcal{\widetilde{N}}_{q}$, so that 
\begin{equation}
\tc(E,V)=\sum_{\nu\in\mathcal{\widetilde{N}}_{q}}\prod_{n=1}^{q}\left[A_{\varphi(\co)}(n)(E,V)\right]_{\nu_{n},\nu_{n+1}}.\label{eq:=000020trace=000020expansion}
\end{equation}
For the last part of this proof, we show a bijection $h:\widetilde{S}_{q}\rightarrow\mathcal{\widetilde{N}}_{q}$
such that the contribution of $\sigma\in\widetilde{S}_{q}$ to (\ref{eq:=000020char=000020poly=000020expansion})
equals the contribution of $h(\sigma)$ to (\ref{eq:=000020trace=000020expansion}).
We explicitly construct this bijection as follows: for any fixed point
$n\in F(\sigma)$ we set 
\[
h(\sigma)_{n}=h(\sigma)_{n+1}=1,
\]
and for any involution $\left(n\quad n+1\right)\in I(\sigma)$ we
set 
\[
h(\sigma)_{n}=1,\quad h(\sigma)_{n+1}=2,\quad h(\sigma)_{n+1}=1.
\]
First, the map $h:\widetilde{S}_{q}\rightarrow\mathcal{\widetilde{N}}_{q}$
is well defined, as no two subsequent entries of $h(\sigma)$ may
be equal to $2$. Furthermore, one can see that it is a bijection
and for each $\nu\in\mathcal{\widetilde{N}}_{q}$ one can uniquely
construct the corresponding $\sigma\in\widetilde{S}_{q}$ such that
$h(\sigma)=\nu$. Finally, it is also not hard to check that the contribution
to the corresponding sum ((\ref{eq:=000020char=000020poly=000020expansion})
or (\ref{eq:=000020trace=000020expansion})) is preserved under the
map $h$.

We end the proof by checking that the statement holds for the particular
cases of $q=1$ and $q=2$.

For $q=1$ we have 
\[
\polyc(\theta;E)=E-\left(2\cos+V\ohrat(1)\right),
\]
and 
\[
\tc(E,V)=\tr\left(A_{\pq}(1)(E,V)\right)=\tr\begin{pmatrix}E-V\omega_{\pq}(1) & -1\\
1 & 0
\end{pmatrix}=E-V\omega_{\pq}(1).
\]
For $q=2$ we have 
\begin{align*}
\polyc(\theta;E) & =\det\begin{pmatrix}E-V\ohrat(1) & -\left(1+e^{-i\theta}\right)\\
-\left(1+e^{i\theta}\right) & E-V\ohrat(2)
\end{pmatrix}\\
 & =\left(E-V\omega_{\pq}(2)\right)\left(E-V\omega_{\pq}(1)\right)-2-2\cos(\theta),
\end{align*}
and 
\begin{align*}
\tc(E,V) & =\tr\left(A_{\pq}(2)(E,V)\cdot A_{\pq}(1)(E,V)\right)\\
 & =\tr\left[\begin{pmatrix}E-V\omega_{\pq}(2) & -1\\
1 & 0
\end{pmatrix}\begin{pmatrix}E-V\omega_{\pq}(1) & -1\\
1 & 0
\end{pmatrix}\right]\\
 & =\tr\left[\begin{pmatrix}\left(E-V\omega_{\pq}(2)\right)\left(E-V\omega_{\pq}(1)\right)-1 & -E+V\omega_{\pq}(2)\\
E-V\omega_{\pq}(1) & -1
\end{pmatrix}\right]\\
 & =\left(E-V\omega_{\pq}(2)\right)\left(E-V\omega_{\pq}(1)\right)-2.
\end{align*}
\end{proof}

\section{Dilated Chebychev polynomials of the second kind \protect\label{sec:=000020Cheby=000020poly}}

In this appendix we collect proofs to the statements and identities
around the dilated Chebychev polynomials of second kind. Recall that
we defined these polynomials recursively by setting 
\[
S{-1}(x):=0,\quad S_{0}(x):=1\quad\text{and }\quad S_{n}(x):=xS_{n-1}(x)-S_{n-2}(x)\text{ for all }n\in\N.
\]
We also remind the reader that the classical Chebychev polynomials
of second kind can be defined using the recursion formula 
\[
U_{-1}(x):=0,\quad U_{0}(x):=1\quad\text{and }\quad U_{n}(x):=2xU_{n-1}(x)-U_{n-2}(x)\text{ for all }n\in\N.
\]

\begin{lem}
For all $n\in\Nmo$ and all $x\in\R$ we have $S_{n}(2x)=U_{n}(x)$.
\end{lem}

\begin{proof}
We perform a proof by induction over $n\in\Nmo$. For $n=-1$ and
$n=0$ the statement follows directly from the definition. Therefore
let $n\in\N$ and assume $S_{n-1}(2x)=U_{n-1}(x)$ and $S_{n-2}(2x)=U_{n-2}(x)$
for all $x\in\R$. Then we get 
\begin{align*}
S_{n}(2x)=2xS_{n-1}(2x)-S_{n-2}(2x)=2xU_{n-1}(x)-U_{n-2}(x)=U_{n}(x).
\end{align*}
\end{proof}
\begin{lem}
\label{Lem-Chebyshev-Traces} Let $x\in\mathbb{R}$ and $n\in\N_{0}$.
Then the following holds.
\begin{enumerate}
\item \label{enu:Chebyshev-Traces_invariant}We have $S_{n+1}(x)S_{n-1}(x)-S_{n}(x)^{2}=-1$.
\item \label{enu:Chebyshev-Traces_x_=00003D2}If $|x|=2$, then $\sgn(x)^{n-1}S_{n-1}(x)=n$.
\item If $|x|\geq2$, then $2|S_{n}(x)|-|S_{n-1}|\geq0$.
\item \label{enu:Chebyshev-Traces_x_geq_2}If $|x|\geq2$, then $\sgn(x)^{n}S_{n}(x)=|S_{n}(x)|$
and 
\[
\sgn(x)^{n}xS_{n-1}(x)\geq2\big|S_{n-1}(x)\big|.
\]
\item \label{enu:Chebyshev-Traces_x_geq_2-sign}If $|x|\geq2$, then 
\[
\sgn(x)^{n}\big(S_{n}(x)-\frac{x}{2}S_{n-1}(x)\big)\geq1.
\]
\item \label{enu:Chebyshev-Traces_x_geq_2-Sn(x)_geq_1}If $|x|\geq2$\textup{,
then $|S_{n}(x)|\geq1$.}
\item \label{enu:Chebyshev-Traces_x>2}If $|x|>2$ and $n\geq1$, then 
\[
\sgn(x)^{n}\big(S_{n}(x)-\frac{x}{2}S_{n-1}(x)\big)>1.
\]
\end{enumerate}
\end{lem}

\begin{proof}
We prove each statement by an induction over $n$.
\begin{enumerate}
\item For $n=0$ and $n=1$, observe 
\[
S_{1}(x)S_{-1}(x)-S_{0}(x)^{2}=-1^{2}=-1,
\]
and 
\[
S_{2}(x)S_{0}(x)-S_{1}(x)^{2}=(x^{2}-1)-x^{2}=-1.
\]
Suppose the statement is true for $n\in\N$ and $n-1$, then 
\begin{align*}
S_{n+1}S_{n-1}-S_{n}^{2} & =\left(xS_{n}-S_{n-1}\right)S_{n-1}-S_{n}^{2}\\
 & =S_{n}\underbrace{\left(xS_{n-1}-S_{n}\right)}_{=S_{n-2}}-S_{n-1}^{2}\\
 & =S_{n}S_{n-2}-S_{n}=-1
\end{align*}
follows.
\item Let $|x|=2$. For $n=0$ and $n=1$, observe in these cases 
\[
\sgn(x)^{-1}S_{-1}(x)=0\quad\textrm{and}\quad\sgn(x)^{0}S_{0}(x)=1.
\]
Suppose the statement is true for $n\in\N$ and $n-1$, then 
\begin{align*}
\sgn(x)^{n+1}S_{n+1} & =\sgn(x)^{n+1}\left(xS_{n}-S_{n-1}\right)\\
 & =|x|(n+1)-n=2(n+1)-n=(n+1)+1.
\end{align*}
\item Let $|x|\geq2$. If $n=0$, then $2|S_{n}|-|S_{n-1}|=2-0\geq0$. Suppose
the statement is true for $n\in\N_{0}$. Then 
\begin{align*}
2|S_{n+1}|-|S_{n}|=2|xS_{n}-S_{n-1}|-|S_{n}|\geq & 2|x|\cdot|S_{n}|-|S_{n-1}|-|S_{n}|\\
\geq & 4|S_{n}|-|S_{n-1}|-|S_{n}|\\
\geq & 2|S_{n}|-|S_{n-1}|\geq0
\end{align*}
by induction hypothesis.
\item Again, let $|x|\geq2$ an consider $n=0$ and $n=1$ for the induction
base. Then 
\[
\sgn(x)^{0}S_{0}(x)=1=|S_{0}(x)|\quad\textrm{and}\quad\sgn(x)^{1}S_{1}(x)=\sgn(x)\cdot x=|x|=|S_{1}(x)|.
\]
Suppose it holds for $n\in\N$. Then 
\begin{align*}
\sgn(x)^{n+1}S_{n+1}(x) & =\sgn(x)^{n+1}\left(xS_{n}-S_{n-1}\right)\\
 & =|x|\cdot|S_{n}|-|S_{n-1}|\geq2|S_{n}|-|S_{n-1}|\geq0,
\end{align*}
where the last follows by the previous induction. Hence, $\sgn(x)^{n+1}S_{n+1}(x)=|S_{n+1}(x)|$
follows proving the first part of (\ref{enu:Chebyshev-Traces_x_geq_2}).
Moreover, this and $|x|\geq2$ lead to 
\[
\sgn(x)^{n}xS_{n-1}(x)=|x|\big|S_{n-1}(x)\big|\geq2\big|S_{n-1}(x)\big|
\]
proving the second part of (\ref{enu:Chebyshev-Traces_x_geq_2}).
\item Let $|x|\geq2$ and suppose $n=0$ for the induction base. Then 
\[
\sgn(x)^{0}\left(S_{0}(x)-\frac{x}{2}S_{-1}(x)\right)=1.
\]
Suppose the statement is true for $n\in\N_{0}$. Then 
\begin{align*}
 & \sgn(x)^{n+1}\big(S_{n+1}(x)-\frac{x}{2}S_{n}(x)\big)\\
= & \underbrace{\sgn(x)\frac{x}{2}}_{\geq1\textrm{ if }|x|\geq2}\sgn(x)^{n}\big(S_{n}(x)-\frac{x}{2}S_{n-1}(x)\big)+\underbrace{\left(\frac{x^{2}}{4}-1\right)}_{\geq0}\underbrace{\sgn(x)^{n-1}S_{n-1}(x)\big)}_{\geq0\textrm{ by (b)}}\\
\geq & \sgn(x)^{n}\big(S_{n}(x)-\frac{x}{2}S_{n-1}(x)\big)
\end{align*}
follows. Thus, the induction hypothesis implies the desired claim.
\item Let $|x|\geq2$ and $n\in\N_{0}$. Then (\ref{enu:Chebyshev-Traces_x_geq_2})
and (\ref{enu:Chebyshev-Traces_x_geq_2-sign}) imply
\[
|S_{n}(x)|=\sgn(x)^{n}S_{n}(x)\geq1+\sgn(x)^{n}\frac{x}{2}S_{n-1}(x)\geq1.
\]
\item Let $|x|>2$. If $n=1$, then 
\[
\sgn(x)^{1}\left(S_{1}(x)-\frac{x}{2}S_{0}(x)\right)=\sgn(x)\big(x-\frac{x}{2}\big)=\frac{|x|}{2}>1
\]
follows. A similar computation as in (\ref{enu:Chebyshev-Traces_x_geq_2-sign})
leads to 
\[
\sgn(x)^{n+1}\big(S_{n+1}(x)-\frac{x}{2}S_{n}(x)\big)\geq\sgn(x)^{n}\big(S_{n}(x)-\frac{x}{2}S_{n-1}(x)\big)>1,
\]
where the last estimate follows by the induction hypothesis.
\end{enumerate}
\end{proof}
\begin{lem}
\label{lem:=000020Cheby=000020Poly=000020-=000020explicit=000020expression}[{\cite[(18.5.2)]{DigLibMathFunc}}]
For all $n\in\N$ and all $\theta\in\R$ we have 
\[
S_{n}(2\cos\theta)=U_{n}(\cos\theta)=\frac{\sin(n+1)\theta}{\sin\theta}.
\]
\end{lem}

\bibliographystyle{alpha}
\bibliography{references_BBL}

\begin{thebibliography}{DEGT08}

\bibitem[BBG92]{BelBovGhe92}
J.~Bellissard, A.~Bovier, and J.-M. Ghez.
\newblock Gap labelling theorems for one-dimensional discrete {S}chr\"odinger
  operators.
\newblock {\em Rev. Math. Phys.}, 4(1):1--37, 1992.

\bibitem[BBL23]{BaBeLo23-MFO}
R.~Band, S.~Beckus, and R.~Loewy.
\newblock {W}orkshop: {A}spects of {A}periodic {O}rder.
\newblock {\em Oberwolfach Rep.}, 2023.
\newblock extended version in arXiv:2309.04351.

\bibitem[BBL24]{BanBecLoe_24}
R.~Band, S.~Beckus, and R.~Loewy.
\newblock {The Dry Ten Martini Problem for Sturmian Hamiltonians}.
\newblock {\em arXiv:2402.16703}, 2024.

\bibitem[Bel92]{Bell92-Gap}
J.~Bellissard.
\newblock Gap {L}abelling {T}heorems for {S}chr\"odinger operators.
\newblock In J.~M.~Luck M.~Waldschmidt, P.~Moussa and C.~Itzykson, editors,
  {\em From number theory to physics ({L}es {H}ouches, 1989)}, pages 538--630.
  Springer, Berlin, 1992.

\bibitem[Ber07]{Berstel_Sturmian_survey}
J.~Berstel.
\newblock Sturmian and episturmian words (a survey of some recent results).
\newblock In {\em Algebraic informatics}, volume 4728 of {\em Lecture Notes in
  Comput. Sci.}, pages 23--47. Springer, Berlin, 2007.

\bibitem[BG13]{BaaGri_book1}
M.~Baake and U.~Grimm.
\newblock {\em Aperiodic {O}rder. {V}ol. 1: A Mathematical Invitation}.
\newblock Cambridge University Press, Cambridge, 2013.

\bibitem[BG17]{BaaGri_book2}
M.~Baake and U.~Grimm, editors.
\newblock {\em Aperiodic {O}rder. {V}ol. 2: Crystallography and Almost
  Periodicity}.
\newblock Cambridge University Press, Cambridge, 2017.

\bibitem[BIST89]{BIST89}
J.~Bellissard, B.~Iochum, E.~Scoppola, and D.~Testard.
\newblock Spectral properties of one-dimensional quasi-crystals.
\newblock {\em Comm. Math. Phys.}, 125(3):527--543, 1989.

\bibitem[BIT91]{BIT91}
J.~Bellissard, B.~Iochum, and D.~Testard.
\newblock Continuity properties of the electronic spectrum of {$1$}{D}
  quasicrystals.
\newblock {\em Comm. Math. Phys.}, 141(2):353--380, 1991.

\bibitem[Cas86]{Casdagli1986}
M.~Casdagli.
\newblock Symbolic dynamics for the renormalization map of a quasiperiodic
  {S}chr\"{o}dinger equation.
\newblock {\em Comm. Math. Phys.}, 107(2):295--318, 1986.

\bibitem[CQ23]{CaoQu_arXiv23}
J.~Cao and Y.~Qu.
\newblock Almost sure dimensional properties for the spectrum and the density
  of states of {S}turmian {H}amiltonians.
\newblock {\em arXiv:2310.07305}, 2023.

\bibitem[Dam07]{Dam07_survey}
D.~Damanik.
\newblock Strictly ergodic subshifts and associated operators.
\newblock In F.~Gesztesy, P.~Deift, C.~Galvez, P.~Perry, and W.~Schlag,
  editors, {\em Spectral theory and mathematical physics: a {F}estschrift in
  honor of {B}arry {S}imon's 60th birthday}, volume~76 of {\em Proc. Sympos.
  Pure Math.}, pages 505--538. Amer. Math. Soc., Providence, RI, 2007.

\bibitem[Dam17]{Dam17-Survey}
D.~Damanik.
\newblock Schr\"odinger operators with dynamically defined potentials.
\newblock {\em Ergodic Theory Dynam. Systems}, 37(6):1681--1764, 2017.

\bibitem[DEG15]{DaEmGo15-survey}
D.~Damanik, M.~Embree, and A.~Gorodetski.
\newblock Spectral properties of {S}chr{\"o}\-dinger operators arising in the
  study of quasicrystals.
\newblock In J.~Kellendonk, D.~Lenz, and J.~Savinien, editors, {\em Mathematics
  of aperiodic order}, volume 309, pages 307--370. Birkh{\"a}user Springer,
  Basel, 2015.

\bibitem[DEGT08]{Damanik2008}
D.~Damanik, M.~Embree, A.~Gorodetski, and S.~Tcheremchantsev.
\newblock The fractal dimension of the spectrum of the {F}ibonacci
  {H}amiltonian.
\newblock {\em Comm. Math. Phys.}, 280(2):499--516, 2008.

\bibitem[DF22]{DaFi22-book_1}
D.~Damanik and J.~Fillman.
\newblock {\em One-dimensional ergodic {S}chr\"{o}dinger operators---{I}.
  {G}eneral theory}, volume 221.
\newblock American Mathematical Society, Providence, RI, 2022.

\bibitem[DF23]{DamFill23-GapLabel}
D.~Damanik and J.~Fillman.
\newblock Gap labelling for discrete one-dimensional ergodic {S}chr\"{o}dinger
  operators.
\newblock In M.~Brown, F.~Gesztesy, P.~Kurasov, A.~Laptev, B.~Simon, G.~Stolz,
  and I.~Wood, editors, {\em From complex analysis to operator theory---a
  panorama}, volume 291, pages 341--404. Birkh\"{a}user/Springer, Cham, 2023.

\bibitem[DG11]{DamanikGorodetski2011}
D.~Damanik and A.~Gorodetski.
\newblock Spectral and quantum dynamical properties of the weakly coupled
  {F}ibonacci {H}amiltonian.
\newblock {\em Comm. Math. Phys.}, 305(1):221--277, 2011.

\bibitem[DG15]{Damanik2015}
D.~Damanik and A.~Gorodetski.
\newblock Almost sure frequency independence of the dimension of the spectrum
  of {S}turmian {H}amiltonians.
\newblock {\em Comm. Math. Phys.}, 337(3):1241--1253, 2015.

\bibitem[DGY16]{DaGoYe16}
D.~Damanik, A.~Gorodetski, and W.~Yessen.
\newblock The {F}ibonacci {H}amiltonian.
\newblock {\em Invent. Math.}, 206(3):629--692, 2016.

\bibitem[DKL00]{DamKilLenz00}
D.~Damanik, R.~Killip, and D.~Lenz.
\newblock Uniform spectral properties of one-dimensional quasicrystals. {III}.
  {$\alpha$}-continuity.
\newblock {\em Comm. Math. Phys.}, 212(1):191--204, 2000.

\bibitem[DL99a]{DamLen99-AbsEig}
D.~Damanik and D.~Lenz.
\newblock Uniform spectral properties of one-dimensional quasicrystals. {I}.
  {A}bsence of eigenvalues.
\newblock {\em Comm. Math. Phys.}, 207(3):687--696, 1999.

\bibitem[DL99b]{DamLen99_II}
D.~Damanik and D.~Lenz.
\newblock Uniform spectral properties of one-dimensional quasicrystals. {II}.
  {T}he {L}yapunov exponent.
\newblock {\em Lett. Math. Phys.}, 50(4):245--257, 1999.

\bibitem[{\relax DLMF}]{DigLibMathFunc}
{\it NIST Digital Library of Mathematical Functions}.
\newblock \url{https://dlmf.nist.gov/}, Release 1.1.12 of 2023-12-15.
\newblock F.~W.~J. Olver, A.~B. {Olde Daalhuis}, D.~W. Lozier, B.~I. Schneider,
  R.~F. Boisvert, C.~W. Clark, B.~R. Miller, B.~V. Saunders, H.~S. Cohl, and
  M.~A. McClain, eds.

\bibitem[DT07]{Damanik_Tcheremchantsev07}
D.~Damanik and S.~Tcheremchantsev.
\newblock Upper bounds in quantum dynamics.
\newblock {\em J. Amer. Math. Soc.}, 20(3):799--827, 2007.

\bibitem[Hof93]{Hof93}
A.~Hof.
\newblock Some remarks on discrete aperiodic {S}chr\"odinger operators.
\newblock {\em J. Statist. Phys.}, 72(5-6):1353--1374, 1993.

\bibitem[Jit07]{Jit07}
S.~Jitomirskaya.
\newblock Ergodic {S}chr\"odinger operators (on one foot).
\newblock In F.~Gesztesy, P.~Deift, C.~Galvez, P.~Perry, and W.~Schlag,
  editors, {\em Spectral theory and mathematical physics: a {F}estschrift in
  honor of {B}arry {S}imon's 60th birthday}, volume 76, Part 2 of {\em Proc.
  Sympos. Pure Math.}, pages 613--647. Amer. Math. Soc., Providence, RI, 2007.

\bibitem[Khi64]{Khinchin_book64}
A.~Ya. Khinchin.
\newblock {\em Continued fractions}.
\newblock University of Chicago Press, Chicago, Ill.-London, 1964.

\bibitem[KKL03]{KiKiLa03}
R.~Killip, A.~Kiselev, and Y.~Last.
\newblock Dynamical upper bounds on wavepacket spreading.
\newblock {\em Amer. J. Math.}, 125(5):1165--1198, 2003.

\bibitem[KKT83]{KohKadTan_prl83}
M.~Kohmoto, L.~P. Kadanoff, and C.~Tang.
\newblock Localization problem in one dimension: Mapping and escape.
\newblock {\em Phys. Rev. Lett.}, 50:1870--1872, Jun 1983.

\bibitem[Len02]{Len02}
D.~Lenz.
\newblock Singular spectrum of {L}ebesgue measure zero for one-dimensional
  quasicrystals.
\newblock {\em Comm. Math. Phys.}, 227(1):119--130, 2002.

\bibitem[Lot02]{Lothaire2002}
M.~Lothaire.
\newblock {\em Algebraic combinatorics on words}, volume~90.
\newblock Cambridge University Press, Cambridge, 2002.

\bibitem[LP86]{LuckPetritis86}
J.~M. Luck and D.~Petritis.
\newblock Phonon spectra in one-dimensional quasicrystals.
\newblock {\em Journal of Statistical Physics}, 42(3):289--310, 1986.

\bibitem[LQW14]{Liu2014}
Q.-H. Liu, Y.-H. Qu, and Z.-Y. Wen.
\newblock The fractal dimensions of the spectrum of {S}turm {H}amiltonian.
\newblock {\em Adv. Math.}, 257:285--336, 2014.

\bibitem[LW04]{LiuWen04}
Q.-H. Liu and Z.-Y. Wen.
\newblock Hausdorff dimension of spectrum of one-dimensional {S}chr\"odinger
  operator with {S}turmian potentials.
\newblock {\em Potential Anal.}, 20(1):33--59, 2004.

\bibitem[Mei14]{Mei14}
M.~Mei.
\newblock Spectra of discrete {S}chr\"{o}dinger operators with primitive
  invertible substitution potentials.
\newblock {\em J. Math. Phys.}, 55(8):082701, 22, 2014.

\bibitem[OK85]{OstKim_phys85}
S.~Ostlund and S.-H. Kim.
\newblock Renormalization of quasiperiodic mappings.
\newblock {\em Phys. Scripta}, T9:193--198, 1985.

\bibitem[Ray]{Raym-AperiodicOrder}
L.~Raymond.
\newblock Constructive gap labelling for one-dimensional {S}chr\"{o}dinger
  operators.
\newblock to appear in: Aperiodic Order, Vol. 4: Schr\"odinger Operators, eds.
  M. Baake, D. Damanik, and N. Ma\~{n}ibo, Cambridge University Press, in
  preparation.

\bibitem[Ray95a]{Raym95}
L.~Raymond.
\newblock A constructive gap labelling for the discrete {S}chr\"odinger
  operator on a quasiperiodic chain, 1995.
\newblock preprint.

\bibitem[Ray95b]{Raym95-thesis}
L.~Raymond.
\newblock {\em {E}tude alg\'{e}brique de milieux quasip\'{e}riodiques}.
\newblock PhD thesis, Aix-Marseille {I}, 1995.

\bibitem[Ray11]{MFO2011_Raym}
L.~Raymond.
\newblock Scaling properties of {S}turmian potential based {S}chr\"{o}dinger
  operator: some useful tools.
\newblock {\em Oberwolfach {R}ep.}, 8(1):142--167, 2011.

\bibitem[Sim82]{Sim82-review}
B.~Simon.
\newblock Almost periodic {S}chr\"odinger operators: a review.
\newblock {\em Adv. in Appl. Math.}, 3(4):463--490, 1982.

\bibitem[Sim11]{Simon2011}
B.~Simon.
\newblock {\em Szego's Theorem and Its Descendants: Spectral Theory for
  L2Perturbations of Orthogonal Polynomials}.
\newblock Princeton University Press, 2011.

\bibitem[S{\"u}t87]{Sut87}
A.~S{\"u}t{\H{o}}.
\newblock The spectrum of a quasiperiodic {S}chr\"odinger operator.
\newblock {\em Comm. Math. Phys.}, 111(3):409--415, 1987.

\bibitem[S{\"u}t89]{Sut89}
A.~S{\"u}t{\H{o}}.
\newblock Singular continuous spectrum on a {C}antor set of zero {L}ebesgue
  measure for the {F}ibonacci {H}amiltonian.
\newblock {\em J. Statist. Phys.}, 56(3-4):525--531, 1989.

\bibitem[Tes00]{Tes00}
G.~Teschl.
\newblock {\em Jacobi operators and completely integrable nonlinear lattices},
  volume~72.
\newblock American Mathematical Society, Providence, RI, 2000.

\end{thebibliography}

\end{document}